\documentclass[sigconf, nonacm, pdfa]{acmart}%, nonacm

\newcommand\vldbdoi{10.14778/3648160.3648187}
\newcommand\vldbpages{1515 - 1527}
\newcommand\vldbvolume{17}
\newcommand\vldbissue{6}
\newcommand\vldbyear{2024}
\newcommand\vldbauthors{\authors}
\newcommand\vldbtitle{\shorttitle} 
\newcommand\vldbavailabilityurl{https://github.com/hououou/AeonG.git}
\newcommand\vldbpagestyle{empty} 
\usepackage[justification=centering]{caption}

\usepackage{graphicx}
\usepackage{xcolor}
\usepackage{url}
\usepackage{enumitem} 
\usepackage{stfloats}

\usepackage[normalem]{ulem}
\usepackage{amsmath}
\usepackage{subfigure}
\usepackage{listings}
\usepackage{multirow}
\usepackage{bm}
\usepackage[linesnumbered,ruled,vlined]{algorithm2e}
\usepackage{pifont}
\usepackage{tikz}
\usepackage{listings}
\usepackage{multicol}

\SetKwInOut{Input}{input}\SetKwInOut{Output}{output}

{}
{}

\theoremstyle{definition}
\newtheorem{definition}{Definition}

\theoremstyle{plain}

\newtheorem{theorem}{Theorem}

\theoremstyle{remark}
\theoremstyle{definition}

\newcommand{\numberedcircle}[1]{%
  \tikz[baseline=(char.base)]{
    \node[shape=circle, draw, fill=black, text=white,inner sep=0.8pt,font=\small] (char) {#1};
  }
}
\newcommand\tgdb{\ensuremath{\textsf{AeonG}}}
\newcommand\tgdbd{\ensuremath{\textsf{AeonG-D}}}
\newcommand\visible{\ensuremath{\text{legal}}}
\newcommand\expgraph{\ensuremath{\text{1}}}
\newcommand\expmodel{\ensuremath{\text{2}}}
\newcommand\expquery{\ensuremath{\text{3}}}
\newcommand\expoverviewstorage{\ensuremath{\text{4}}}
\newcommand\expoverviewquery{\ensuremath{\text{5}}}
\newcommand\expocurrent{\ensuremath{\text{6}}}

\newcommand{\todo}[1]{{\color{black} {#1}}}

\newif\ifextended\extendedtrue

\newcommand{\maintext}[1]{\ifextended\relax\else#1\fi} 

\newcommand{\extended}[1]{\ifextended#1\else\relax\fi}

\begin{document}

\maintext{\title{AeonG: An Efficient Built-in Temporal Support in Graph Databases}}

%%%%% Extended
\extended{\title{AeonG: An Efficient Built-in Temporal Support in Graph Databases (Extended Version)}}
    
%%
%% The "author" command and its associated commands are used to define the authors and their affiliations.
\author{Jiamin Hou}
\affiliation{%
  \institution{Renmin University of China}
}
\email{jiaminhou@ruc.edu.cn}

\author{Zhanhao Zhao}
\affiliation{%
  \institution{Renmin University of China}
}
\email{zhanhaozhao@ruc.edu.cn}

\author{Zhouyu Wang}
\affiliation{%
  \institution{Renmin University of China}
}
\email{zyu\_wang@ruc.edu.cn}

% \begin{multicols}{4}
\author{Wei Lu}
\affiliation{\institution{Renmin University of China}}
\email{lu-wei@ruc.edu.cn}

% \columnbreak % 切换到下一列

\author{Guodong Jin}
\affiliation{\institution{University of Waterloo}}
\email{g35jin@uwaterloo.ca}

% \columnbreak % 切换到下一列

\author{Dong Wen}
\affiliation{\institution{UNSW, Australia}}
\email{dong.wen@unsw.edu.au}

% \columnbreak % 切换到下一列

\author{Xiaoyong Du}
\affiliation{\institution{Renmin University of China}}
\email{duyong@ruc.edu.cn}
% \end{multicols}

\begin{abstract}
Real-world graphs are often dynamic and evolve over time. 
It is crucial for storing and querying a graph's evolution %by introducing temporal features 
in graph databases.
However, existing works either suffer from high storage overhead or lack efficient temporal query support, or both.
In this paper,
we propose {\tgdb}, a new graph database with built-in temporal support.
% that efficiently provides built-in temporal support.
{\tgdb} is based on a novel temporal graph model.
To fit this model, we design a storage engine and a query engine. 
Our storage engine is hybrid, with one current storage to manage the most recent versions of graph objects, and another historical storage to manage the previous versions of graph objects.
This separation makes the performance degradation of querying the most recent graph object versions as slight as possible. 
To reduce the historical storage overhead, we propose a novel \textit{anchor+delta} strategy, in which we periodically create a complete version (namely anchor) of a graph object, and maintain every change (namely delta) between two adjacent anchors of the same object.  
To boost temporal query processing, we propose an anchor-based version retrieval technique in the query engine to skip unnecessary historical version traversals.
Extensive experiments are conducted on both real 
and synthetic datasets. The results show 
that {\tgdb} achieves up to {5.73$\times$} lower storage consumption
% {{2.82$\times$} lower latency for graph operations, 
and {2.57$\times$} lower temporal query latency against state-of-the-art approaches, while introducing only {{9.74\%}} performance degradation for supporting temporal features.

\end{abstract}

\maketitle

\setcounter{figure}{0}
\setcounter{table}{0}
\setcounter{lstlisting}{0}
\setcounter{page}{1}
\pagestyle{\vldbpagestyle}
\begingroup\small\noindent\raggedright\textbf{PVLDB Reference Format:}\\
\vldbauthors. \vldbtitle. PVLDB, \vldbvolume(\vldbissue): \vldbpages, \vldbyear.\\
\href{https://doi.org/\vldbdoi}{doi:\vldbdoi}
\endgroup
\begingroup
\renewcommand\thefootnote{}\footnote{
\noindent Wei Lu is the corresponding author.

\noindent This work is licensed under the Creative Commons BY-NC-ND 4.0 International License. Visit \url{https://creativecommons.org/licenses/by-nc-nd/4.0/} to view a copy of this license. For any use beyond those covered by this license, obtain permission by emailing \href{mailto:info@vldb.org}{info@vldb.org}. Copyright is held by the owner/author(s). Publication rights licensed to the VLDB Endowment. \\
\raggedright Proceedings of the VLDB Endowment, Vol. \vldbvolume, No. \vldbissue\ %
ISSN 2150-8097. \\
\href{https://doi.org/\vldbdoi}{doi:\vldbdoi} \\
}\addtocounter{footnote}{-1}\endgroup

\ifdefempty{\vldbavailabilityurl}{}{
\vspace{.3cm}
\begingroup\small\noindent\raggedright\textbf{PVLDB Artifact Availability:}\\
The source code, data, and/or other artifacts have been made available at \url{\vldbavailabilityurl}.
\endgroup
}

\section{Introduction}
\label{sec:intro}
Graphs are prevalent to model relationships between real-world entities.
Many graph databases,  such as Neo4j~\cite{neo4j}, ArangoDB~\cite{ArangoDB}, Dgraph~\cite{Dgraph}, and Memgraph~\cite{Memgraph}, 
% have been 
are developed to manage graph data efficiently.
Despite the fact that real-world graphs are often dynamic and evolve over time, these 
% graph 
databases are typically designed to manage
% model and analyze 
{up-to-date} graph data:
when a graph changes, the database only stores
% focus on \textit{static} graphs, indicating 
the current (latest) state of the graph,
% which are only capable of maintaining the latest state of a graph. 
i.e., 
% the database  will only retains 
the most recent values
% state 
of vertices and edges, while discarding any previous (historical) state.
However, 
time-evolving (temporal) graph data, which contains both the latest and historical states of a graph, is important in many applications, such as financial fraud detection~\cite{abdallah2016fraud}, traffic prediction in road networks~\cite{laddada2020graph}, etc.

\begin{figure}[]
\centering
\includegraphics[width=0.45\textwidth]{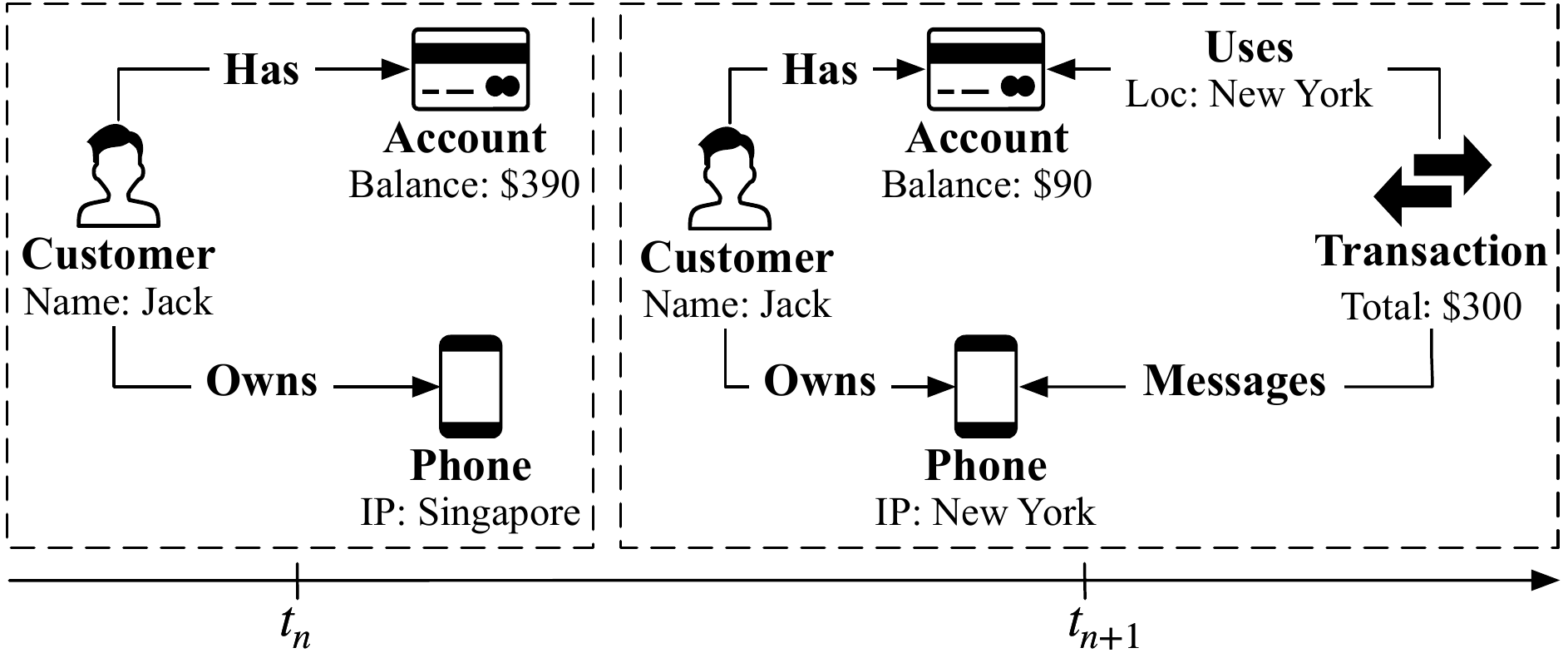}
% \vspace{-2mm}
\caption{The Evolution of a Customer Purchase Graph}
\label{fig:exp1}
% \vspace{-2mm}
\end{figure}

\textbf{Example {\expgraph}.}
Figure~\ref{fig:exp1} shows the evolution of a customer purchase graph, where 
customers, bank accounts, phones, and transactions are modeled as entities, and the relationships between entities are modeled as edges.
The phone logs its location during various activities, such as receiving a message or web browsing.
Each customer purchase invokes a transaction to update the account balance and record the location where it occurs.
Let us assume that at time $t_n$, a customer named Jack has an account balance of \$390, and his phone's location falls in Singapore.
% makes a phone call in Singapore. 
We store a graph reflecting this state, as shown 
% This activity constructs a graph, depicted 
in the left portion of Figure~\ref{fig:exp1}.
%\sout{, recording the call location as Singapore}.
Subsequently, at time $t_{n+1}$ (one minute after $t_n$), Jack invokes a purchase transaction totaling \$300, resulting in a new graph state, as shown in the right part of Figure~\ref{fig:exp1}.
This transaction occurs in New York, identical to the location of Jack's phone, thus it appears to be legitimate.
However, when comparing the states of $t_{n+1}$ and $t_{n}$, we observe that Jack's phone location changes from Singapore to New York within one minute.
Considering it is impossible for Jack to travel such a distance so quickly, this transaction is likely fraudulent.
We would like to emphasize that changes in phone location alone are not inherently suspicious. 
However, when such a location shift is associated with a transaction, it warrants vigilance.
As discussed above, this potentially fraudulent activity can only be identified by tracking the evolution of the graph structure over time.
Therefore, traditional graph databases, which only maintain the latest state at $t_{n+1}$, would fail to detect such fraudulent transactions. \qed
% \end{example}

Thus far, various works have been proposed to manage temporal graph data.
Several proposals~\cite{T-GQL,Frame,liu2017keyword, GRADOOP, Timebased3, Timebased4, durand2017backlogs} 
assign each vertex or edge in the graph with a timestamp property to reflect its lifespan, as shown in Figure~\ref{fig:exp2}.
Rather than discarding the previous state when the graph changes, these approaches maintain both the current and historical states in a single graph.
For example, at time $t_{n+1}$, two vertices of Jack's phone coexist -- one represents the previous state with a time interval of $[t_{n}, t_{n+1})$, and the other denotes the current state with $[t_{n+1}, +\infty)$.
Consequently, they can detect fraudulent transactions, as in Example {\expgraph}, by identifying the irregular sub-graph structure (highlighted in the red box) that indicates a transaction proceeded with location inconsistency.
However, in these approaches where timestamps are treated as regular properties, executing temporal queries (which select data based on given timestamps) often requires the traversal of the entire graph.
As the graph size inevitably increases due to the addition of historical states, query efficiency can significantly degrade over time.
Another line of research~\cite{DeltaGraph, ChronoGraph, Raphtory, llama, ImmortalGraph,clock-g, graphone, Chronos, Auxo, khurana2016storing}
% \todo{introduce temporal features at the system level,} 
manages time-evolving graph data by periodically materializing the snapshots of the entire graph and logging the deltas between two successive snapshots.
% the temporal data as of a specified point in time 
Querying a historical state in these methods requires first identifying the nearest snapshot based on the timestamp, and then reconstructing the complete graph state using the snapshot and associated deltas.
% , and finally executing the temporal query on this state. 
Therefore, these methods incur substantial storage overhead due to the maintenance of complete snapshots, and can lead to sub-optimal query performance because of the historical state reconstruction.

\begin{figure}[]
\centering
\includegraphics[width=0.45\textwidth]{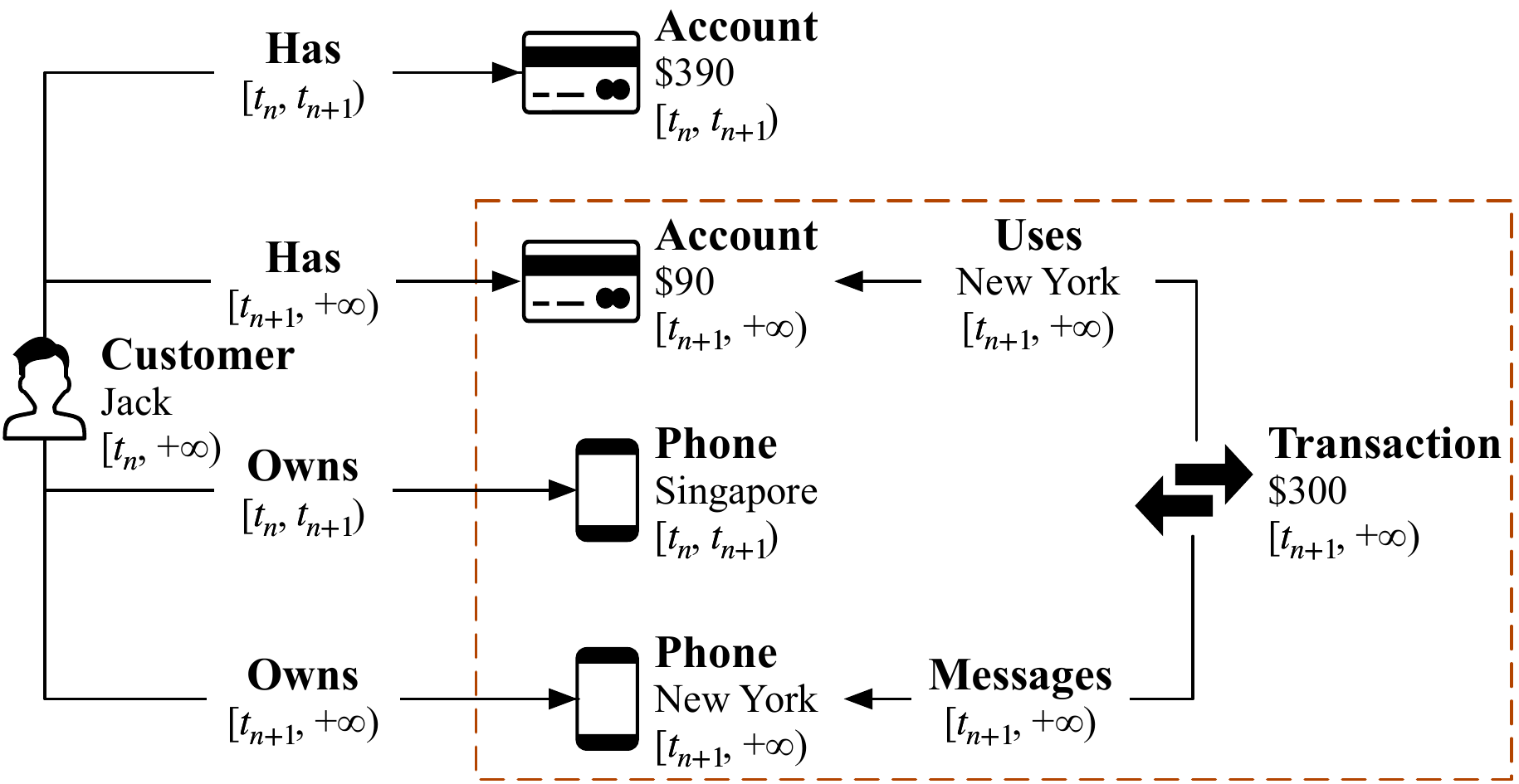}
% \vspace{-2mm}
\caption{Customer Purchase Graph with a Time Dimension}
\label{fig:exp2}
% \vspace{-3mm}
% \vspace{-2mm}
\end{figure}

% \textbf{Challenges.} 
Our goal is to design a graph database for efficient temporal graph data management. 
However, achieving this requires addressing three key challenges to minimize storage overhead and facilitate swift temporal query processing.
First, as the volume of historical graph states continuously escalates, achieving minimal storage overhead is not straightforward\ding{182}. 
Second, given the considerable amount of historical states, it is not trivial to 
% the task of 
process temporal queries efficiently while upholding data 
% freshness and 
consistency 
% is not trivial
\ding{183}.
Lastly, the database needs native temporal support to enable users to conveniently access temporal graph data\ding{184}.

% In this paper, we propose a built-in temporal support implementation in graph database systems which has desirable consume space of Log and also provide near direct access like Copy. 
In this paper, we propose {\tgdb}, a new graph database that efficiently offers built-in temporal support.
By integrating the widely-accepted static property graph model~\cite{property_graph, cypher,property_graph_model,pg-key} with a time dimension, we first define a \textit{temporal property graph model} to formalize the representation and manipulation of temporal graph data.
\todo{Guided by this model, we then extend the common graph database architecture to design {\tgdb}.
In particular, we enhance the storage engine, query language, and query engine, with efficient temporal support.}
% Building upon existing foundations, {\tgdb} further expands these components to introduce temporal support.
\todo{We build a \textit{hybrid storage engine}\ding{182}, constituting the current storage and historical storage, to store temporal graph data with minimal storage overhead.
This engine maintains multiple versions for each vertex and edge, with the most recent versions retained in the current storage and previous versions in the historical storage. 
We integrate time dimensions into the data layout, and develop the current storage based on the multi-version storage engine used in various existing graph databases~\cite{livegraph,Weaver,G-Tran,Memgraph,Dgraph,ArangoDB}.
We propose a novel ``anchor+delta'' strategy to compactly organize historical data in the historical storage.
% , which is built based on key-value store~\cite{rocksdb,tikv}.
% based on key-value store and equip it with, compactly storing  versions.
% , asynchronously migrated from the current storage. 
In particular, we periodically create a complete version (namely anchor) of a graph object and maintain every change (namely delta) between two consecutive anchors of the same object to reduce the historical storage overhead. 
Moreover, we introduce an asynchronous migration mechanism to transfer outdated versions from the current storage to the historical storage. 
Instead of synchronously migrating previous versions with every update or deletion of a vertex/edge, we defer the migration until the database’s periodic garbage collection is invoked~\cite{mvcc_overall}.
This mechanism ensures that the migration is non-intrusive, thereby reducing the performance degradation caused by the maintenance of temporal data.

We then present a \textit{temporal-enhanced query language}\ding{184}, which extends Cypher~\cite{cypher}, a common-used graph query language, to conveniently access temporal graph data. 
% by integrating temporal dimensions,
% Furthermore, we fundamentally redesign the storage and query engines to align with the new temporal data model and query language. 
Building upon the hybrid storage engine, we introduce a built-in \textit{temporal query engine}\ding{183}.
We inherit two fundamental operations from existing graph databases, namely scan and expand, and extend them to enable consistent and efficient temporal query processing.
We propose a unique anchor-based version retrieval technique to {minimize} unnecessary historical version traversals in the scan and expand operators.
Specifically, we directly locate the nearest anchor that aligns with the given query conditions, and apply the subsequent deltas on the obtained anchors to reconstruct the desired version, thus minimizing the historical version traversal overhead. 
}

In summary, we make the following contributions:
% \vspace{-3mm}
\begin{itemize}[leftmargin=*,itemsep=2pt,topsep=0pt,parsep=0pt]
% [leftmargin=*]
% ,itemsep=2pt,topsep=0pt,parsep=0pt]
 \todo{\item We present {\tgdb}, a new graph database providing efficient built-in temporal support. Built with a temporal-enhanced query language, query engine, and storage engine, {\tgdb} regards temporal features as the first citizen, making it simple and intuitive to manipulate temporal graph data.}
 
 \todo{\item We propose a hybrid storage engine, which employs separate storage engines with an ``anchor+delta'' strategy to reduce storage overhead for historical data.
 We further introduce an asynchronous migration strategy to minimize performance degradation for maintaining temporal graph data.
 }
 
 \todo{\item We design a temporal query engine, featuring an anchor-based version retrieval technique, to provide consistent and efficient temporal query processing with minimal historical version traversal overhead.}

\item We implement {\tgdb} based on Memgraph~\cite{Memgraph}, a real-world native graph database. 
We conduct extensive experiments on both real and synthetic datasets,
% \cz{one real dataset and two synthetic datasets },
% both real and synthetic datasets, 
and compare {\tgdb} against two state-of-the-art temporal graph databases~\cite{T-GQL,clock-g}.
% , T-GQL~\cite{T-GQL} and Clock-G~\cite{clock-g}.
The results demonstrate that {\tgdb} achieves up to {5.73$\times$} lower storage consumption 
% {{2.82$\times$} lower latency for graph operations, 
and {2.57$\times$} lower latency for temporal queries, while only introducing {9.74\%} performance degradation for supporting temporal features.
% compared with the vanilla Memgraph .
\end{itemize}

\extended{
The remainder of the paper is structured as follows.
The next section 
% provides relevant background on temporal graph databases, 
formulates the temporal graph model, and presents the temporal query language.
% We next formalize our temporal graph data model, query language and constraints, and 
Section~\ref{sec:overview} overviews the architecture of {\tgdb}.
Section~\ref{sec:hybrid_storage} details the hybrid storage engine, and Section~\ref{sec:query-engine} elaborates on the temporal query engine.
Section~\ref{sec:implementation} describes the system implementation.
% , and 
Section~\ref{sec:evaluation} presents the experimental results.
Section~\ref{sec:relatedwork} discusses the related work, and Section~\ref{sec:conclusion} concludes.
}

\section{Modeling and query language}
\label{sec:background}

In this section, we formulate the temporal graph model and present the temporal query language used in {\tgdb}. 

\subsection{Temporal Property Graph Model}
\label{sec:graph_model}

We define the temporal graph model by extending the static property graph model \todo{\cite{property_graph, cypher,property_graph_model,pg-key}} with a time dimension.
In the property graph model, real-world entities are represented as vertices, and the relationships between these entities are modeled as edges. Each vertex or edge has a unique identifier (id for short), possibly several labels (e.g., customer, phone), and properties (e.g., Name: Jack).

\todo{
\begin{definition}[Property Graph]
\label{def:pgm}
% {\textbf (Property Graph Model)} 
Let $\mathcal{N}$ and $\mathcal{E}$ denote sets of vertex ids and edge ids, respectively. Assume countable sets $\mathcal{L}$, $\mathcal{K}$, and $\mathcal{V}$ of \textit{labels, property names}, and \textit{property values}. A property graph is a tuple $G=\langle N, E, \rho, \lambda, \pi \rangle$ where:
\begin{itemize}[leftmargin=*,itemsep=2pt,topsep=0pt,parsep=0pt]
\item $N$ is a finite subset of $\mathcal{N}$, whose elements are referred to as the \textit{vertices} of $G$;
\item $E$ is a finite subset of $\mathcal{E}$, whose elements are referred to as the \textit{edges} of $G$ and $N \cap E =\emptyset$;
\item $\rho$: $E$ $\to$ $(N \times N)$ is a total function mapping each edge to its source and destination vertices;
 \item $\lambda: (N \cup E) \rightarrow 2^{\mathcal{L}}$ is a total function mapping vertices and edges to finite sets of labels (including the empty set); 
\item $\pi:(N \cup E) \times \mathcal{K} \rightarrow \mathcal{V}$ is a finite partial function, mapping a vertex/edge and a property key to a value.
\end{itemize}
\end{definition}
}

\begin{figure}[]
\centering    %居中
\includegraphics[width=0.45\textwidth]{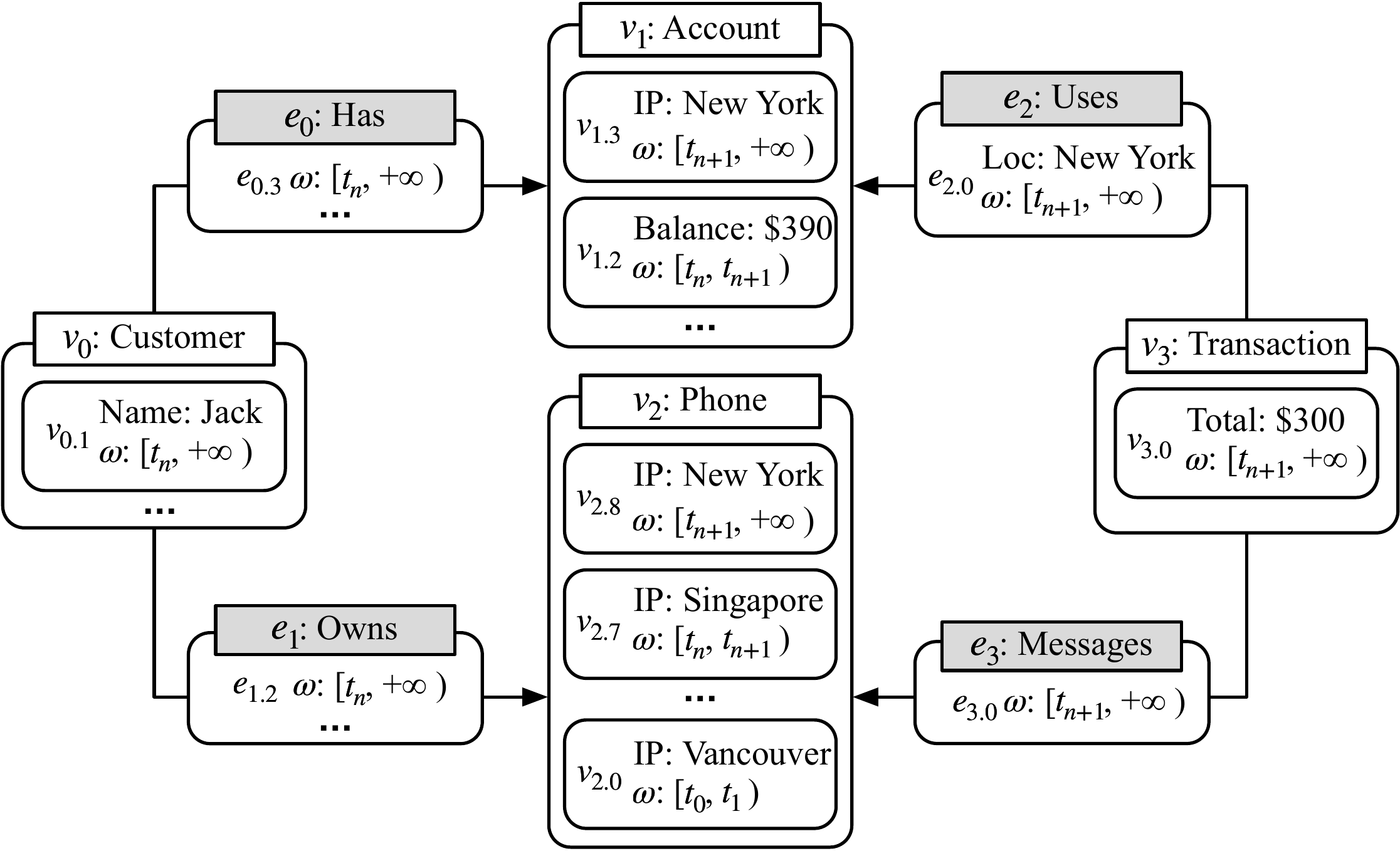}
\centering
% \vspace{-2mm}
\caption{A Running Example of Temporal Property Graphs}
% \vspace{-2mm}
\label{fig:exp_model}
\end{figure}

The property graph model, originally designed for static graphs, lacks the inherent ability to capture the evolution of graphs over time. 
In the context of relational databases, the concept of ``Transaction Time''~\cite{sql-2011} is proposed to bring a time dimension to the relational model.
This transaction time is created and maintained by the database system itself, tracking the lifespan of each data item within the system.
Inspired by the transaction time, we integrate the time dimension into the property graph model to formally define the temporal property graph model.

\todo{
\begin{definition}[Temporal Property Graph]
\label{def:tpgm}
A temporal property graph is a tuple $G=\langle \Omega, N, E, \rho,\lambda, \pi, \sigma, \tau \rangle$ where:
\begin{itemize}[leftmargin=*,itemsep=2pt,topsep=0pt,parsep=0pt]
\item $\Omega$ is a temporal domain, which is a finite set of consecutive timestamps, that is, $\Omega=\{i\in T | a \leq i\leq b\}$ for some $a$, $b \in T$ such that $a \leq b$.
$T$ represents the universe of time points;
\item $N$, $E$, $\rho$, $\lambda$, $\pi$ inherit  their definitions from Definition \ref{def:pgm};
\item $\sigma: (N \cup E) \times \Omega \to \{true, false\} $ is a total function that maps a vertex or an edge and a time period $\omega$ to a Boolean variable, indicating whether this vertex or edge exists during period $\omega$;
\item $\tau:(N \cup E) \times \mathcal{K} \times \Omega \rightarrow \mathcal{V}$ is a partial function that maps a vertex or an edge and a property key, and a time period $\omega$ to a value.
\end{itemize}
\end{definition}
}

\todo{
\textbf{{Constraints.}} 
In our temporal property graph model, we impose two constraints to enforce that the graph at any time point corresponds to a valid property graph. At any time point $t$:  (1) An edge exists only if both source and destination vertices exist at $t$. Formally, if $e \in E$, $\sigma(e,t)=true$, and $\rho(e)=(v_1,v_2)$, then $\sigma(v_1,t)=true$ and $\sigma(v_2,t)=true$; (2) A property can only take on a value during the time period when the corresponding vertex or edge exists. Formally, if $\tau(o,k,t)=val$, where $ o\in(N \cup E)$, $k\in\mathcal{K}$, $val\in\mathcal{V}$, then $\sigma(o,t)=true$. 
}

\renewcommand\thelstlisting{\arabic{lstlisting}}
\maintext{
\begin{lstlisting}[mathescape,caption={Syntax of Temporal-enhanced Cypher},label={lst:cypher}, abovecaptionskip=0pt]
[OPTIONAL] MATCH pattern_tuple 
           [WHERE expr] 
           [FOR TT AS OF expr| FOR TT FROM expr TO expr]
\end{lstlisting}
}
%%%%% Extended
\extended{
\begin{figure*}[] % 开始双栏
\begin{lstlisting}[mathescape,caption={Syntax of Expressions, Queries, and Clauses in Cypher Enhanced with Temporal Features},label={lst:cypher}, abovecaptionskip=0pt,escapeinside=&&]
  //EXPRESSIONS
  expr::=v|a| $f$@(@expr_list@)@  $v \in \mathcal{V}, a \in \mathcal{A}, f \in \mathcal{F}$ &{\hfill values/variables }&  
         |expr.$k|$ @{}@ | @{@prop_list@}@ &{\hfill maps    }&
         |@[]@ | @[@expr_list@]@ | expr IN expr | expr@[@expr@]@ | expr@[@expr..@]@ | expr@[..@expr@]@ | expr@[@expr@..@expr@]@ &{\hfill lists    }& 
         |expr START WITH expr | expr ENDS WITH expr | expr CONTAINS expr &{\hfill strings    }&
         |expr OR expr| expr AND expr| expr XOR expr| NOT expr | expr IS NULL | expr IS NOT NULL &{\hfill logic }&
         |expr @<@ expr| expr @<=@ expr| expr @>=@ expr| expr @>@ expr | expr @=@ expr | expr @<>@ expr &{\hfill inequalities    }&
  expr_list::=expr | expr@,@ expr_list &{\hfill expression lists    }&

  //QUERIES
  query ::=query$^{\circ}$ | query UNION query | query UNION ALL query &{\hfill unoins }&
  query$^{\circ}$::=RETURN ret | clause query$^{\circ}$ &{\hfill sequences of clauses }&
  ret::= @*@ | expr [AS $a$] | |ret@,@ expr [AS $a$] &{\hfill return lists }&

  //CLAUSES
  clause::= [OPTIONAL] MATCH pattern_tuple [WHERE expr] &\underline{[\color{blue}{FOR TT AS OF} \color{black}{expr} | \color{blue}{FOR TT FROM} \color{black}{expr} \color{blue}{TO} \color{black}{expr}]}& &{\hfill matching clauses }&
            |WITH ret [WHERE expr] | UNWIND expr AS $a$ $a \in \mathcal{A}$ &{\hfill relational clauses }&
  pattern_tuple::= pattern | pattern@,@ pattern_tuple &{\hfill tuples of patterns }&
\end{lstlisting}
\end{figure*} % 结束双栏
}

\todo{In our model, each graph object comprises multiple corresponding versions, including one current/latest version and potentially several historical versions. 
Unlike existing works such as T-GQL, which assigns a time period to each graph object, our model assigns the time period to each version of a graph object (vertex or edge).
For example, as depicted in Figure~\ref{fig:exp2},
consider updating the entity ``Phone''.
In existing models, this update results in retaining two entire ``Phone'' vertices within the same graph, leading to two redundant unchanged ``Owns'' edges.
In contrast, we create a new version of the ``Phone'' vertex and re-link the ``Owns'' edge to this version, with changed attributes stored in the historical version.
Consequently, our model is less complex but more efficient by avoiding the creation of redundant vertices and edges.}
% A graph object (vertex or edge) may have one or more versions, each associated with $\omega$ to convey various semantics at different time periods. 
At any given time $t$, a graph object version with $\omega=[st,ed)$ is said to be \textbf{{\visible}} if  $st\leq t<ed$. We classify a graph object version as a current version if it is legal at the current time, and
as a historical version if it is not legal at the current time.

\todo{
\textbf{{Graph operations.}} 
Our temporal model supports diverse graph operations as follows. Assume these graph operations are issued by a transaction committed at time $t_1$.
\begin{itemize}[leftmargin=*,itemsep=2pt,topsep=0pt,parsep=0pt]
\item \textit{Creating} a vertex or an edge: This involves adding a vertex or edge with 
a current version having a time period $\omega=[t_1, +\infty)$.
\item \textit{Deleting} a vertex or an edge whose current version is with $\omega=[st, +\infty)$: This entails updating $\omega$ to  $[st, t_1)$.
\item  \textit{Updating} a vertex or an edge whose current version is with $\omega=[st, +\infty)$: 
This marks the current version as a historical version by updating $\omega$ to $[st, t_1)$ and generates a new current version with $\omega=[t_1, +\infty)$ representing the up-to-date semantics. 
\end{itemize}
}

%对应graph model的时间，举例子
% \vspace{-2mm}
% \begin{example}
% \label{exp:tpgm}
\textbf{Example {\expmodel}.} 
According to Definition \ref{def:tpgm}, we present the corresponding temporal property graph of Example {\expgraph} in Figure \ref{fig:exp_model}. Here, $\Omega=[t_0, +\infty]$, $N=[ v_0, v_1, v_2, v_3]$, $E=[ e_0, e_1, e_2, e_3]$. Each vertex and edge owns a current version and several historical versions from $t_0$ to $t_n$. 
For brevity, we omit graph states before $t_n$. At $t_n$, there exists three vertices ($v_0$, $v_1$ and $v_2$) and two edges ($e_0: (v_0,v_1))$ and $e_1:(v_0,v_2)$).
For instance, $v_2$ owns a current version $v_{2.7}$, which has 
a unique id $2$, a label ``Phone'', a property with the key-value pair (IP, Singapore), and a lifespan $\omega=[t_{n},\infty)$. 
Subsequently, at $t_{n+1}$, consider there is a customer purchase transaction.
% to change the graph state. 
It updates the properties of $v_1$ and $v_2$, resulting in new versions for each of them. 
Take $v_2$ as an example: it marks $v_{2.7}$ as a historical version and generates a new current version $v_{2.8}$. 
Specifically, $\tau$ maps $v_{2.8} \times {IP} \times [t_{n+1},+\infty)$ to $\operatorname{\textit{New York}}$ and maps $v_{2.7} \times {IP} \times [t_{n},t_{n+1})$ to $\operatorname{\textit{Singapore}}$. 
% \makesure{ Note $v_2$ exists during the {\visible} periods from $v_{2.0}$ to $v_{2.8}$. $\sigma$ maps $(v_2) \times [t_{0},+\infty)$ to $true$, indicating $v_2$ is created at $t_0$ and exists in the database from then on.}
% \extended
We regard $v_{2.7}$ as legal at $t_n$, but not legal at $t_{n+1}$.
{Moreover, this transaction also creates $v_3$, $e_2$ and $e_3$, which have only the current version with $\omega=[t_{n+1}, +\infty)$.
All the aforementioned graph operations adhere to the defined constraints. 
For instance, $e_2$ can be successfully  created at $t_{n+1}$ only after verifying linked vertices $v_3$ and $v_1$ exit at $t_{n+1}$ (Constraint 1).}  \qed

\subsection{Temporal Graph Query Language}
\label{sec:temporal_query_language}
% {\tgdb} is designed to be an online system that inherently integrates temporal support.
% Therefore, b
{\tgdb} incorporates a temporal-enhanced Cypher~\cite{cypher}, which extends the standard syntax defined in OpenCypher \cite{cypher} 
% \cz{of[jm:for]} read queries
to support temporal queries.
% In addition, {\tgdb}
% extends the standard syntax defined in OpenCypher~\cite{cypher} to support temporal queries, 
As illustrated in Listing \ref{lst:cypher}, {\tgdb} introduces two temporal syntax extensions in the \texttt{MATCH} clause \maintext{(line 3):}
\extended{(underlined in {Listing \ref{lst:cypher}}):}
% it introduces two syntax extensions within the \texttt{MATCH} clause of Cypher (line 3):
% For fetching historical data, there are two syntax extensions introduced: 
(1) FOR $TT$ AS OF $t$, which retrieves all graph objects {\visible} at time $t$, and 
(2) FOR $TT$ FROM $t_1$ TO $t_2$, which locates all graph objects consistently {\visible} within the time range from $t_1$ to $t_2$. 
The former is referred to as ``time-point'' queries, while the latter is known as ``time-slice'' queries. 
Users can apply any time conditions to temporal queries,
spanning a wide time range from the oldest historical records up to the most recent updates.
Further, apart from retrieving temporal graph data of user interest using temporal queries, {{\tgdb} allows users to submit common (non-temporal) queries and data manipulation operations (creating, updating, and deleting) with the standard Cypher syntax.} 

% \begin{comment}
\textbf{Example {\expquery}.}
\todo{Consider the query ``What was Jack's phone IP at $t_n${''}. This query can be answered by issuing the following statement, where the temporal syntax is underlined: ``MATCH (:Customer {name: `Jack'})-[r]-(p:Phone) \uline{FOR TT AS OF $t_n$} return p.IP. ''
}

\begin{figure*}[]
\centering    %居中
\includegraphics[width=0.99\textwidth]{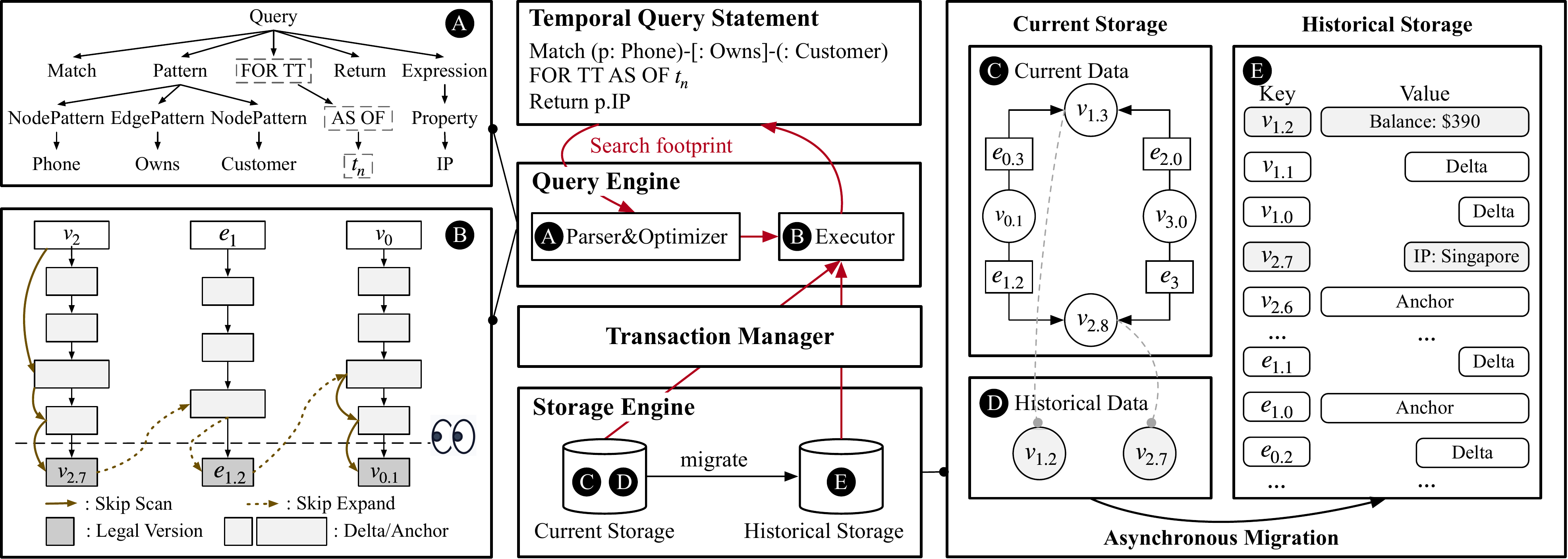}
% \centering
% \vspace{-2mm}
\captionsetup{justification=raggedright}
\caption{An Overview of {\tgdb} - {\rm {\tgdb} consists of a temporal query engine, a hybrid storage engine, and an MVCC-based transaction manager.
We employ an "anchor+delta" strategy to reduce the historical storage overhead, while using an anchor-based version retrieval technique to ensure efficient temporal query processing.}}
% \vspace{-3mm}
\label{fig:overview}
\end{figure*}

% \vspace{-3mm}
\section{System architecture} 
\label{sec:overview}
In this section, we introduce the system architecture of {\tgdb} as shown in 
% the middle part of 
Figure ~\ref{fig:overview}. 
% \todo{
{\tgdb} includes a transaction manager 
% component to support transactions, 
which enables handling a sequence of graph operations with ACID properties.
We process transactions by employing the Multi-Version Concurrency Control (MVCC)~\cite{mvcc}.
% utilize the transaction management component to process transactions, 
MVCC ensures that transactions only see a consistent snapshot of the data that is \textbf{visible} to them, thus enabling multiple transactions to work concurrently without interfering with one another~\cite{mvcc,mvcc2,SI1,SI2,SI3,DBLP:journals/tkde/ZhaoLZHLPD23}.
Given our primary focus on temporal data management, we now describe how we utilize the MVCC mechanism to manage temporal data effectively.
% The system's fundamental unit is a transaction, which manages  
{\tgdb} supports built-in temporal features through two major components: the storage engine and the query engine. 

 % (create, read, update, and delete)
\subsection{Storage Engine}
\label{sec:overview_storage}
The storage engine of {\tgdb} has two physically isolated storages: current storage and historical storage.
The current storage typically maintains the 
current versions
% most recent versions 
of graph objects.
% \todo{It handles various graph operations and serves both read and write queries.}
% and serves for both read and write requests.
% both read-only, write-only, and read-write query requests.
In contrast, the historical storage manages 
historical versions
% the previous versions 
of graph objects, which are 
% \todo{generated by graph operations} and 
asynchronously migrated from the current storage. 

\textbf{Current storage.} As discussed in Section \ref{sec:graph_model}, graph involves under various graph operations. 
To efficiently record these changes, {\tgdb} builds its current storage as a multi-version storage, maintaining multiple versions for each graph object.
Each graph object includes one current version
% most recent version 
retaining the up-to-date state and is linked to a list of historical versions preserving the previous states. 
When a graph object is updated,
% or deleted, 
instead of directly overwriting the data, we create a new current version and move the previous one to the list of historical versions.
% create a new version and move the previous version to the list of historical versions.
We further integrate time dimensions into the data layout and modification paradigm to trace accurate graph evolution.
% following the MVCC mechanism.
% Based on this data layout, we design a modification paradigm on the graph object to ensure the graph evolution is precisely maintained. 
We will introduce the details 
% including its data layout and modification paradigm, 
in Section \ref{sec:current-storage}.

\textbf{Historical storage.} 
{\tgdb} does not store historical versions in the current storage permanently. 
Instead, we migrate them to the historical storage for long-term maintenance. 
To handle the potentially large volume of historical data, we properly compress the historical storage.
% of {\tgdb}.
% utilizes a compact historical storage model while ensuring fast access.
We organize migrated historical versions in a key-value format. 
% groups three types of deltas (VP, EP, and VE deltas) into their respective segments. In each segment, the historical delta is organized in a key-value format. 
% Each key-value pair is a version of a specific graph object.
The key contains the metadata of a historical version, including vertex/edge id and version's lifespan $\omega$, while the value holds detailed properties of this version.
Instead of retaining all properties for every version, we organize versions 
% of a vertex or edge 
in an ``anchor+delta'' manner.
We utilize deltas to record relative differences between subsequent versions, minimizing the storage cost of ever-growing historical data.
In addition, after a series of deltas, we maintain an anchor to preserve the complete state of a graph object, facilitating the reconstruction process when executing temporal queries. 
We will introduce 
% the design of historical storage 
the details in Section \ref{sec:historical-storage}.
% By doing so, the historical storage has compact storage overhead while enabling efficient retrieval simultaneously.

%historical storage: anchor+delta
%migration
\textbf{Asynchronous migration.} 
{\tgdb} utilizes an asynchronous migration approach to transfer historical data from the current storage to the historical storage.
Rather than triggering a migration immediately following an update, this migration is postponed and occurs during the garbage collection of MVCC. 
This design ensures that transferring ever-growing historical data is lightweight, minimizing its overhead on the current storage. 
% excessively burden the current storage. Furthermore, this migration is asynchronous, makiitsng it
% lightweight and non-intrusive to the current storage.
We will present our asynchronous migration in Section \ref{sec:historical-storage}.

% \begin{example}
% \label{exp:overview_storeg}
\textbf{Example {\expoverviewstorage}.}
In the right part of Figure \ref{fig:overview}, we demonstrate how {\tgdb} stores the customer purchase graph as presented in Example {\expmodel}. 
% Take the customer purchase graph in Figure \ref{fig:exp} as an example. 
In the current storage, component \numberedcircle{C}records the current versions at $t_{n+1}$.
Besides, the historical versions
% previous versions, representing the previous graph state 
at $t_n$ ($v_{1.2}$ and $v_{2.7}$ in this case), \todo{are} stored in component \numberedcircle{D}.
Take $v_2$ as an example. To capture the change in $v_2$'s IP from Singapore to New York at $t_{n+1}$,
{\tgdb} performs two steps. First, it
updates $v_2$ in place to create a new current version $v_{2.8}$.
% , representing the current state. 
Second, to maintain the previous state, {\tgdb} generates a historical version $v_{2.7}$, which is linked to $v_{2.8}$ in a chain and managed by MVCC.
We migrate historical data in component \numberedcircle{D}to the historical storage (component \numberedcircle{E}) asynchronously.
% In Figure \ref{fig:overview} \numberedcircle{E}, 
In the historical storage, the historical versions, $v_{1.2}$ and $v_{2.7}$, are organized as an anchor (represented as a long rectangle) and a delta (represented as a short rectangle).
\qed

\subsection{Query Engine}
% \textbf{Query engine.} 
The query engine is responsible for handling user-issued queries, retrieving relevant graph data from the hybrid storage engine.
Adhering to the ``textbook'' separation of components, {\tgdb} consists of a parser,
% a logical planner, 
an optimizer, 
% a physical planner, 
and an executor.
While inheriting those components from existing graph databases, 
{\tgdb} further extends them to support temporal queries.

\textbf{Parser and optimizer.} The parser translates queries and generates the corresponding syntax tree for the query optimizer.
% Besides handling regular non-temporal queries defined in Cypher~\cite{cypher}, {\tgdb} also supports temporal queries defined in Section~\ref{sec:temporal_query_language}. 
To accommodate the syntax of
temporal queries, {\tgdb} extends its lexical, syntactic, and semantic analyses to recognize time qualifiers as defined in Section~\ref{sec:temporal_query_language}.
% The query engine then employs
Leveraging the resulting syntax tree, the optimizer generates the execution plan for the executor.

\textbf{Executor.} 
\todo{
{\tgdb} builds upon and extends two core fundamental operations from traditional graph databases: scan and expand. The scan operator retrieves the required vertex versions for each query, while the expand operator fetches relevant edge and adjacent vertex versions. 
% While inheriting common capabilities from existing executors, w
We enhance these operators to provide the consistent and efficient processing of temporal queries.
}
To ensure consistent query results, we obtain current and historical data separately from two storage engines and then combine the results together. For current data, we follow the conventional query mechanism, which simply executes the plan over the current storage and captures visible graph object versions under MVCC's snapshot visibility check~\cite{mvcc}. 
However, accessing historical data solely from the historical storage may yield incomplete results. Due to asynchronous migration, a portion of data is still in the current storage. 
% For historical data, however, accessing solely from the historical storage could yield incomplete results. Due to asynchronous migration, a partial of data are still in the current storage.  
To address this, we introduce a legal check mechanism that retrieves relevant data from both storages. 
This mechanism verifies if a version is legal within the given time condition to extract appropriate versions.
Note that the snapshot visibility check is required when retrieving historical versions in the current storage.
These steps ensure that the requested version(s) is from the consistent snapshot(s), thereby guaranteeing consistency.

\todo{To efficiently traverse historical versions from substantial historical data, we propose an anchor-based version retrieval technique to minimize unnecessary traversals.
For the scan operator, we fetch relevant vertex versions that satisfy the provided temporal condition. 
To reconstruct a desired version, we directly locate the nearest anchor with a lifespan $\omega$ aligning with the query time constraint. Subsequently, we traverse subsequent deltas from the obtained anchor, applying all fitting deltas.
Regarding the expand operator, we further eliminate unnecessary traversals by directly locating the corresponding edge and adjacent vertex anchors using acquired held vertex versions.
Further elaboration 
% on the anchor-based version retrieval technique 
can be found in Section \ref{sec:query-engine}.
}

\textbf{Example {\expoverviewquery}.}
Figure \ref{fig:overview} \numberedcircle{A} depicts a simplified syntax tree for a given temporal query statement. Based on it, {\tgdb} then utilizes the executor to fetch query results from the hybrid storage engine. 
Figure \ref{fig:overview} \numberedcircle{B} illustrates the search footprint of the given query statement to answer ``What were the phone IPs of all customers at $t_n$''.
% to answer ``What was Jack's phone IP at $t_n$''. 
We only reconstruct four relevant graph vertices/edges.
% with an efficient skip loop-up strategy.
We start to scan the vertex $v_2$, which we are interested in. We skip to seek its nearest anchor and collect all relevant deltas, to reconstruct the {\visible} version $v_{2.7}$ we want. 
We then expand $v_{2.7}$ to get its linked edge $e_{1.2}$ and adjacency vertex $v_{0.1}$ without traversing the entire version chain of $e_1$ and $v_0$. 
% {\color{red} All operators are under transaction management to ensure data consistency.}
\qed
% \vspace{-3mm}
\section{Hybrid Storage Engine}
\label{sec:hybrid_storage}
In this section, we now elaborate on the design of {\tgdb} 's hybrid storage engine.

\subsection{Current Storage} 
\label{sec:current-storage}
Inheriting existing native graph databases~\cite{GuptaMS21}, {\tgdb} organizes graph data into three storage components: (i) vertex properties (VP), (ii) edge properties (EP), and (iii) graph topology, i.e., vertex's incoming and outgoing edges (VE). 
Like most native graph databases~\cite{neo4j, Memgraph,graphflow}, we retain the topology within the vertices, enabling swift neighborhood traversal for each vertex. However, it is not trivial to record 
graph evolution under this design. The graph could change in not only its semantics, i.e., properties of graph objects, but also its structure. We have to identify different types of operations applied on the graph. For example, we should prevent the creation of a new vertex version when the vertex's relevant graph topology changes but its properties remain unchanged. 
To address this problem, we associate the time dimension to each independent storage component to separately record semantic changes and structural changes. 

\begin{figure}[]%width=1\textwidth. scale=0.6
\centering    %居中
\includegraphics[width=0.45\textwidth]{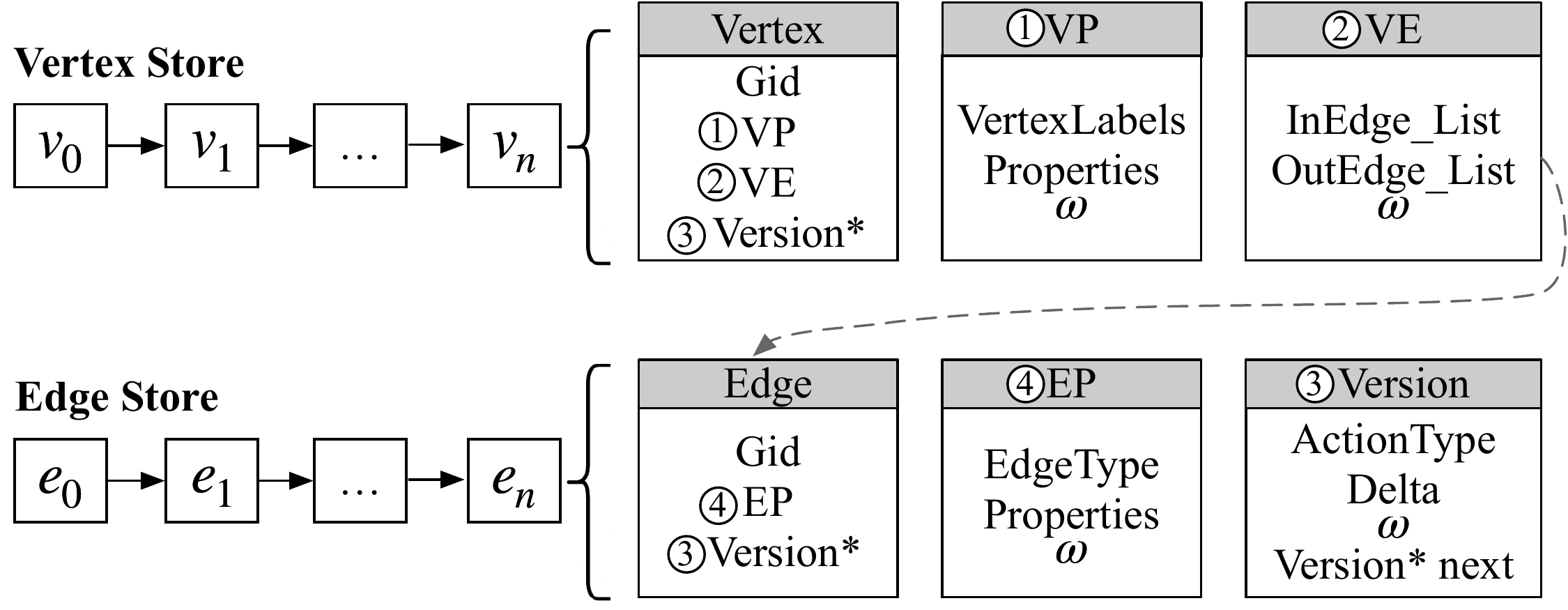}
\centering
% \vspace{-2mm}
 \centering
 \caption{Data Layout of Current Storage}
\label{fig:current_storage}
  % \vspace{-2mm}
\end{figure}
\label{sec:hybrid-storage}

\textbf{Data layout.} As shown in Figure \ref{fig:current_storage}, the data store comprises two components: the Vertex Store, which maintains a list of vertex objects, and the Edge Store, which stores a list of edge objects.
Every vertex object $v$ has a unique graph identifier $Gid$, a VP part, a VE part, and a pointer to a linked list of historical versions (version chain). While every edge object $e$ has a unique identifier $Gid$, an EP part, and a pointer to a linked list of historical versions. Specifically, 
\begin{itemize} %enumerate itemize
[leftmargin=*,itemsep=2pt,topsep=0pt,parsep=0pt]
\item The VP part stores a set of vertex labels and property value pairs associated with the current version of $v$, along with a time period $\omega$ indicating 
$v$'s current semantic lifespan.
\item The VE part keeps track of the current version of $v$'s incoming and outgoing edges, with each entry in a list of (edge $Gid$, neighbor vertex $Gid$) pairs. The VE part also includes $\omega$ to record 
% the structural lifespan of $v$'s current version.
$v$'s current structural lifespan. 
\item The EP part stores an edge type and property value pairs of the current version of $e$, along with a time period $\omega$ indicating $e$'s current semantic lifespan.
\item Each historical version contains: an action type indicating the changes made to a VP part, VE part, or EP part, a delta recording the steps to revert the changes to restore the previous version, a time period $\omega$ capturing the lifespan of the historical version, and a pointer to the next historical version.  All historical versions generated by the same transaction are clustered in an undo buffer following the MVCC mechanism.
\end{itemize}

\textbf{Modification paradigm.} We now discuss how 
{\tgdb} evolves the graph to handle various graph operations.
% As discussed in Section \ref{sec:graph_model}, graph operations include (1) creating vertices or edges, (2) deleting vertices or edges, and  (3) updating the properties of vertices or edges.
Suppose the graph operation is invoked by a transaction $T_i$ whose commit time is $t_i$. The modification paradigm on the graph data layout is as follows.
\begin{enumerate} %enumerate itemize
[leftmargin=*,itemsep=2pt,topsep=0pt,parsep=0pt]
\item When a vertex is created, we create a vertex object, set its VP part's time as $\omega=[t_i, +\infty)$, set its VE part's time as $\omega=[-\infty, +\infty)$, and link it in the vertex object lists. 
\item When an edge is created, we create an edge object, set its EP part's time as $\omega=[t_i, +\infty)$, and link it in the edge object lists. We also create connections to relevant vertex objects by setting their VE part's time to $\omega=[t_i, +\infty)$ if their previous VE part's time is $[-\infty, +\infty)$.
\item When updating/creating/deleting a property value of a vertex object with $\omega=[t_j, +\infty)$, we first update relevant property values in the VP part and set its time as $\omega=[t_i, +\infty)$. 
Next, we create a historical VP version capturing the state of the vertex prior to the modification and set its $\omega$ as $[t_j, t_i)$. This VP version is then linked to the vertex's version chain. 
Updating a property value of an edge follows the same logic.

\item When a vertex is deleted, we first delete all property values of the relevant vertex object, following the (3) paradigm, and then delete all connected edges, following the (5) paradigm.

% decompose it into the deletion of all property values of the relevant vertex object, followed by the (3) paradigm, and deletion of all linked edges, followed by the (5) paradigm. 

\item When an edge is deleted, we decompose it into the deletion of all property values and the deletion of connections with relevant vertices. 
The former acts on the edge object, following the (3) paradigm.
% Suppose the edge state chunk is associated with time $\omega=[t_j, +\infty)$. We first clear all property values of the chunk and place a flag ``delete'' to notify the system and then create an EP delta to record all property values with $\omega=[t_j, t_i)$. 
The latter acts on the source and destination vertex object. Take the source vertex object with $\omega=[t_j, +\infty)$ as an example. We first update its VE part to delete this edge from outgoing edge lists and set its VE part's time as $\omega=[t_i, +\infty)$. We then create a VE version to record this deleted edge with $\omega=[t_j, t_i)$ and link it to the edge's version chain.
\end{enumerate}

\begin{figure}[]%width=1\textwidth. scale=0.6
\centering    %居中
\includegraphics[width=0.45\textwidth]{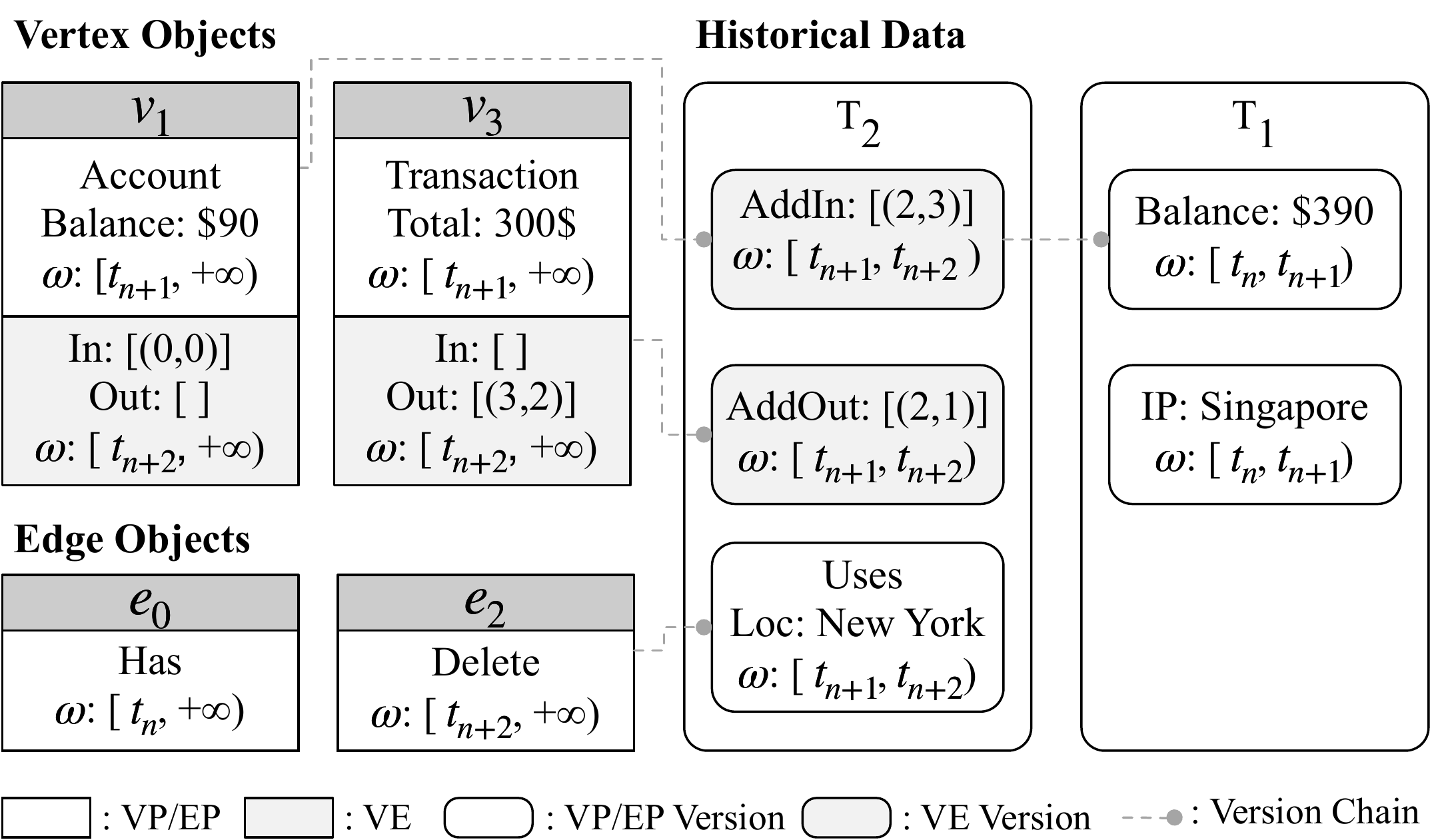}
\centering
% \vspace{-2mm}
 \centering
 \caption{An Example of Current Storage Layout}
  % Figure \ref{fig:exp_model}
\label{fig:current_storage2}
  % \vspace{-2mm}
\end{figure}

\textbf{Example {\expocurrent}.}
We illustrate the data layout of current storage.
As shown in Figure \ref{fig:current_storage2}, we reconsider Example {\expmodel}.
To further represent the structural change, we suppose an event deleting $e_2$ at $t_{n+2}$.
We focus on two vertices $v_1$ and $v_3$, and two edges $e_0$ and $e_2$ to showcase the graph evolution. 
% Each vertex object contains a graph identifier, a VP part, a VE part, and a 
% pointer to a linked list of historical versions. Similar to the edge object. 
At $t_{n+1}$, the transaction $T_1$, representing a customer purchase, is committed. It updates the VP part of $v_1$ and $v_2$, generates two VP versions linked to them, and creates graph objects $e_2$, $e_3$, and $v_3$. Figure \ref{fig:current_storage2} shows these elements except $e_3$ and $v_2$ due to space limitations.
At $t_{n+2}$, $T_2$ is committed to delete $e_2$, which affects $v_1$, $v_3$, and $e_2$ objects. It first acts on $e_2$'s EP part to clear all semantic information and generates an EP version to record the previous edge state. Then, it acts on the VE part of \todo{$v_1$ and $v_3$} and generates two VE versions. \qed
% \end{example}

\subsection{Historical Storage}
\label{sec:historical-storage}

%历史数据在KV中的组织
In MVCC, historical versions are not retained in the current storage permanently. Instead, once these versions are no longer needed by any active transaction, they are safely removed through garbage collection (GC) to optimize the performance of the current storage.
{\tgdb} utilizes this mechanism to transfer those inaccessible versions to the historical storage for long-term maintenance. For the sake of communication, historical versions in the current storage are referred to as ``unreclaimed'', while those in the historical storage are referred to as ``reclaimed''. In this subsection, we first present the optimized key-value format used for storing historical versions and then outline the process of migrating unreclaimed versions into the historical storage.

\maintext{\begin{figure}[t!]%width=1\textwidth. scale=0.6
\centering    %居中
\includegraphics[width=0.32\textwidth]{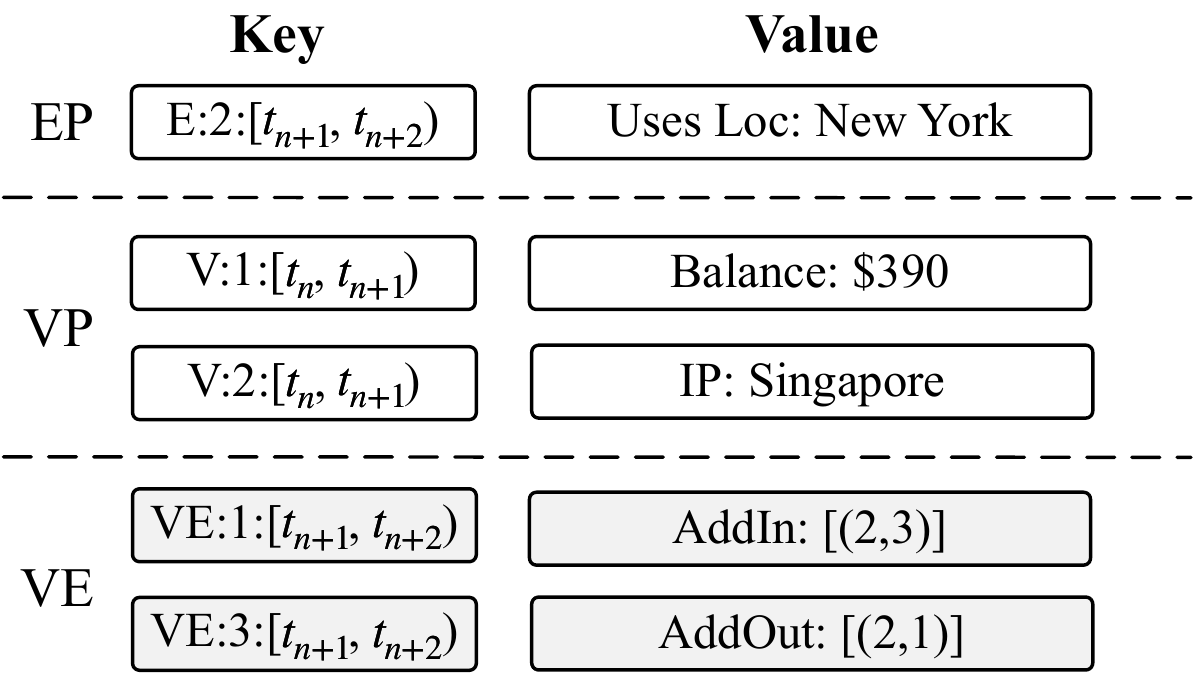}
% \vspace{-2mm}
 \centering
\caption{Key-value Format in Historical Storage}
\label{fig:kvformat}
  % \vspace{-2mm}
\end{figure}}

%key value的组织形式，总体是key:type+Gid+tt value就是deltas剩余的信息。
\textbf{KV format.} 
Reclaimed versions are organized in a key-value format, where the key represents the metadata of the version and the value stores the corresponding detailed information. 
% Three types of deltas, VP, EP, and VE deltas, are grouped into their respective segments.
{\tgdb} groups three types of historical versions (VP, EP, and VE versions) into their respective segments.
In each segment, the key is formed by combining three elements: the $Type$ prefix, the graph identifier $Gid$, and $\omega$ of the version. The $Type$ prefix indicates the type of version contained in the segment: `V' for VP versions, `E' for EP versions, and `VE' for VE versions. 
The $Gid$ is a unique identifier of the graph object linked to the version, while  $\omega$ represents the version's lifespan. As for the value field, it contains the remaining semantic information, i.e., the delta of the version recording only data changes compared to the previous version. 
% This design guarantees that the key is unique for the entire temporal graph. 
By organizing the data in this way, 
% It is worth mentioning that 
data sharing the same prefix in the key are physically clustered together in a SkipList~\cite{rocksdb}, which ensures different versions of the same entity are automatically sorted based on their lifespan $\omega$. As a result, it becomes efficient to retrieve in chronological order. 
Figure \ref{fig:kvformat} depicts the reclaimed historical versions' KV format of unreclaimed historical versions in Figure \ref{fig:current_storage2}.

\extended{\begin{figure}[t!]%width=1\textwidth. scale=0.6
\centering    %居中
\includegraphics[width=0.32\textwidth]{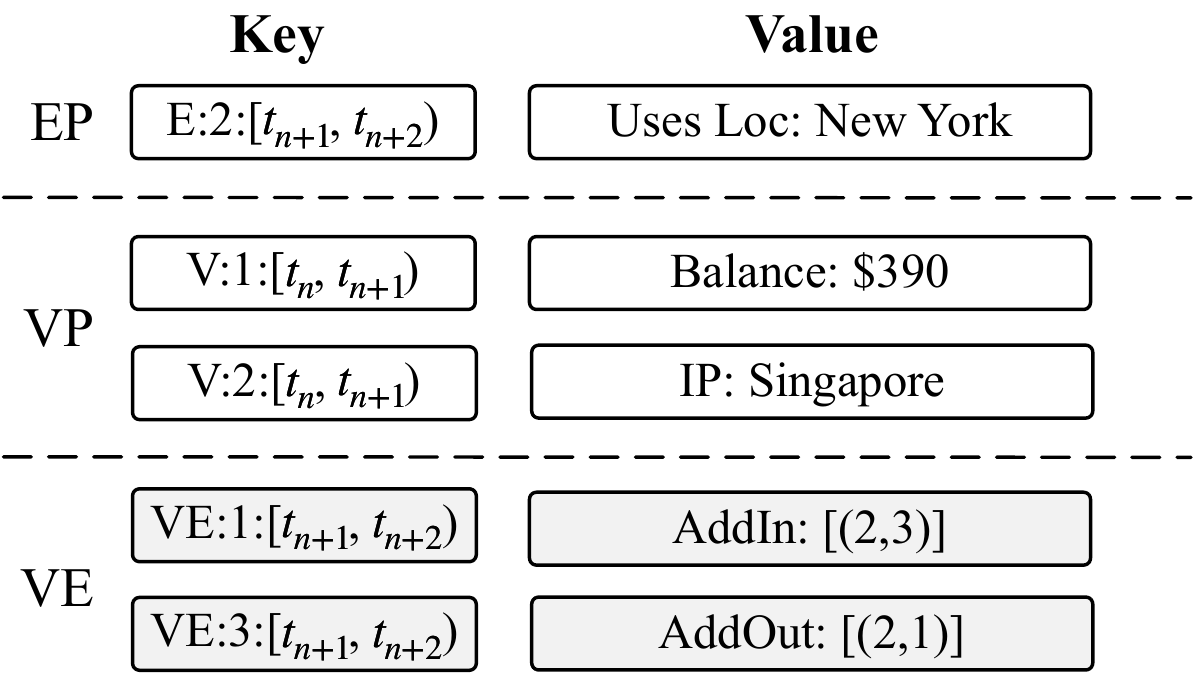}
\vspace{-2mm}
 \centering
\caption{Key-value Format in Historical Storage}
\label{fig:kvformat}
  % \vspace{-2mm}
\end{figure}}

\textbf{Anchor+delta.} We utilize deltas to reduce the storage overhead. However, retrieving a reclaimed graph object requires assembling the latest version with all previous deltas, incurring significant reconstruction costs for long retrieval histories. To mitigate this, 
% regular 
we introduce anchors at intervals in the delta data, where an anchor represents the complete state of a graph object. To differentiate anchors from deltas in the KV store, we append a one-bit character suffix to the key's $Type$, where `A' denotes anchors and `D' represents deltas.
% Anchors can shorten the recovery chains, thereby accelerating the reconstruction process. 
Specifically, to reconstruct a certain reclaimed version $o_1$, we seek its most recent anchor $o_2$, collect all deltas from $o_2$ to $o_1$, and combine them to reconstruct $o_1$.

\todo{ 
% It is intuitive to globally place an anchor after a constant $u$ value of deltas for each graph object. 
% \cz{The evolution frequency varies among graph objects, with some undergoing frequent changes (e.g., a phone's IP) while others change rarely or never (e.g., a customer's name). To address this variability,}%[jm:fillspace]
We propose an adaptive anchoring approach, which assigns different anchor intervals $u$ for different graph objects.
A higher $u$ leads to more deltas between successive anchors, potentially increasing query latency but reducing storage overhead. 
Therefore, we assign a larger $u$ to frequently updated objects to strike a balance between query latency and storage efficiency.
Given a graph object $o$, we use Equation \ref{eq:anchor_threshold} to determine its $u$ according to the update frequency $f(o)$, the number of updates conducted on $o$.
\begin{equation}
u_o=\left\{\begin{array}{lll}
\tau_1 *c & f(o) \leq \tau_1 & \text { low frequency } \\
\tau_2 *c & \tau_1< f(o) \leq \tau_2 & \text { medium frequency } \\
\tau_2 ^2 /\tau_1 *c & \tau_2 \leq f(o) & \text { high frequency }
\end{array}\right.
\label{eq:anchor_threshold}
\end{equation}
}This equation categorizes update frequencies into three levels (low, medium, and high) using two thresholds ($\tau_1$ and $\tau_2$).
Each frequency level is assigned a specific anchor interval, calculated heuristically by multiplying the respective threshold values with a predefined parameter $c$.
% We determine the graph object $o$'s anchor interval $u$ by assessing which frequency level $f(o)$ falls into.
Currently, {\tgdb} enables users to set parameters in Equation~\ref{eq:anchor_threshold}, such as $\tau_1$, during database initialization and runtime. 

\textbf{Data migration.} 
In MVCC, unreclaimed historical versions will be physically removed from the current storage through an asynchronous GC phase when their relevant commit transactions are no longer active. 
% Those inaccessible versions are pruned by a Garbage Collector periodically.
{\tgdb} collects those versions and migrates them to the historical storage for long-term maintenance, as detailed in Algorithm \ref{alg:migration}. 
We use $undo$ to maintain the unreclaimed version to be migrated (line 2) and $kv$ to store reclaimed data in a key-value format (line 3).
Each unreclaimed version $undo$ in the $CT$ is initially encoded into a key-value pair $kv$ (line 5). Subsequently, we store it in the historical KV store (line 6).
Finally, we lock $undo$ in the version chain and physically delete it  (line 7).

\begin{algorithm}[tp]
 \small
 \SetKwInOut{Input}{input}
 \SetKwInOut{Output}{output}
 \caption{Data migration}
 \label{alg:migration}
 \SetKwFunction{FMain}{Migrate}
    \SetKwProg{Fn}{Function}{:}{}
    \Fn{\FMain{$CT$}}{
    \KwIn { $CT$, committed transaction no longer active;}
    %\KwOut  {$\Sigma$, the result set;}$ deltas  \longleftarrow Rec.deltas$\;
    $undo \leftarrow \emptyset$; //unreclaimed version\;
    $kv \leftarrow \emptyset$;
    //reclaimed version\;
    % $undos \leftarrow undos \cup$ \texttt{getUndos}($CT$)\;
    % reclaim versions of $CT$\; 
    \ForEach{$ undo \in CT$}
    {
       $kv$=\texttt{encode2KV}($undo$)\;
       \texttt{KV\_store::put}($kv$)\;
       physically delete undo\;
       % $kvs \leftarrow kv$
    }
    % \texttt{KV\_store::putMultiples}($kvs$)\;
    % reclaimed versions of $CT$\;
}
\textbf{End Function}
\end{algorithm}
\section{Temporal Query Engine}
\label{sec:query-engine}
{\tgdb} inherits and extends scan and expand operators to empower consistent and efficient temporal query processing.

\subsection{Scan Operator}
\label{sec:scan}
{\tgdb} uses the scan operator to efficiently fetch vertex versions while ensuring data consistency for both current data and historical data.
% with both efficient and consistent capability. 
We elaborate on it from the aspect of fetching data from each storage component.
% on fetching temporal graph data from two isolated storage components: the current storage and the historical storage. 
% design two strategies for efficiently and consistently retrieving data from two isolated storage.
% \vspace{-4mm}
When fetching data from the current storage, it is essential to ensure consistent data capture in the presence of concurrent transactions. To achieve this, we start by locating relevant vertex object(s) of interest.
For each vertex object, we first employ the snapshot visibility check~\cite{mvcc} to find a visible version of the given transaction. 
All versions preceding this visible version in the version chain are candidate legal versions we may want. 
Then we utilize a {\visible} check mechanism, which
verifies whether each candidate version $v'$ is {\visible} to the given query time condition, as per the following equation.
% \vspace{-2mm}
\begin{equation}
% o.\omega.st \leq {C}.t_2 \wedge o.\omega.ed>{C}.t_1
% \texttt{TemporalCheck}(\omega, C)=
\omega.st \leq {C}.t_2 \wedge \omega.ed>{C}.t_1
% \vspace{-1.5mm}
\end{equation}
Here, 
% $o$ is the historical version being evaluated; 
$\omega.st$ and $\omega.ed$ represent the start and end time of $v'$'s lifespan, respectively; $C$ represents the time condition of the given query with begin time $t_1$ and end time $t_2$. For a time-point query, $t_1=t_2$. 

When fetching data from the historical storage, there is no need to handle transaction conflicts as the historical storage serves read-only queries that users cannot change the data in the historical storage.
Therefore, we directly employ the {\visible} check mechanism to get desired versions.
To further enhance digging out historical versions, we employ an anchor-based skip retrieval strategy to reconstruct desired versions. To restore a specific {\visible} version $v'$, we directly seek the most recent anchor $v$ in the KV store by the probe prefix ``AV:$id$:$C$'', where `AV' represents the anchors in the VP segment, $id$ is the unique id of interest vertex and $C$ is the given query time constraint. We then assemble $v'$ with all previous versions from $v$ to $v'$. Thanks to the special design of the key-value format in the historical storage, we can leverage the probe prefix to swiftly find the nearest anchor.

\begin{algorithm}[tp]
 \small
 \SetKwInOut{Input}{input}
 \SetKwInOut{Output}{output}
 \caption{Retrieving vertices}
 \label{alg:vertex}
 
 \SetKwFunction{FMain}{VertexRead}
    \SetKwProg{Fn}{Function}{:}{}
    \Fn{\FMain{$C$}}{
    \KwIn { $C$,  temporal condition;}
    \KwOut  {$\Sigma$, the result set;}
    $ v  \leftarrow$ the vertex  which we start to scan\;
    \While{$v$}{
     % $ flag  \longleftarrow $true \;
     // fetch from the current storage\;
       \ForEach{$ v' \in (v \cup v.versions )$ } 
        {
             \lIf{{\texttt{!SnapshotCheck}}($v'$)}{ continue}
             \If{\texttt{TemporalCheck}($v'.\omega, C$)}{
           $\Sigma$ $\leftarrow$ $\Sigma$ $\cup$ $\texttt{Reconstruct}(v')$\;
         }
        }
    // fetch from the historical storage\;
     % \lIf{$flag$}{  \texttt{FetchFromKV}($v.id, C, \Sigma, flag$)}
     \texttt{FetchFromKV}($v.id, C, \Sigma$)\;
        $ v  \leftarrow v$.\texttt{next()}\;
    }
    \textbf{return} $\Sigma$\;
   }
% \textbf{End Function}

\SetKwFunction{FMain}{FetchFromKV}
    \SetKwProg{Fn}{Function}{:}{}
    \Fn{\FMain{$id, C, \Sigma$}}{
    {
     % $kv_{a} \longleftarrow$ $o$'s oldest version from current storage\;
       %\ForEach{$ Gid \in Set(Gid) $}
        {
        $ kv_{a}  \leftarrow$ KV\_store::\texttt{seeknext}($id, C$)\;
        % \lIf{$it_{a}$}{$kv_{a} \leftarrow \operatorname{copy}\left(it_{a} \rightarrow kv\right)$}
        $ kv_{d}  \leftarrow$ KV\_store::\texttt{seeknext}($ id, kv_{a}.key$)\;
        %$num\leftarrow 0$\;
        \While{$kv_{d}$ $\wedge$ $kv_d.\omega.st \leq C.t_2$}{
        %$num++$\;
        $kv_{a} \leftarrow$ \texttt{combine} $(kv_{a}, kv_{d})$\;
        %\lIf{num==1}{$kv_{a} \leftarrow null$ }          
          % \texttt{TemporalCheck}($kv_{a}$, $C$, $\Sigma$, $flag$ )\;
          \If{\texttt{TemporalCheck}($kv_{d}.\omega, C$)}{
           $\Sigma$ $\leftarrow$ $\Sigma$ $\cup$ $kv_{a}$\;
         }
          $kv_{d} \leftarrow kv_{d}.\texttt{next}()$\;
   	 }
    }

    }
	}

\end{algorithm}

% \vspace{-3mm}

Algorithm \ref{alg:vertex} shows the pseudo-code of 
% how to 
fetching vertices from the hybrid storage engine. We start by scanning from the vertex object $v$, which is either the first vertex of the whole graph or the vertex pointed by the index (line 2). We first retrieve data from the current storage (lines 5-8). We check whether $v$ and its historical unreclaimed versions are visible to the current transaction (line 6). We also check whether they are {\visible} using the function \texttt{TemporalCheck()} based on Equation 1 (line 7). 
Next, we catch data from the historical storage (line 10) using the function \texttt{FetchFromKV()}. We first find the most recent anchor $kv_a$ based on the probe prefix (line 15).
We then seek all previous deltas $kv_d$ that satisfy temporal check (line 19) and assemble $kv_a$ with them to get desired versions (line 18). 

\todo{\textbf{Complexity analysis.}
The scan operator queries versions of a vertex $v$ in a dataset with a total of $n$ vertices. This process consists
of two parts: (1) locating the current version of $v$ and (2) querying the historical versions of $v$.
The complexity of locating the current version depends on the specific 
retrieval mechanisms selected in the current storage,
% of existing conventional graph databases,
such as $log(n)$ for B$^+$-tree index look-up and $n$ for non-index lookup, denoted as $O(\iota(n))$. The complexity of querying the historical versions depends on the chosen approach for introducing temporal
features. In {\tgdb}, we first locate the nearest anchor.
{Since we organize historical versions in the key-value store using SkipList, the complexity of this process is $O(log(A_v))$, where $A_v$ is the average number of anchors for vertices.}
Then, we sequentially scan deltas from the anchor until satisfying the query time condition, with a time complexity of $O(u)$,
where $u$ represents the average length of $u_o$ defined in Equation~\ref{eq:anchor_threshold}.
In conclusion, the scan operator has a complexity of $O(\iota(n)+log(A_v)+u)$.
}

\begin{algorithm}[!t]
\small
 \SetKwInOut{Input}{input}
 \SetKwInOut{Output}{output}
 \caption{Expanding Vertices }
 \label{alg:expend}
 
 \SetKwFunction{FMain}{ExpandVertices }
    \SetKwProg{Fn}{Function}{:}{}
    \Fn{\FMain{$v, C$}}{
    \KwIn {$v$, the graph vertex need to expand; $C$,  temporal condition;}
    \KwOut  {$\Sigma$, the result set;}
      $\Sigma_{ve} \leftarrow $ \texttt{VERead}($v, C$) //get adjacency lists\;
      \ForEach{$(e_{id},nv_{id}) \in\Sigma_{ve}$}{
      $\Sigma_{e} \leftarrow $ \texttt{EdgeRead}($e_{id}, f(C,v.\omega)$)  \;
      \ForEach{$e \in\Sigma_{e}$}{
      $\Sigma_{nv} \leftarrow $ \texttt{VertexRead}($nv_{id}, f(C,e.\omega)$)\;
      \ForEach{$nv \in\Sigma_{nv}$}{
       $\Sigma \leftarrow \Sigma \cup (e,nv)$ \;
       }
      }
    }
    \textbf{return} $\Sigma$\;
}
% \vspace{-1mm}
\end{algorithm}

\subsection{Expand Operator}
\label{sec:expand}
{\tgdb} utilizes the expand operator 
% using the adjacency list index
to fetch linked edge and adjacency vertex versions. The overall design insight of the expand operator is similar to that of the scan operator, which employs different retrieval strategies in two separate storage engines. Additionally, the expand operator considers the retrieval of graph structures. We next elaborate on how the expand operator fetches edge and neighboring vertex versions of a given vertex.

As shown in Algorithm \ref{alg:expend}, given a vertex $v$, we first use the function \texttt{VERead()} to access the $v$'s adjacency list version(s) from the VE part/segment (line 2), which contains a list of (edge id $e_{id}$, neighbor vertex id $nv_{id}$) pairs. The function \texttt{VERead()} shares a similar logic as the function \texttt{VertexRead()} in Algorithm \ref{alg:vertex}, which combines current and historical data to get desired versions. 
% bypass unnecessary historical versions 
% shorten 
We then fetch {\visible} versions of specific semantic information of edges and adjacency vertices based on their unique ids (lines 3-8). 
We first obtain linked edge versions using the \texttt{EdgeRead()} function, which follows a similar logic to \texttt{VertexRead()} (line 4). To expedite the search for the linked edge version, we optimize the skip look-up strategy for anchor locating to skip more unnecessary versions.
Guided by Constraint 1 defined in Section \ref{sec:graph_model}, the edge must be {\visible} for its connected vertices. This implies that the lifespan of the edge version must intersect with the lifespan of its connected vertex version. 
Since we already hold a scanned vertex version $v$, we can leverage $v$'s lifespan $\omega$ to refine the probe time scope $C$ 
based on the following equation. 
\begin{equation}
f(C,\omega)=[max(C.t_1, \omega.st), min(C.t_2, \omega.ed)]
\end{equation} 
By implementing this approach, we efficiently bypass more unnecessary versions, thereby enhancing the query performance. 
Subsequently, we retrieve the neighbor vertex versions of each holding edge version based on $nv_{id}$ (lines 5-6) with a similar logical process. Finally, we obtain the final results (lines 7-8). The illustrative examples of both the scan and expand operators are detailed in our extended manuscript~\cite{aeongsupplement}.
% \begin{equation}
% f(C,\omega)=[max(C.t_1, \omega.st), min(C.t_2, \omega.ed)]
% \end{equation} 

\begin{comment}
\todo{For instance, consider continue to retrieve $v_{2.7}$'s linked edge and adjacency vertex at time $t_n$ in Figure~\ref{fig:overview}.
We first obtain the corresponding version of $v_{2.7}$'s adjacency list, which contains $[(1, 0)]$ (line 2). 
Then, we fetch the corresponding edge version $e_{1.2}$ (line 4) and the neighboring vertex version $v_{0.1}$ (line 6).}
\end{comment}

\todo{\textbf{Complexity analysis.} The expand operator retrieves linked edges and adjacent vertices of a specific vertex version in the following three steps. 
First, we fetch the adjacency list version. Similar to the complexity of the scan operator, the associated complexity is $O(log(A_{ve})+u)$, 
where $A_{ve}$ is the total number of anchors for adjacency lists, and $u$ is the average length of $u_o$. 
Next, for each pair ($e_{id}, nv_{id}$) in the adjacency list, we fetch the corresponding edge version.
Finally, we locate the neighbor vertex version. 
Similar to the first step, the complexities of these two steps are $O(log(A_{e})+u)$ and $O(log(A_{v})+u)$, where $A_{ve}$ and $A_{v}$ are the total number of anchors for edges and vertices, respectively. 
In conclusion, the overall complexity of the expand operator is $O(log(A_{ve})+u+D\times (log(A_{e})+u+log(A_{v})+u))$, where $D$ is the average number of vertex degrees.
\extended{ We would like to highlight that the complexity of both Algorithm~\ref{alg:vertex} and Algorithm~\ref{alg:expend} is at the $O(log(n))$ level, which is reasonable and generally acceptable in graph query processing~\cite{livegraph,g*,Teseo}.}
}

\extended{
\subsection{Proof of the Correctness}
\label{sec:corret}
We now prove that we ensure the transaction to fetch consistent graph data from two isolated storage engines.
% Now, we prove the transactions, either read-write transactions or read-only transactions, are verifiable and correct.

% \vspace{-3mm}
\begin{theorem}
Given a read-only transaction $T$ fetching temporal graph data with a condition time $C$, $T$ can be verified to have retrieved the data from the hybrid storage with a consistent state.
\end{theorem}
% \vspace{-5mm}
\begin{proof}
First, temporal graph data are globally consistent. 
The lifespan $\omega$ of each version is the commit time of the transaction, allocated by the system from a global clock.
% . {\tgdb} allocates each transaction's commit timestamp uniformly from a global sequence counter. 
As a result, the timeline of two isolated storages follows the same serialization order. 
% Consequently, we can ensure global data consistency within the two isolated storages.
Second, we employ a
% a meticulously designed 
fetching strategy to return consistent requested data based on the lifespan $\omega$.
% We fetch the desired version $o$ based on its lifespan $\omega$.
For a version $o$ in the current storage, we verify its visibility to $T$ and then align $o$'s $\omega$ to $C$ (lines 6-8, Algorithm \ref{alg:vertex}). For $o$ in the historical storage, we solely align $o$'s $\omega$ to $C$ (line 10, Algorithm \ref{alg:vertex}).
By so doing, we can guarantee data consistency within each separate storage.
It is possible to find a version $o$ duplicated in both the current storage and the historical storage. Unreclaimed versions in the current storage might be transferred to the historical storage before being cut from the version chain during the asynchronous migration process (Algorithm \ref{alg:migration}). 
% In such a critical situation, it is possible to find $o$ duplicated both in the current storage and the historical storage. 
However, this does not lead to errors, as both versions share the same lifespan $\omega$ and contain identical information. In such cases, we incorporate a deduplication step, retaining only the version from the current storage. 
Considering all data-fetching scenarios, 
% of data fetching, 
we ensure overall data consistency.
\end{proof}

% \todo{

}
\section{Implementation}
\label{sec:implementation}
{\tgdb} is built on Memgraph\cite{Memgraph} and RocksDB\cite{rocksdb}. Memgraph is a commercial native graph database that supports the property graph model, Cypher, and MVCC.
We utilize and extend Memgraph to serve as the primary database engine, providing the basic query engine and current storage engine for {\tgdb}. We then integrate RocksDB, a popular KV store, into Memgraph as the historical storage to manage historical data.
% This integration allows {\tgdb} to effectively manage temporal graph data while taking advantage of the capabilities offered by Memgraph and RocksDB.
% We would like to emphasize that
% our proposed temporal graph management mechanism is generally
% applicable to native graph databases that support MVCC
Our proposed approach 
% temporal graph management mechanism 
is generally applicable to native graph databases that support MVCC.

\textbf{Query engine.}
{\tgdb} extends the \textit{parser} and \textit{executor} components of the query engine in Memgraph to support temporal queries. 
% enhances the \textit{parser} and \textit{executor} of the query engine in Memgraph to support temporal queries. 
{\tgdb} extends the \textit{parser} to recognize temporal queries defined in Section \ref{sec:temporal_query_language}, incorporating the temporal qualifier into \texttt{Cypher.g4} and enhancing \texttt{CypherMainVisitor()} to recognize temporal qualifier.
% {\tgdb} enables the \textit{parser} to recognize temporal queries defined in Section \ref{sec:temporal_query_language},
% . We incorporate the temporal qualifier into \texttt{Cypher.g4} and enhance \texttt{CypherMainVisitor()} to recognize temporal qualifier.
Furthermore, {\tgdb} enhances the \textit{executor} by modifying two fundamental operators: scan and expand operators.
% We implement the \textit{executor} of {\tgdb} by modifying two fundamental operators in Memgraph: the \textit{Scan} operator and the \textit{Expand} operator.
In the \texttt{ScanAllCursor.Pull()} function, besides retrieving current vertices,  we introduce a function \texttt{AddHistoricalVertices()} to capture both unreclaimed and reclaimed historical versions (Algorithm 2). In the \texttt{ExpandCursor.Pull()} function, a similar adaptation is made with the inclusion of the function \texttt{AddHistoricalEdges()} for
% we add a function \texttt{AddHistoricalEdges()} to 
getting historical edges and neighbor vertices (Algorithm \ref{alg:expend}). 

\textbf{Storage engine.} The storage engine 
of {\tgdb} is hybrid, consisting 
% consists 
of a \textit{current} storage and a \textit{historical} storage,
% , with an asynchronous migration, 
as detailed in Section \ref{sec:hybrid_storage}. 
% \textit{Current storage.} 
The \textit{current} storage is derived from Memgraph's storage, where the \texttt{Vertex} structure maps to the vertex object, the \texttt{EdgeRef} structure maps to the edge object, and the \texttt{Delta} structure represents historical unreclaimed versions.
We then associate the time dimension with those structures to introduce temporal support, as discussed in Section \ref{sec:current-storage}. \todo{Timestamps are assigned by a global clock when relevant transactions are committed, with a time granularity of milliseconds.}
{\tgdb} integrates RocksDB into Memgraph as the \textit{historical} store by starting a RocksDB process when the Memgraph instance starts. We introduce a function, \texttt{Migrate()}, within Memgraph's \texttt{CollectGarbage()} function to transfer unreclaimed data to a key-value pair and subsequently migrate them to RocksDB (Algorithm \ref{alg:migration}). 
\todo{We further implement a distributed version of {\tgdb}, named {\tgdbd}, using TiKV~\cite{tikv}, an efficient distributed key-value store, for historical storage by replacing the interfaces of RocksDB with TiKV.
We introduce a system parameter, \texttt{retention\_period}, in {\tgdb} to set the historical data retention period.
For instance, setting \texttt{retention\_period} to one month enables the periodic removal of historical data generated one month ago from the historical storage.
}

\vspace{-2mm}
\section{Evaluation}
\todo{In this section, we first introduce the experimental setup.
We then compare {\tgdb} against two state-of-the-art temporal systems, Clock-G and T-GQL, and provide in-depth performance analyses for {\tgdb}, with two metrics: 1) latency of
temporal queries/graph operations; 2) storage overheads of temporal graph data.}
\label{sec:evaluation}
\extended{\begin{table}
    \centering
    \setlength{\abovecaptionskip}{0.1cm}
    \caption{\todo{Workload Characteristics}}
    \label{tabel:dataset}
    \resizebox{\columnwidth}{!}{
\begin{tabular}{|c|c|c|cccc|}
\hline
\multirow{2}{*}{}                                                                    & \multirow{2}{*}{T-mgBench} & \multirow{2}{*}{T-LDBC} & \multicolumn{4}{c|}{T-gMark}                                                             \\ \cline{4-7} 
                                                                                     &                            &                         & \multicolumn{1}{c|}{Bib}  & \multicolumn{1}{c|}{WD}   & \multicolumn{1}{c|}{LSN}  & SP   \\ \hline
\# of Vertices                                                                       & 10K                        & 3,181K                  & \multicolumn{1}{c|}{100K} & \multicolumn{1}{c|}{103K} & \multicolumn{1}{c|}{100K} & 100K \\ \hline
\# of Edges                                                                          & 122K                       & 17,256K                 & \multicolumn{1}{c|}{121K} & \multicolumn{1}{c|}{93K}  & \multicolumn{1}{c|}{200K} & 385K \\ \hline
Density                                                                              & 12.17                      & 5.42                    & \multicolumn{1}{c|}{1.2}  & \multicolumn{1}{c|}{0.90} & \multicolumn{1}{c|}{2}    & 3.85 \\ \hline
\# of Vertex Labels                                                                  & 1                          & 8                       & \multicolumn{1}{c|}{5}    & \multicolumn{1}{c|}{24}   & \multicolumn{1}{c|}{15}   & 7    \\ \hline
\# of Edge Labels                                                                    & 1                          & 25                      & \multicolumn{1}{c|}{4}    & \multicolumn{1}{c|}{82}   & \multicolumn{1}{c|}{27}   & 7    \\ \hline
\begin{tabular}[c]{@{}c@{}}\# of Graph Operations\\ for Data Generation\end{tabular} & 320K                       & 1M                      & \multicolumn{1}{c|}{320K} & \multicolumn{1}{c|}{320K} & \multicolumn{1}{c|}{320K} & 320K \\ \hline
\end{tabular}
}
\end{table}
}
% In this section, we first introduce the experimental setup, and then evaluate {\tgdb} against two state-of-the-art temporal graph systems in a range of settings to show its performance.

\subsection{Experimental Setup}
\label{sec:experimental_setup}
% We next evaluate our approach in varies configurations. Our approach
% is implemented inside Memgraph and we call this version of Memgraph as AeonG. 
{\tgdb} is built on Memgraph $v2.2.0$, RocksDB $v6.14.6$, and TiKV $v7.1.2$ for evaluation. 
% By default, {\tgdb} employs RocksDB as the historical storage without implementing historical data retention cleanup. 
We run the experiments in a cluster of up to 5 nodes.
Each node is equipped with 32 Intel(R) Xeon(R) Gold 5220 CPU @ 2.20GHz, 128 GB memory, running CentOS 7.9.
% .2009 operation system with kernel version 3.10.0-1127.19.1.el7.x86\_64. 

\subsubsection{Baseline Systems}
We compare {\tgdb} with two baseline systems that support temporal features:
% T-GQL at the application level and Clock-G at the system level.
% that periodically \todo{materializes} the snapshots
% of the entire graph and \todo{logs} the deltas between two successive
% snapshots (abbreviated as the \textit{Copy+Log} method).
% periodically copies 
% snapshots of the entire database and simultaneously logs differences between successive snapshots.

\noindent 
\textbf{T-GQL}~\cite{T-GQL}: 
A state-of-the-art graph database that assigns a time period to each graph object (vertex or edge) for temporal support.
Since we cannot obtain the source code of T-GQL, for fair comparisons, we implement T-GQL based on Memgraph.
Note that T-GQL stores all the vertices and edges in memory.
% its buffer structure into our prototype system.
% T-GQL adopts a specific representation of temporal graphs, where conventional vertices are decomposed into 
% Object, Attribute, Value vertices, and conventional edges remain the same. 
% T-GQL introduces the time dimension as properties of designed vertices and edges. 
% To ensure fair comparisons, we implemented T-GQL on Memgraph.

\noindent 
{\textbf{Clock-G}}~\cite{clock-g}: 
A state-of-the-art graph storage engine that manages temporal data by periodically creating snapshots of the entire databases.
Since Clock-G is not open-sourced, we implement its temporal data management approach into our codebase for fair comparisons.
We record both the snapshots and the logs between successive snapshots in RocksDB. 
We further introduce a query engine like that used in {\tgdb} to support temporal queries.
By default, Clock-G creates a snapshot after executing 80k graph operations.

\maintext{
\begin{table}
    \centering
    \setlength{\abovecaptionskip}{0.1cm}
    \caption{\todo{Workload Characteristics}}
    \label{tabel:dataset}
    \resizebox{\columnwidth}{!}{
\begin{tabular}{|c|c|c|cccc|}
\hline
\multirow{2}{*}{}                                                                    & \multirow{2}{*}{T-mgBench} & \multirow{2}{*}{T-LDBC} & \multicolumn{4}{c|}{T-gMark}                                                             \\ \cline{4-7} 
                                                                                     &                            &                         & \multicolumn{1}{c|}{Bib}  & \multicolumn{1}{c|}{WD}   & \multicolumn{1}{c|}{LSN}  & SP   \\ \hline
\# of Vertices                                                                       & 10K                        & 3,181K                  & \multicolumn{1}{c|}{100K} & \multicolumn{1}{c|}{103K} & \multicolumn{1}{c|}{100K} & 100K \\ \hline
\# of Edges                                                                          & 122K                       & 17,256K                 & \multicolumn{1}{c|}{121K} & \multicolumn{1}{c|}{93K}  & \multicolumn{1}{c|}{200K} & 385K \\ \hline
Density                                                                              & 12.17                      & 5.42                    & \multicolumn{1}{c|}{1.2}  & \multicolumn{1}{c|}{0.90} & \multicolumn{1}{c|}{2}    & 3.85 \\ \hline
\# of Vertex Labels                                                                  & 1                          & 8                       & \multicolumn{1}{c|}{5}    & \multicolumn{1}{c|}{24}   & \multicolumn{1}{c|}{15}   & 7    \\ \hline
\# of Edge Labels                                                                    & 1                          & 25                      & \multicolumn{1}{c|}{4}    & \multicolumn{1}{c|}{82}   & \multicolumn{1}{c|}{27}   & 7    \\ \hline
\begin{tabular}[c]{@{}c@{}}\# of Graph Operations\\ for Data Generation\end{tabular} & 320K                       & 1M                      & \multicolumn{1}{c|}{320K} & \multicolumn{1}{c|}{320K} & \multicolumn{1}{c|}{320K} & 320K \\ \hline
\end{tabular}
}
\end{table}
}

\subsubsection{Workloads.}
\label{sec:evaluaton_workload}
% We conduct the experiments using the following workloads.

We conduct the experiments using three temporal-enhanced workloads, characteristics of which are detailed in Table~\ref{tabel:dataset}. We next describe each of these three workloads.
% extended from one real-world workload, mgBench~\cite{mgbench}, and two synthetic workloads, LDBC~\cite{LDBC}, and gMark~\cite{gmark}, respectively.
% , T-mgBench, T-LDBC, and T-gMark, 
% These foundational workloads are selected due to their wide use in assessing graph database performance~\cite{Memgraph,kuzu,graindb,livegraph,DBLP:conf/sigmod/JachietGGL20}.

\noindent \textbf{T-mgBench} is based on the real-world Pokec dataset~\cite{pokec}, which is used in Memgraph's mgBench~\cite{mgbench} workload.
% Given that {\tgdb} is built upon Memgraph, we therefore include this workload in our experiments.
% T-mgBench consists of 10k vertices, 122k edges, with each vertex or edge having 1 label.
% We conduct 320k graph operations to generate historical data. 
As outlined in Table \ref{table:temporal_queries}, T-mgBench includes four temporal queries, by extending the non-temporal queries in mgBench with temporal dimension.
Specifically, we add ``FOR TT AS OF $t$'' to Q1 and Q3, forming ``time-point'' queries, and add ``FOR TT FROM $t_1$ to $t_2$'' to Q2 and Q4, forming ``time-slice'' queries.
% we introduce four temporal queries, as detailed in Table \ref{table:temporal_queries}: Q1 and Q3 are `time-travel' queries, while Q2 and Q4 are `time-slice' queries.

\noindent \textbf{T-LDBC} derives from LDBC~\cite{LDBC}, a well-known synthetic graph workload, by incorporating ``FOR TT AS OF $t$'' into the LDBC IS queries (IS1-IS7).
\maintext{The full queries of T-LDBC are listed in our {extended manuscript~\cite{aeongsupplement}}.}
\extended{More specifics of these queries can be found in Table I in our Appendix.}
% We use the dataset LDBC with a scale factor of 1.
% consisting of 3.18M vertices, 17.26M edges, 8 vertex labels, and 25 edge labels. 
% We conduct one million graph operations by default to generate historical data.
% Regarding temporal queries, we employ the IS queries introduced in the LDBC workload by associating ``FOR TT AS OF'' in the Match clause.
% Due to space constraints, the detailed queries can be found in our extended version~\cite{}.

\noindent \textbf{T-gMark} is based on gMark\cite{gmark}, a well-known synthetic graph workload.
It consists of four datasets as shown in Tables~\ref{tabel:dataset}.
We use gMark's query generation tool to create non-temporal queries, and then transform them into temporal queries by adding the time condition ``FOR TT AS OF $t$''. 
The query generation follows gMark's default configuration, which includes constraints on arity (0-4), query shape (25\% chain, 25\% star-chain, 25\% cycle, and 25\% star), selectivity (33\% constant, 33\% linear, and 33\% quadratic), probability recursion (50\%), and query size ([1,1], [3,4], [1,3], [2,4]).
% a schema-driven graph and query generation focusing on structural properties. 
% T-gMark includes four datasets, namely Bib, SP, LSN, and WD, generated by gMark's default graph generation with fixed 100k vertices. 
% Specifically, Bib comprises 100k vertices, 121k edges, 5 vertex labels, and 4 edge labels; SP includes 103k vertices, 93k edges, 24 vertex labels, and 82 edge labels; LSN encompasses 100k vertices, 200k edges, 15 vertex labels, and 27 edge labels; WD includes 118k vertices, 385k edges, 7 vertex labels, and 7 edge labels.
% We inject 320k graph operations for each dataset to generate historical data. 
% For temporal queries, we utilize gMark's query generation to generate non-temporal queries and then enhance those common queries with temporal features by associating with the time condition ``FOR TT AS OF $t$''.
% For our evaluation, we adhere to the default query configuration of gMark to generate diverse queries, encompassing arity constraints (0-4), shape constraints (25\% chain, 25\% star-chain, 25\% cycle, and 25\% star), selectivity constraints (33\% constant, 33\% linear, and 33\% quadratic), and query size constraints ([1,1],[3,4],[1,3],[2,4]).

To effectively evaluate the efficiency of temporal features in {\tgdb} and baseline systems, for each workload, we generate additional historical data before evaluations.
% extended the conducted two following extensions to set up these workloads.
% First, we 
% enhanced the data generation process of each workload.
Unless otherwise specified, we first use the data generation tools from the original workload to create the initial dataset, and then execute graph operations with a mix of 80\% updates, 10\% creates, and 10\% deletes to generate historical data.
The access distribution of update operations and queries follows the Zipf{~\cite{Zipf49}} distribution to simulate real-world graph manipulation scenarios. 
% Since different systems potentially have different graph operation performances, we control to execute one graph operation per second, 
% Second, we introduced representative temporal queries by integrating time conditions into their provided non-temporal queries.

\begin{table}[]
\setlength{\abovecaptionskip}{0.1cm}
\caption{Temporal Queries of T-mgBench}
\resizebox{\columnwidth}{!}{
\renewcommand\arraystretch{1.2}{
\begin{tabular}{|c|c|}
\hline
Query     & Statement               \\ \hline
Q1      & Match (n: User \{id: \textit{\$id}\}) FOR TT AS OF $t$ RETURN n                        \\ \hline
Q2      & Match (n: User \{id: \textit{\$id}\}) FOR TT From $t_1$ to $t_2$ RETURN n                        \\ \hline
Q3      & Match (n: User \{id: \textit{\$id}\})-[e]->(m) FOR TT AS OF $t$ RETURN n,e,m     \\ \hline
Q4      & Match (n: User \{id: \textit{\$id}\})-[e]->(m) FOR TT From $t_1$ to $t_2$ RETURN n,e,m  \\ \hline
\end{tabular}
}
}
\label{table:temporal_queries}
\end{table}

\subsubsection{Default Configuration}
By default, 
% we perform 320k graph operations for data generation. 
the Zipf distribution factor is set to 1.1.
% The Zipf distribution factor is set to 1.1. 
The parameters of the adaptive anchoring approach defined in Equation \ref{eq:anchor_threshold} are configured as $\tau_1=1k$, $\tau_2=10k$ and $c=1\%$.
% default settings: 
% We set
The \texttt{retention\_period} discussed in Section \ref{sec:implementation} is set to 0, indicating that all historical data will be retained permanently. 

\todo{
\subsection{{\tgdb} vs Baseline Systems}
\label{sec:performance_comparison}
We now compare {\tgdb} with 
% first report the performance evaluation of {\tgdb} in comparison to 
two baseline systems, Clock-G and T-GQL. 
% Our evaluation measures storage consumption, the write performance of graph operations, and the read performance of temporal queries with a specific focus on accessing historical data. 
%
% \makesure{In the following experiments, we close the historical data retention cleanup in {\tgdb} and apply the adaptive anchoring configuration with the default configuration of $\tau_1=1k$, $\tau_2=10k$ and $c=1\%$.} 
As T-GQL is an in-memory database, to ensure fair comparisons, we make sure all data is cached in memory for {\tgdb} and Clock-G by configuring RocksDB's MemTable size to 640MB.
}

\subsubsection{Experiment on the storage consumption.} 
\label{sec:storage}
% We first present the performance comparison on the T-mgBench workload. 
% In this experiment, we inject graph operations ranging from 80k to 400k.
% \textbf{Storage consumption with varying graph operations.} 
We first conduct experiments using T-mgBench, and plot the storage consumption of each system under different numbers of graph operations in Figure \ref{Fig.tpokec.storage}.
As observed,  {\tgdb} reduces the storage consumption by up to {$2.4\times$} compared to Clock-G and {$2.09\times$} compared to T-GQL, with the increased number of graph operations. 
The lowest storage consumption of {\tgdb} can be attributed to our ``anchor+delta'' strategy, which compactly stores most historical graph data as deltas, minimizing the storage consumption for maintaining temporal data.
In contrast, T-GQL's graph model prevents it from storing historical graph data compactly, while Clock-G periodically creates historical snapshots of the entire graph. 
Both systems incur higher storage overheads than {\tgdb}.

We also run experiments using T-LDBC and T-gMrak to evaluate the storage consumption of these systems.
As shown {in} Figure \ref{Fig.ldbc.storage_graph_op}, {\tgdb} exhibits up to {$5.73\times$} and {$3.59\times$} lower storage consumption compared to Clock-G and T-GQL, under T-LDBC. 
% In the T-gMark workload, {\tgdb} maintains its superiority across diverse datasets.} 
\todo{Further, as observed in Figure \ref{Fig.gmark.storage}, the storage consumption of {\tgdb} is lower than that of Clock-G and T-GQL by up to a {$4.34\times$} and a {$2.39\times$}, under T-gMrak. 
% This result is consistent with that observed in Figure \ref{Fig.tpokec.storage} and Figure \ref{Fig.ldbc.storage_graph_op}. 
This trend is consistent with that observed in Figure \ref{Fig.tpokec.storage}, demonstrating that {\tgdb} still achieves lower storage overhead when handling large and complex graph workloads.}

\extended{
\begin{figure}[t]
  \centering
    %graph op
    \subfigure[T-mgBench: Storage Consumption]{
    \begin{minipage}[t]{0.5\linewidth}
     \label{Fig.tpokec.storage}
    \includegraphics[width=1\textwidth]{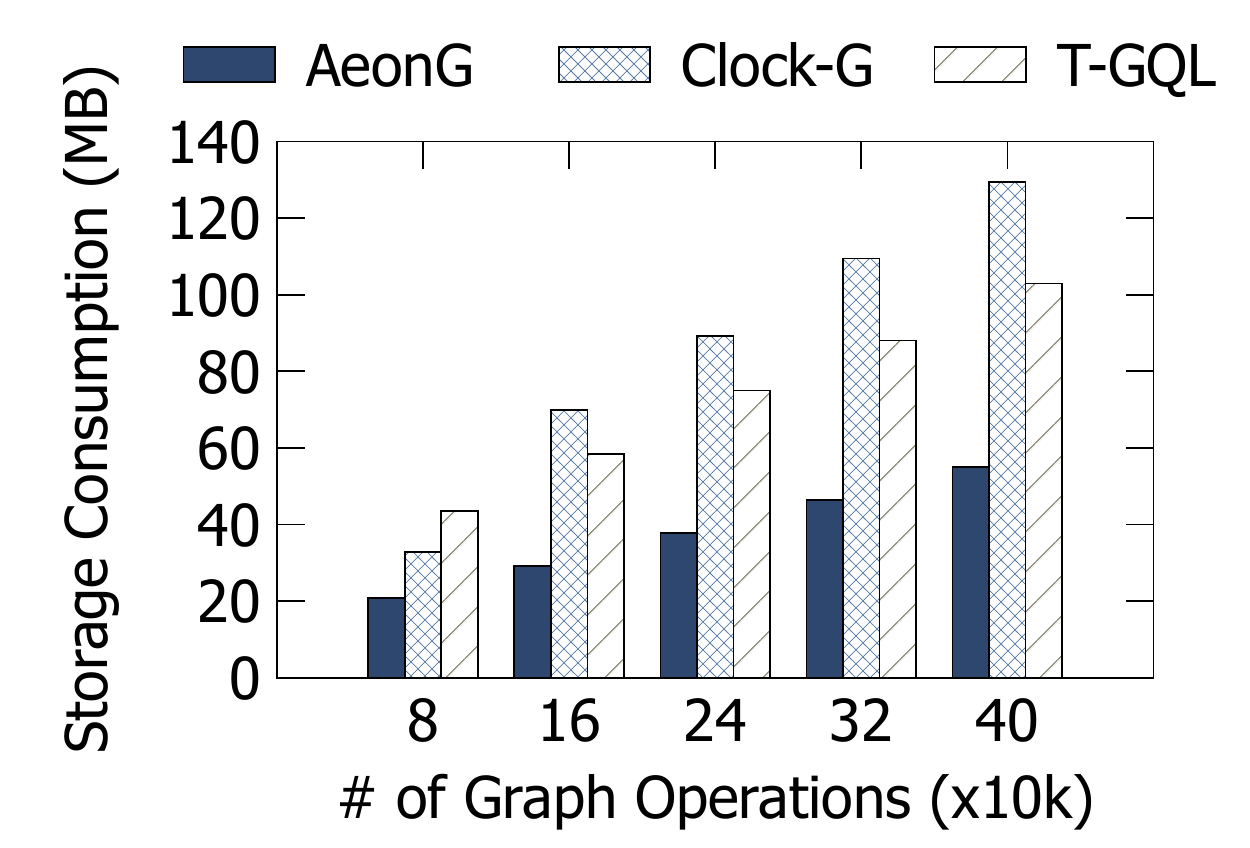}
    \end{minipage}%
   }%
   \subfigure[T-mgBench: Graph Operation Latency]{%varying $n$ (DBpedia)
    \begin{minipage}[t]{0.5\linewidth}
     \label{Fig.tpokec.graph_op}
    \includegraphics[width=1\textwidth]{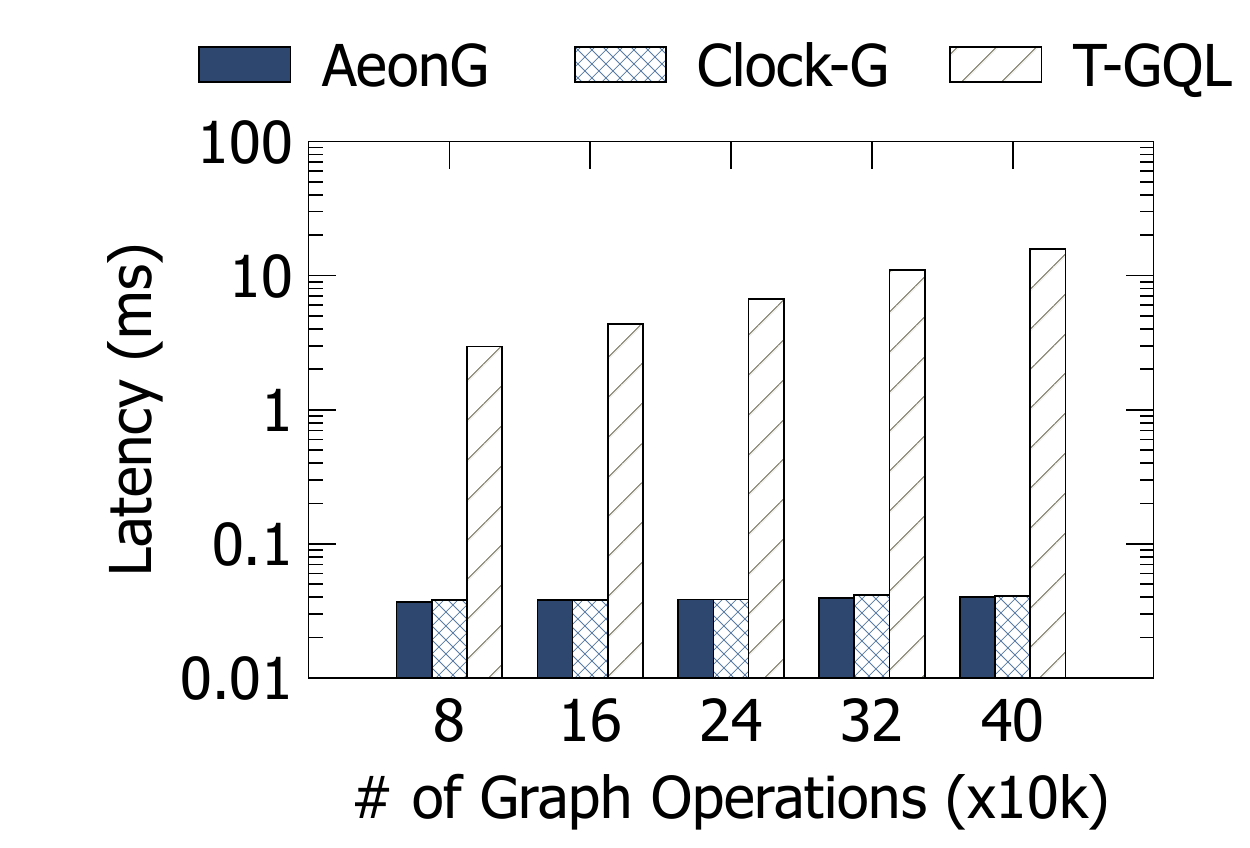}
    \end{minipage}%
   }%
    \vspace{-3mm}
     \subfigure[T-LDBC: Storage Consumption]{%varying $n$ (DBpedia)
    \begin{minipage}[t]{0.5\linewidth}
     \label{Fig.ldbc.storage_graph_op}
\includegraphics[width=1\textwidth]{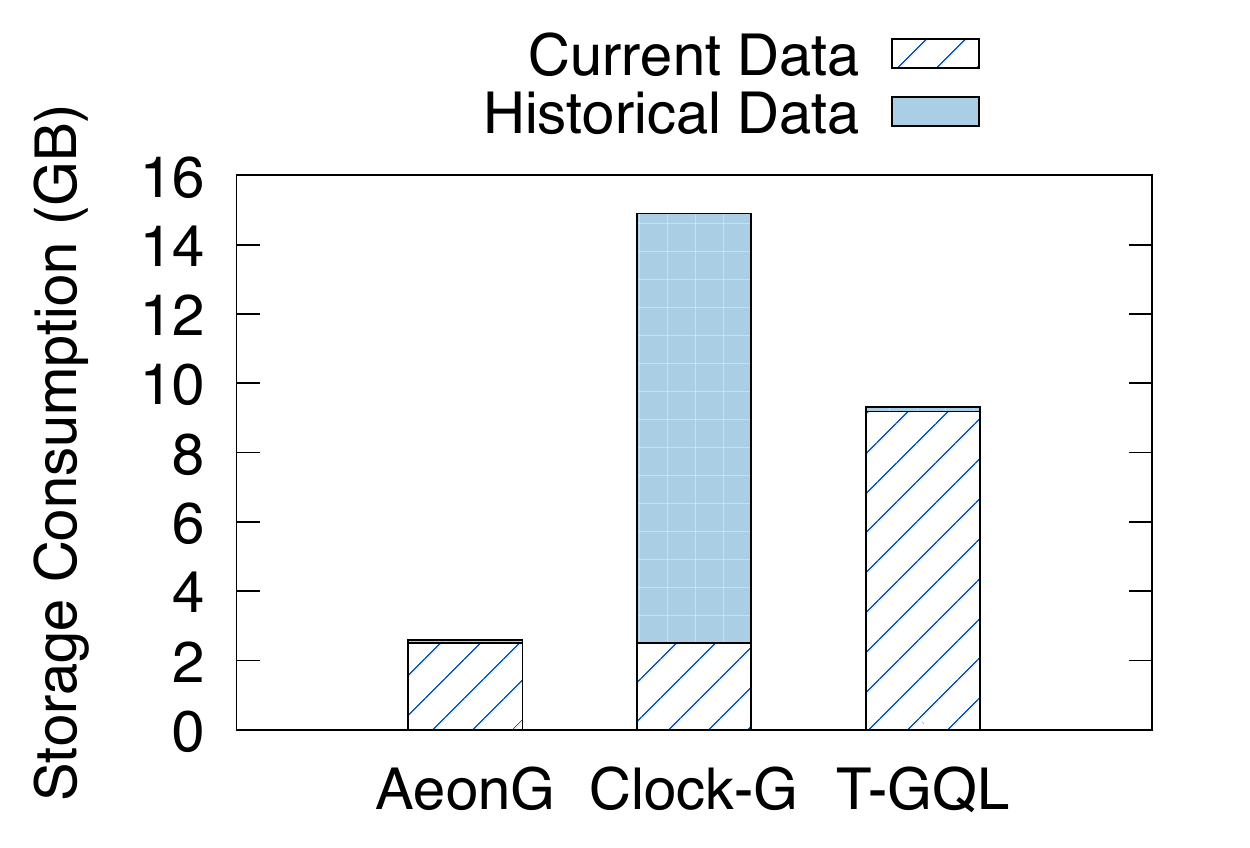}
    \end{minipage}%
    }%
     \subfigure[T-LDBC: Graph Operation Latency]{%varying $n$ (DBpedia)
    \begin{minipage}[t]{0.5\linewidth}
     \label{Fig.ldbc.storage_graph_op1}
\includegraphics[width=1\textwidth]{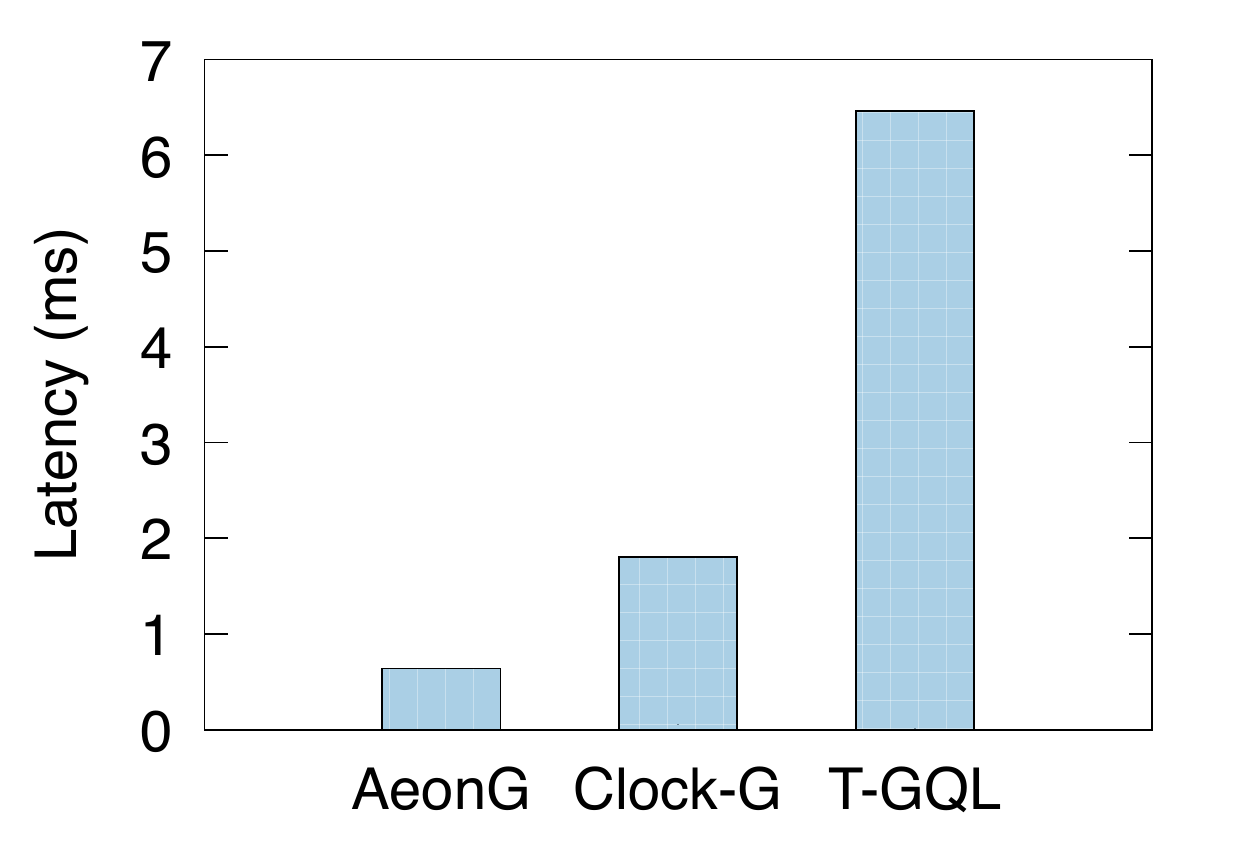}
    \end{minipage}%
    }%
    \vspace{-3mm}
    \subfigure[T-gMark: Storage Consumption]{%varying $n$ (DBpedia)
    \begin{minipage}[t]{0.5\linewidth}
     \label{Fig.gmark.storage}
    \includegraphics[width=1\textwidth]{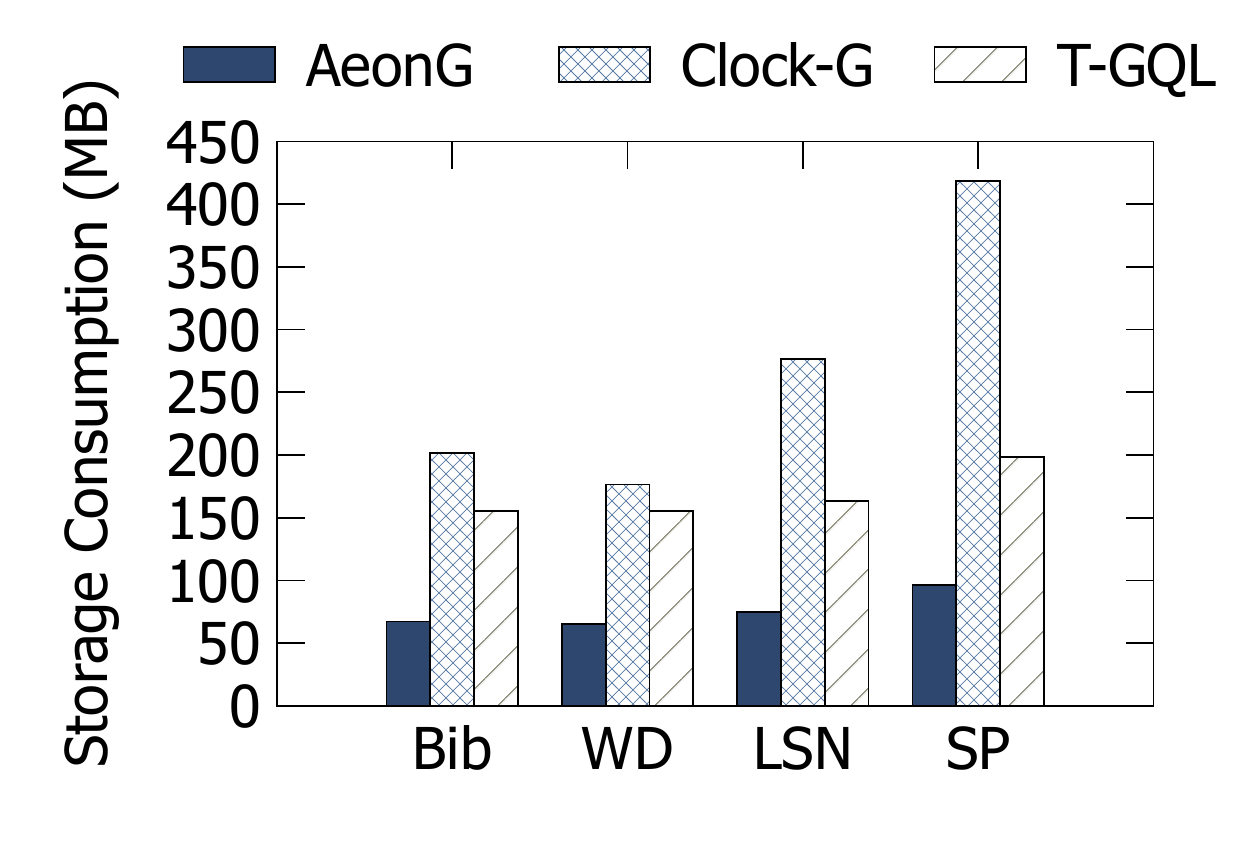}
    \end{minipage}%
   }%
   \subfigure[T-gMark: Graph Operation Latency]{%varying $n$ (DBpedia)
    \begin{minipage}[t]{0.5\linewidth}
     \label{Fig.gmark.graph_op}
    \includegraphics[width=1\textwidth]{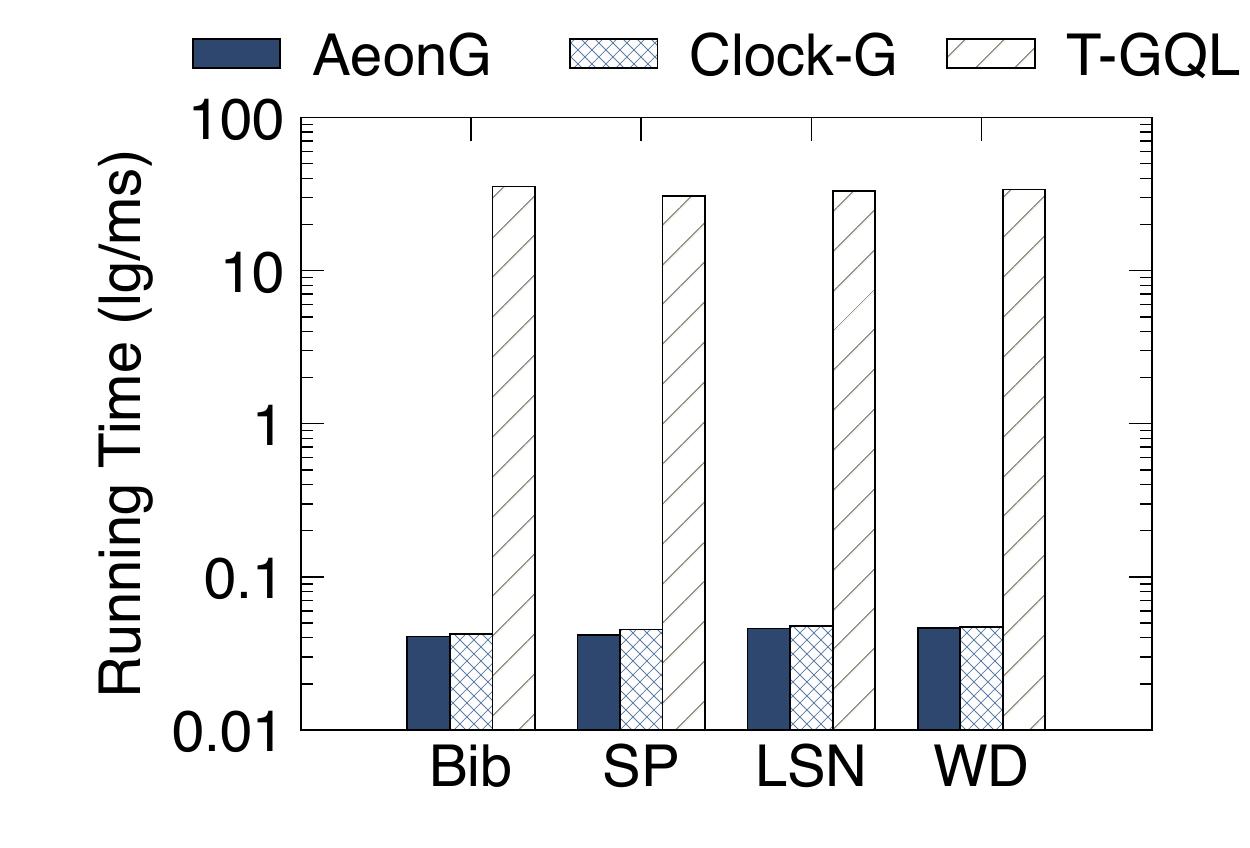}
    \end{minipage}%
   }%
  \vspace{-4mm}
  % \caption{Performance Comparison on the T-mgBench workload}
\captionsetup{justification=raggedright}
\caption{Comparisons on Storage Consumption and Graph Operation Latency}
% \vspace{-2mm}
  \label{Fig.overall}
  
\end{figure}
}

\maintext{
\begin{figure}[t]
  \centering
    %graph op
    \subfigure[T-mgBench: Storage Consumption]{
    \begin{minipage}[t]{0.5\linewidth}
     \label{Fig.tpokec.storage}
    \includegraphics[width=1\textwidth]{figures/results/0_tmgbench/storage.pdf}
    \end{minipage}%
   }%
   \subfigure[\todo{T-mgBench: Graph Operation Latency}]{%varying $n$ (DBpedia)
    \begin{minipage}[t]{0.5\linewidth}
     \label{Fig.tpokec.graph_op}
    \includegraphics[width=1\textwidth]{figures/results/0_tmgbench/graph_op.pdf}
    \end{minipage}%
   }%
    \vspace{-3mm}
     \subfigure[\todo{T-LDBC: Storage Consumption \&  Graph Operation Latency}]{%varying $n$ (DBpedia)
    \begin{minipage}[t]{0.5\linewidth}
     \label{Fig.ldbc.storage_graph_op}
\includegraphics[width=1\textwidth]{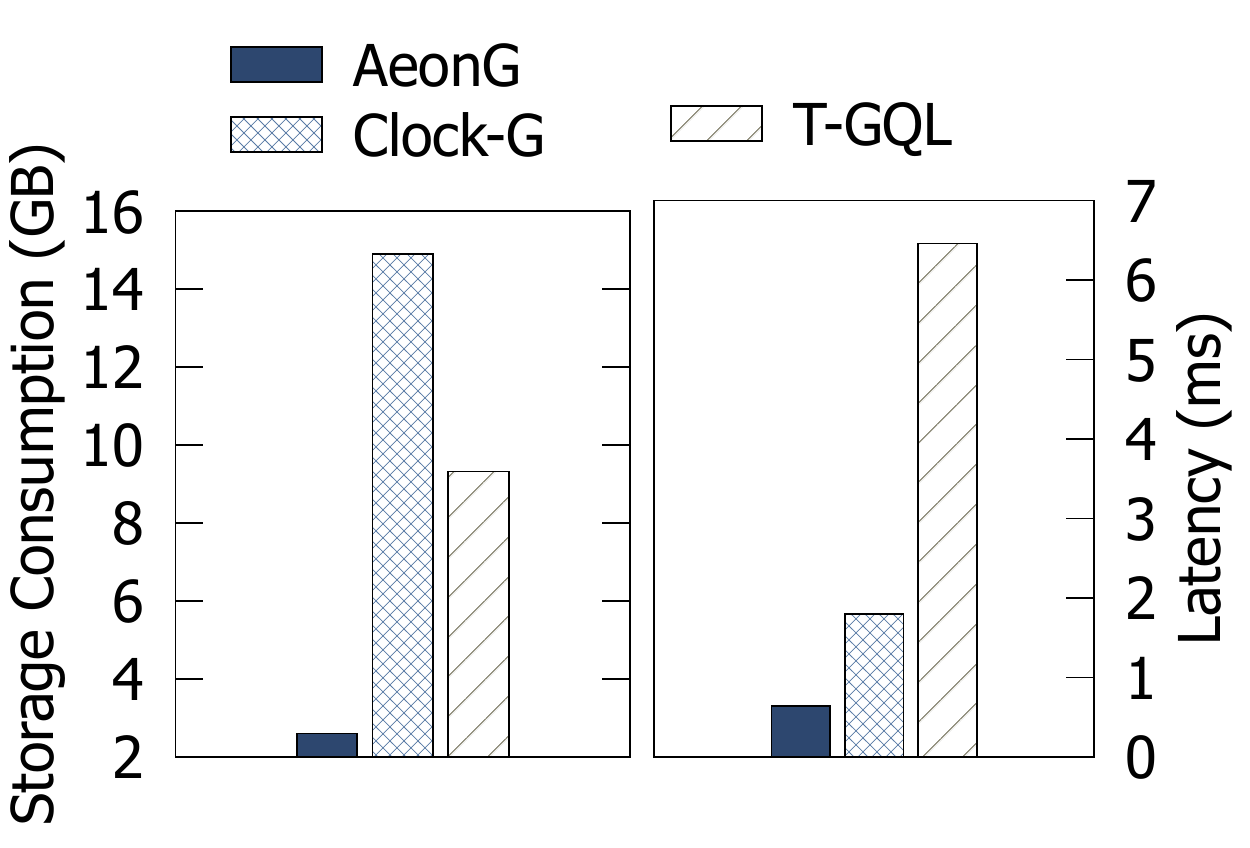}
    \end{minipage}%
    }%
    \subfigure[\todo{T-gMark: Storage Consumption}]{%varying $n$ (DBpedia)
    \begin{minipage}[t]{0.5\linewidth}
     \label{Fig.gmark.storage}
    \includegraphics[width=1\textwidth]{figures/results/2_gmark/storage.pdf}
    \end{minipage}%
   }%   
  \vspace{-2mm}
  % \caption{Performance Comparison on the T-mgBench workload}
\captionsetup{justification=raggedright}
\caption{Comparisons on Storage Consumption and Graph Operation Latency}
\label{Fig.overall}
  
\end{figure}
}
\todo{
\subsubsection{Experiments on the graph operation latency.}
\label{sec:graph_op}
We then evaluate the graph operation latency with varying numbers of graph operations using T-mgBench.
% \marginpar{R1.D6}
% Figure \ref{Fig.tpokec.graph_op} shows the latency of graph operations with varying scales. 
% We next evaluate the performance of graph operations.
% \todo{we need to describe the figure}
As shown in Figure \ref{Fig.tpokec.graph_op}, 
{\tgdb} performs similarly to Clock-G, but significantly outperforms T-GQL by up to {$397.06\times$}. 
% is lower than Clock-G and T-GQL by an average of {$2.24\times$} and {$521.4\times$}, respectively. 
% As the size of graph operations increases, {\tgdb} exhibits minimal performance degradation, with a slight increase of {$1.05$}. In contrast, Clock-G experiences a relatively higher degradation of {$2.91$}, and T-GQL encounters a substantial degradation of {$17.6$}. 
As the number of graph operations grows, both {\tgdb} and Clock-G exhibit a performance degradation of {$5.4\times$} from 80k to 400k, whereas T-GQL shows a much larger degradation of {$34.95\times$}.
The performance difference is mainly due to that T-GQL does not separate the storage of current and historical data.
Therefore, graph operations, such as updates, require traversing through a larger number of graph objects (both current and historical data) to reach the specific graph object for updating. 
However, {\tgdb} separate current and historical data, leading to much smaller latency overheads.
The similar performance of {\tgdb} and Clock-G in T-mgBench can be explained as the overhead from snapshot creation is not significant when handling a relatively small size graph.
% but outperforms Clock-G by a large extent on T-LDBC (shown in Figure\ref{Fig.ldbc.storage_graph_op}), which has a much larger dataset than T-mgBench. 
% The reason will be described in detail later in Section \ref{sec:tldbc_tgmark}.
Therefore, to further evaluate storage operation latency with larger graphs, we conducted additional experiments using the T-LDBC workload.
As observed in \maintext{Figure\ref{Fig.ldbc.storage_graph_op}}
\extended{Figure \ref{Fig.ldbc.storage_graph_op1}}, the graph operation latency of {\tgdb} is lower than Clock-G and T-GQL by up to {$2.82\times$} and {$10.11\times$}, respectively. 
% Clock-G exhibits performance degradation compared to the results observed in Figure \ref{Fig.tpokec.graph_op}. 
As T-LDBC is more substantial, Clock-G requires extra CPU and IO resources to periodically create large historical snapshots, which can negatively affect graph operation performance due to resource contention.
\maintext{We also report the graph operation latency under T-gMark in the extended version [1], where the observations are similar to those reported.}
\extended{As shown in Figure \ref{Fig.gmark.graph_op}, {\tgdb} performs similarly to Clock-G and significantly outperforms T-GQL by up to $766\times$. This trend is consistent with the one observed in Figure \ref{Fig.tpokec.graph_op}.}

}

\maintext{
\begin{figure}[]
  \centering
  \subfigure[T-mgBench]{
    \begin{minipage}[t]{0.5\linewidth}
     \label{Fig.tpokec.various_queries}
    \includegraphics[width=1\textwidth]{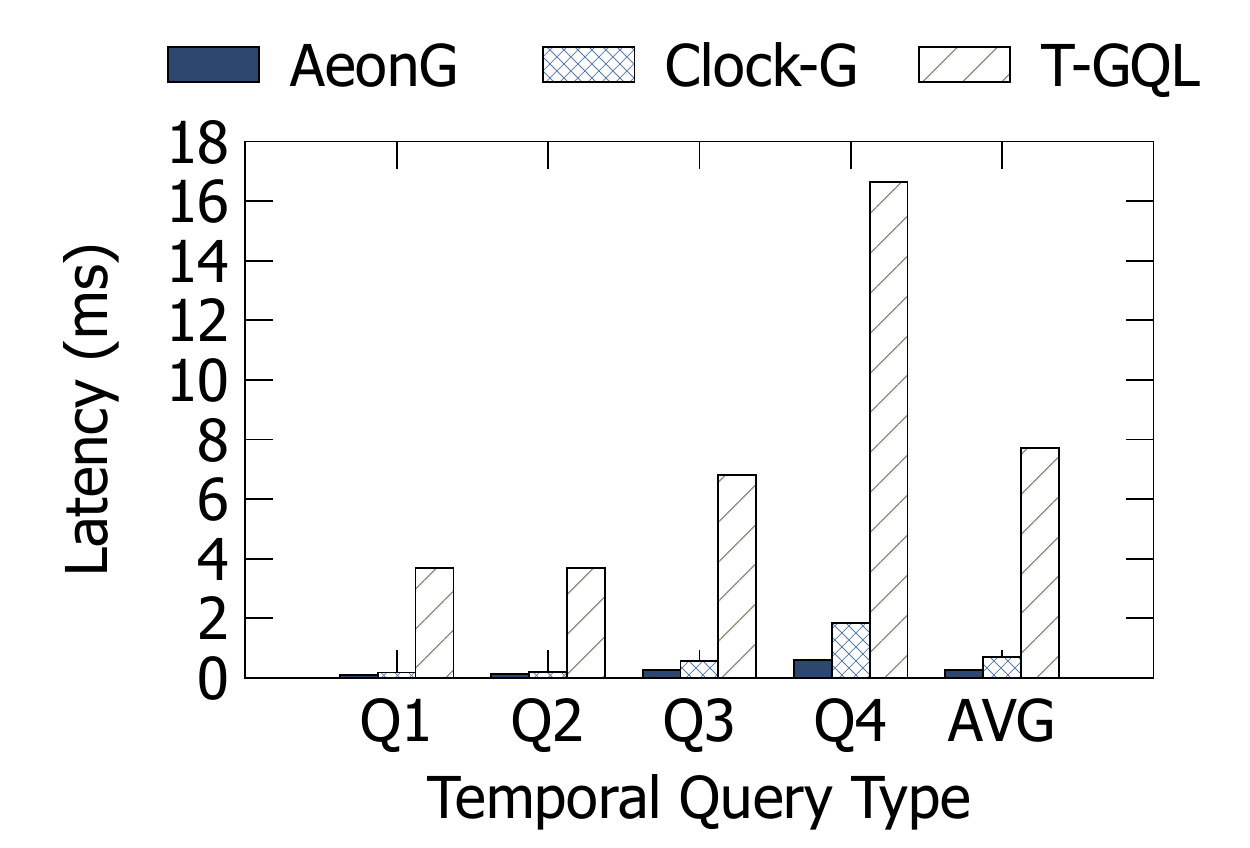}
    \end{minipage}%
   }%
    \subfigure[T-mgBench (Q1): \# of Graph Operations]{
    \begin{minipage}[t]{0.5\linewidth}
     \label{Fig.tpokec.q1_time_scale}
    \includegraphics[width=1\textwidth]{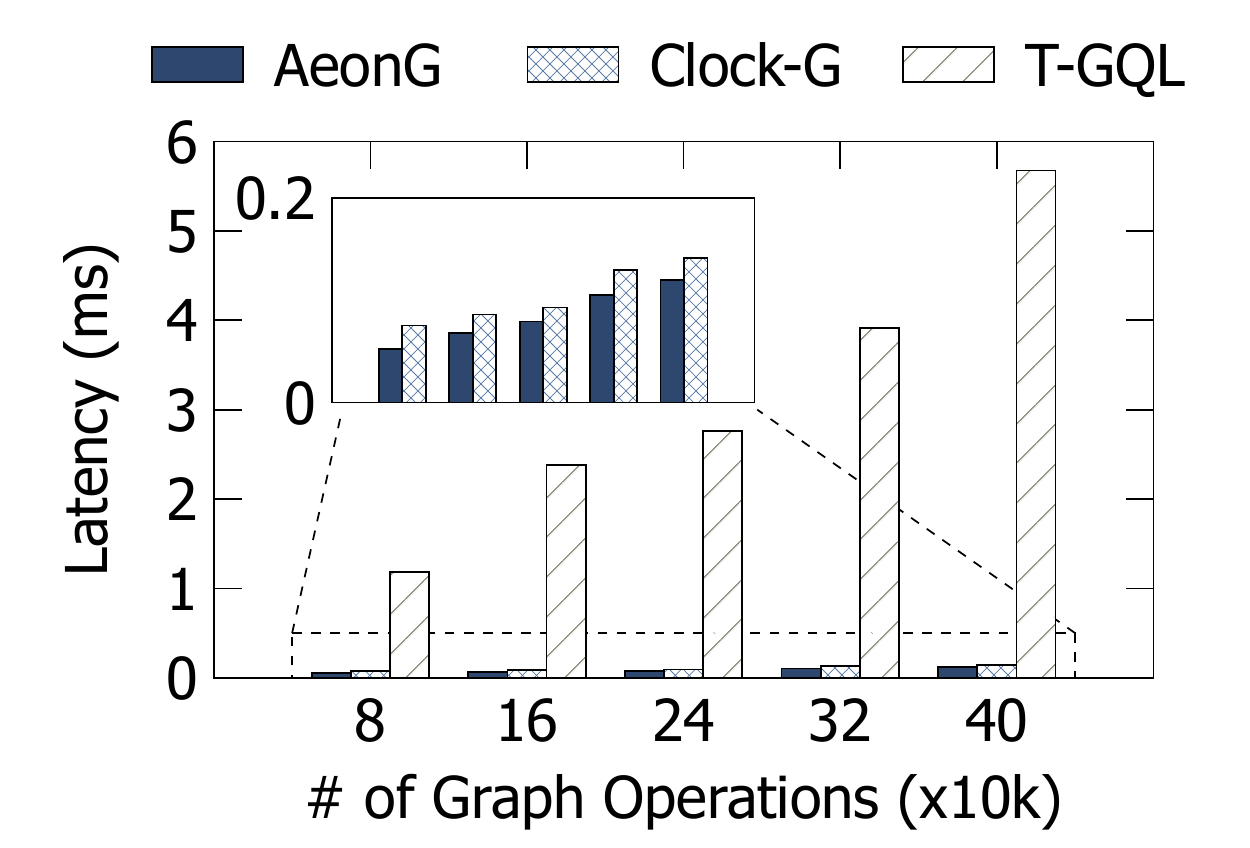}
    \end{minipage}%
   }%
    \vspace{-3mm}
    \subfigure[T-mgBench (Q3): Accessed Vertex Types]{
    \begin{minipage}[t]{0.5\linewidth}
     \label{Fig.tpokec.q3_frequency}
    \includegraphics[width=1\textwidth]{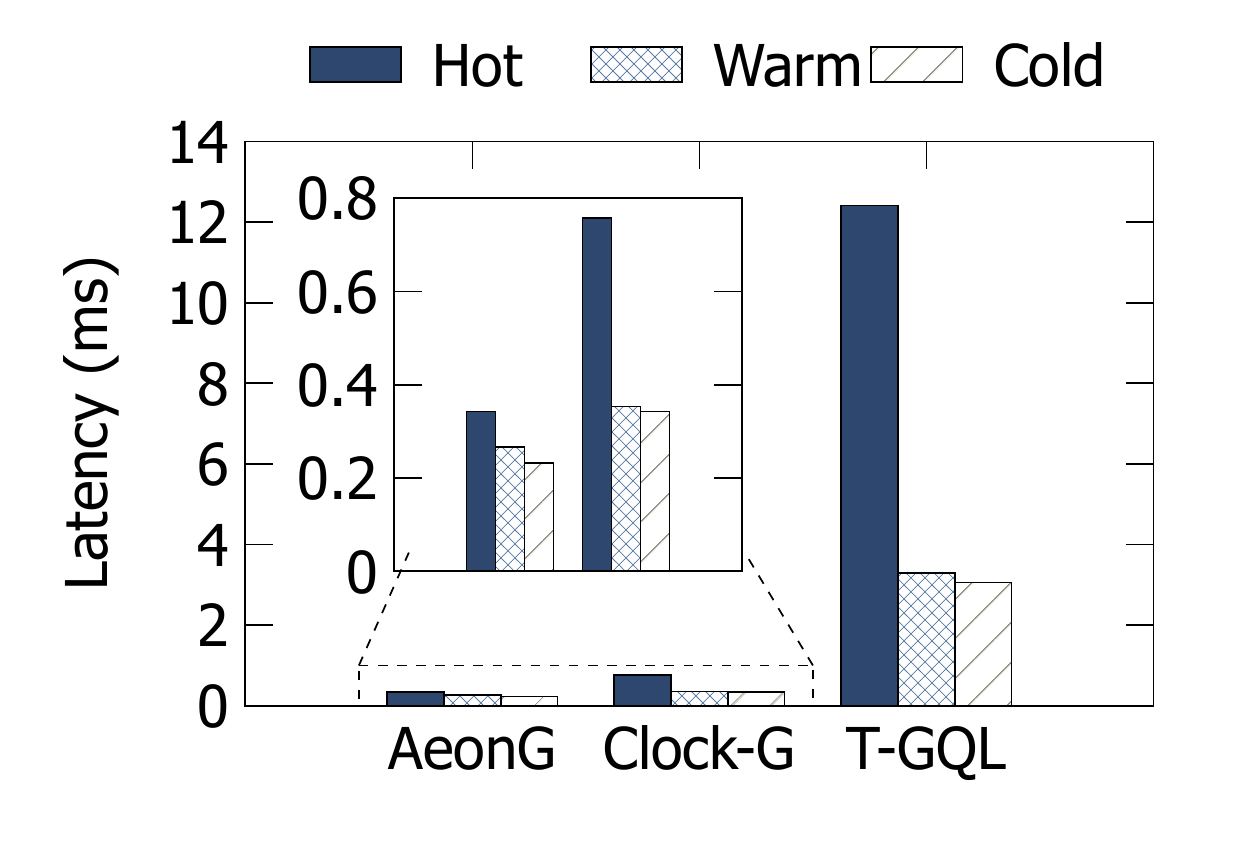}
    \end{minipage}%
   }%
    \subfigure[T-mgBench (Q4): Time Slice]{
    \begin{minipage}[t]{0.5\linewidth}
     \label{Fig.tpokec.q4_timeslice}
    \includegraphics[width=1\textwidth]{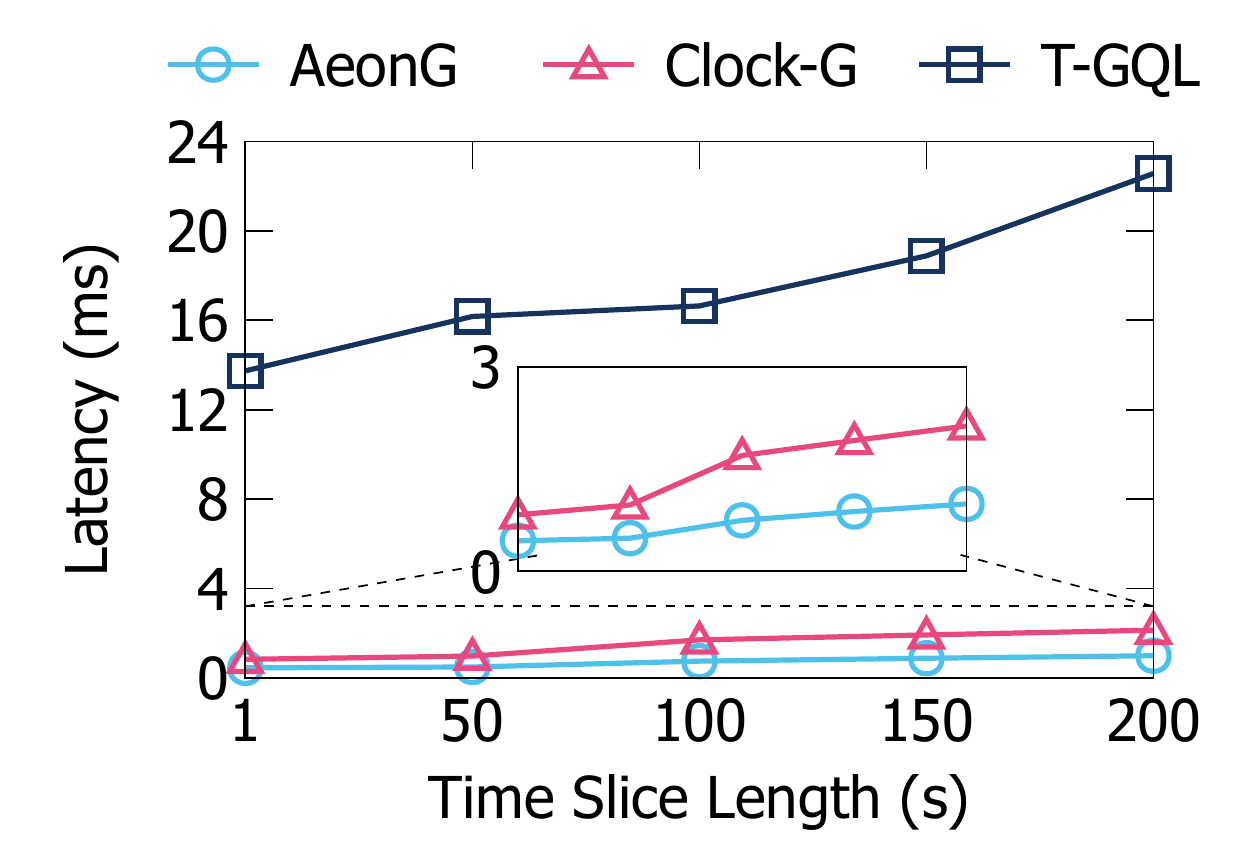}
    \end{minipage}%
   }%
   \vspace{-3mm}
    \subfigure[T-LDBC]{%varying $n$ (DBpedia)
    \begin{minipage}[t]{0.5\linewidth}
     \label{Fig.ldbc.various_q}
    \includegraphics[width=1\textwidth]{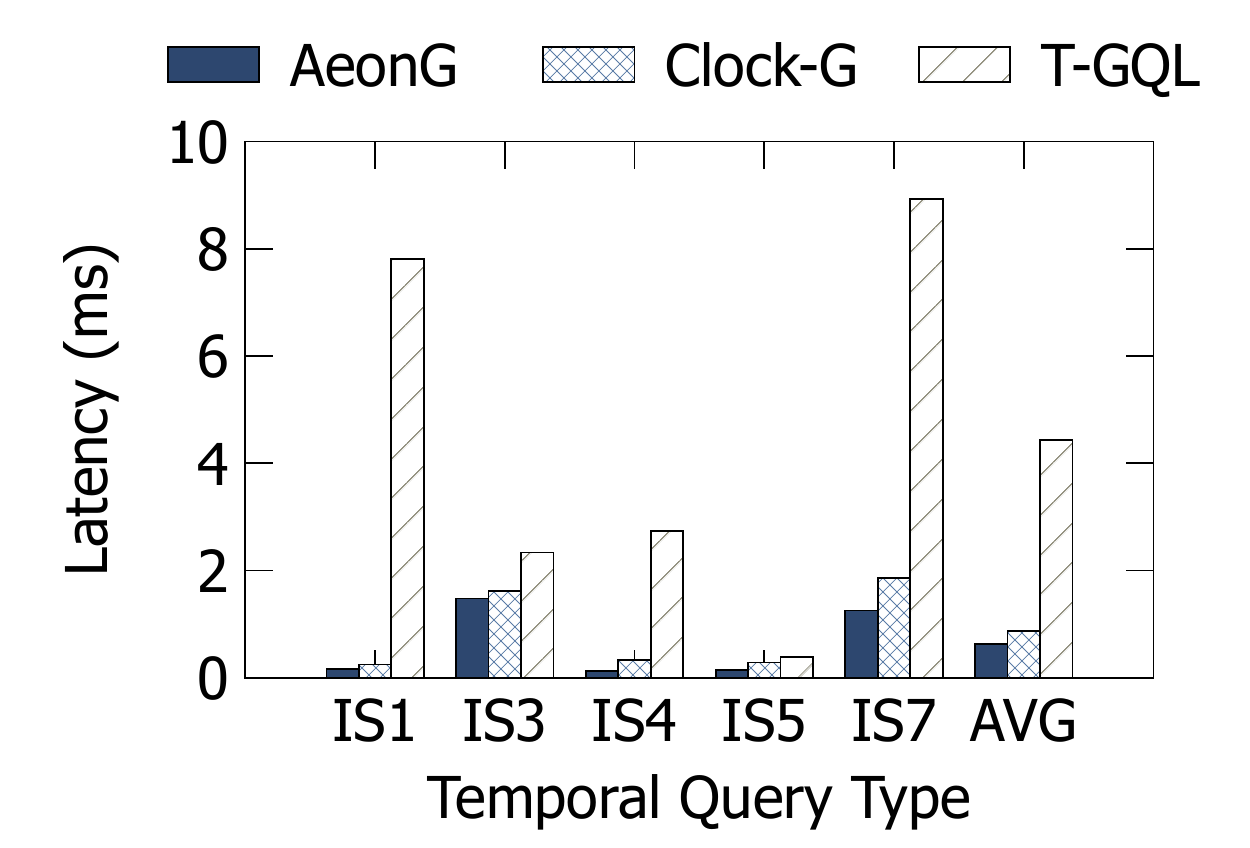}
    \end{minipage}%
    }%
   \subfigure[T-gMark]{%varying $n$ (DBpedia)
    \begin{minipage}[t]{0.5\linewidth}
     \label{Fig.gmark.time}
    \includegraphics[width=1\textwidth]{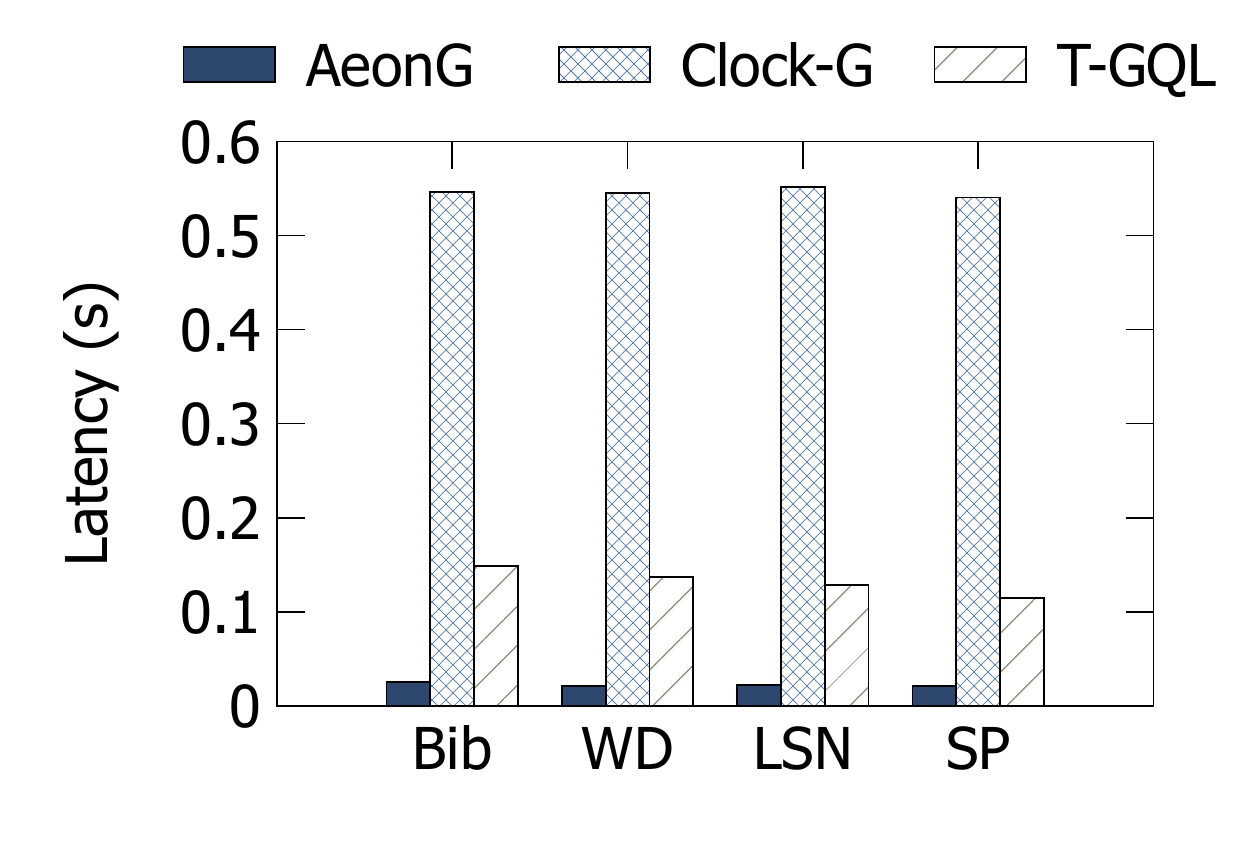}
    \end{minipage}%
   }%
  \vspace{-2mm}
   \captionsetup{justification=raggedright}
\caption{\todo{Comparisons on Temporal Query Latency}}
%alternative workloads
  \label{Fig.alternative.workload}
  % \vspace{-2mm}
\end{figure}
}

\subsubsection{Experiments on temporal query latency.}
\label{sec:temporal_query}
% \textbf{Latency of temporal queries.}  
We 
% provide a detailed analysis of the 
now analyze the 
performance of temporal queries 
% for accessing historical data 
under various configurations. 
% In the subsequent experiment, we fix the size of graph operations at 320k, standardize all systems to inject one graph operation per second, and 
% \textit{Various temporal queries.} 
We first conduct performance evaluations across various temporal queries in T-mgBench, and plot the query latency on different temporal queries in Figure \ref{Fig.tpokec.various_queries}. 
By default, we set the time slice length for Q2 and Q4 to 100s. 
We can observe that {\tgdb} reduces the query latency by {$2.57\times$} compared to Clock-G and {$37.57\times$} compared to T-GQL.
% consistently outperforms the other systems for all types of query statements. On average, {\tgdb} 
The superior performance of {\tgdb} can be attributed to our built-in query engine, which employs an efficient anchor-based version retrieval technique to avoid unnecessary version traversal.
In contrast, to access desired historical graph elements, T-GQL necessitates traversing the entire graph , while Clock-G requires fetching the corresponding historical snapshot and appending logs on it, thereby resulting in slower performance.

% \textit{Varying the number of graph operations.} 
We then study the latency of temporal query Q1 with varying the number of graph operations.
Figure \ref{Fig.tpokec.q1_time_scale} shows that {\tgdb} outperforms Clock-G by up to {$1.43\times$} and achieves an up to {$47.18\times$} improvement compared to T-GQL.
This performance gap becomes increasingly pronounced with a growing number of graph operations. 
% When the number of graph operations reaches 400k, {\tgdb} outperforms T-GQL by up to {$47.18\times$}.
% as the number of graph operations increases.
% the query latency of each system on Q1 under different numbers of graph operations.
As discussed, {\tgdb} exhibits better performance due to the proposed optimized temporal query engine.
% specifically optimized for temporal query processing.
% , highlighting a substantial performance advantage.
% It is worth mentioning that despite having more compact storage compared to Clock-G, {\tgdb} achieves similar or even faster query performance, benefiting from its efficient temporal query engine
% enhanced by the anchor-based version retrieval technique.
% \textit{Varying the operation frequency.} 
\todo{We further evaluate the query latency of Q3 across ``cold'', ``warm'', and 
``hot'' queries. 
We categorize these queries based on the vertices they access.
For example, a query accessing ``hot'' vertices is classified as a ``hot'' query. 
We divide all vertices in the database into ``cold'', ``warm'', and 
``hot'' categories according to their access possibility, which ranges from 0\% to 100\% based on the Zipf distribution.
%%%%hjm: 100->0?
% \todo{describe the figure here.}
As shown in Figure \ref{Fig.tpokec.q3_frequency}, {\tgdb} outperforms the next-best system, Clock-G, in all query categories by up to {$2.21\times$}.
We can also observe that ``hot'' queries, which require accessing more historical data, generally have lower performance than ``warm'' and ``cold'' queries.
However, in {\tgdb}, ``hot'' queries are only {$1.29\times$} slower than ``warm'' queries, while in T-GQL and Clock-G, ``hot'' queries underperform ``warm'' queries by up to {$2.14\times$} and {$3.77\times$}, respectively. 
As discussed, the smaller performance gap of {\tgdb} can be attributed to our anchor-based version retrieval technique, which avoids unnecessary version traversal.}
% We next evaluate the performance with varying operation frequencies, focusing on Q3 with hot, warm, and cold queries.
% As depicted in Figure \ref{Fig.tpokec.q3_frequency}, all temporal systems experience an increase in query time as the frequency of hotter queries rises. Hotter queries require more historical data, resulting in more time. 
% {\tgdb} demonstrates superior capabilities in handling all query types, while T-GQL and Clock-G experience significant performance degradation for hotter queries.
% When transitioning from warm to hot queries, {\tgdb} experiences a minimal performance drop of \makesure{$1.29\times$}. In contrast, Clock-G endures a \makesure{$2.14\times$} drop and T-GQL suffers a substantial \makesure{$3.77\times$} performance decline.
% Compared to T-GQL, {\tgdb} avoids unnecessary subgraph retrieval through the utilization of two separate storage systems. 
% Compared to Clock-G, {\tgdb} implements an efficient anchor-based version retrieval technique, which accelerates query performance.
% \textit{Varying the length of the time slice.} 
We also study the latency of temporal query Q4 with varying time slice length
% . Given the disparate write performance of each system, as depicted in Figure \ref{Fig.tpokec.graph_op}, we standardize all systems to inject one graph operation per second while manipulating time slice lengths 
from 1s to 200s. 
As observed in Figure \ref{Fig.tpokec.q4_timeslice}, {\tgdb} outperforms Clock-G and T-GQL by up to {$2.27\times$} and {$33.23\times$}, respectively, showing a consistently superior performance under 
different time slice lengths.
% \todo{describe the figure here.}
% As shown in Figure \ref{Fig.tpokec.q4_timeslice}, 
% {\tgdb} achieves a performance improvement of up to \makesure{$2.27\times$} compared to Clock-G and an impressive \makesure{$33.23\times$} compared to T-GQL. 
% {\tgdb} exhibits a consistently superior performance under different time spans, owing to our anchor-based retrieval strategy.

{
% \subsubsection{ Experiment on the T-LDBC and T-gMark workloads.} 
\todo{We additionally evaluate the temporal query performance on T-LDBC.
\maintext{As depicted in Figure \ref{Fig.ldbc.various_q}, {\tgdb} outperforms among all temporal query types and achieves lower latency by up to {$1.37\times$} and {$7\times$} than Clock-G and T-GQL, in alignment with the trends observed in Figure \ref{Fig.tpokec.various_queries}.}
\extended{As depicted in Figure \ref{Fig.ldbc.various_q} and \ref{Fig.ldbc.various_q.i26}, {\tgdb} outperforms among all temporal query types and achieves lower latency by up to {$1.46\times$} and {$12.07\times$} than Clock-G and T-GQL, in alignment with the trends observed in Figure \ref{Fig.tpokec.various_queries}. We also rerun the experiments on T-LDBC by utilizing a larger graph with 9.28M vertices and 52.7M edges.
During these experiments, T-GQL encounters the out-of-memory issue due to the increased graph size.
However, as shown in Figure \ref{Fig.ldbc.s3}, {\tgdb} continues to outperform Clock-G by up to $2.06\times$, which aligns with the findings shown in Figure \ref{Fig.ldbc.various_q}. }
% achieving an average {$1.37\times$} faster performance than Clock-G and {$7\times$} faster performance than T-GQL, aligning with the results observed in Figure \ref{Fig.tpokec.various_queries}. 
\maintext{Due to space limitations, we leave the latency details of IS2 and IS6 in our extended manuscript \cite{aeongsupplement}.
Their trends align with Figure \ref{Fig.ldbc.various_q}, only but their scales differ. }
Furthermore, we conduct experiments using T-gMrak.
As depicted in Figure \ref{Fig.gmark.time}, {\tgdb} consistency outperforms in all the datasets and demonstrates up to {$26.16\times$} faster temporal query performance than Clock-G and {$6.56\times$} faster than T-GQL. 
}
}

% \extended{\input{extend_version/exp_tldbc_tgmark}}

% \todo{Moreover, it demonstrates up to {$2.82\times$} faster graph operation performance than Clock-G and {$10.11\times$} faster performance than T-GQL.} In terms of temporal queries, as illustrated in Figure \ref{Fig.ldbc.various_q}, {\tgdb} outperforms all temporal query types, achieving an average {$1.37\times$} faster performance than Clock-G and {$7\times$} faster performance than T-GQL. {\tgdb} benefit from its efficient temporal query engine enhanced by the anchor-based version retrieval technique.

\todo{

\extended{
\begin{figure*}[]
  \centering
  \subfigure[T-mgBench]{
    \begin{minipage}[t]{0.25\linewidth}
     \label{Fig.tpokec.various_queries}
    \includegraphics[width=1\textwidth]{figures/results/0_tmgbench/various_queries.pdf}
    \end{minipage}%
   }%
    \subfigure[T-mgBench (Q1): \# of Graph Operations]{
    \begin{minipage}[t]{0.25\linewidth}
     \label{Fig.tpokec.q1_time_scale}
    \includegraphics[width=1\textwidth]{figures/results/0_tmgbench/q1_time.pdf}
    \end{minipage}%
   }%
    \subfigure[T-mgBench (Q3): Accessed Vertex Types]{
    \begin{minipage}[t]{0.25\linewidth}
     \label{Fig.tpokec.q3_frequency}
    \includegraphics[width=1\textwidth]{figures/results/0_tmgbench/Q3_frequency.pdf}
    \end{minipage}%
   }%
    \subfigure[T-mgBench (Q4): Time Slice]{
    \begin{minipage}[t]{0.25\linewidth}
     \label{Fig.tpokec.q4_timeslice}
    \includegraphics[width=1\textwidth]{figures/results/0_tmgbench/q4_timeslice.pdf}
    \end{minipage}%
   }%
   \vspace{-3mm}
    \subfigure[T-LDBC]{%varying $n$ (DBpedia)
    \begin{minipage}[t]{0.25\linewidth}
     \label{Fig.ldbc.various_q}
    \includegraphics[width=1\textwidth]{figures/results/1_ldbc/ldbc_short_q.pdf}
    \end{minipage}%
    }%
    \subfigure[T-LDBC]{%varying $n$ (DBpedia)
    \begin{minipage}[t]{0.25\linewidth}
     \label{Fig.ldbc.various_q.i26}
    \includegraphics[width=1\textwidth]{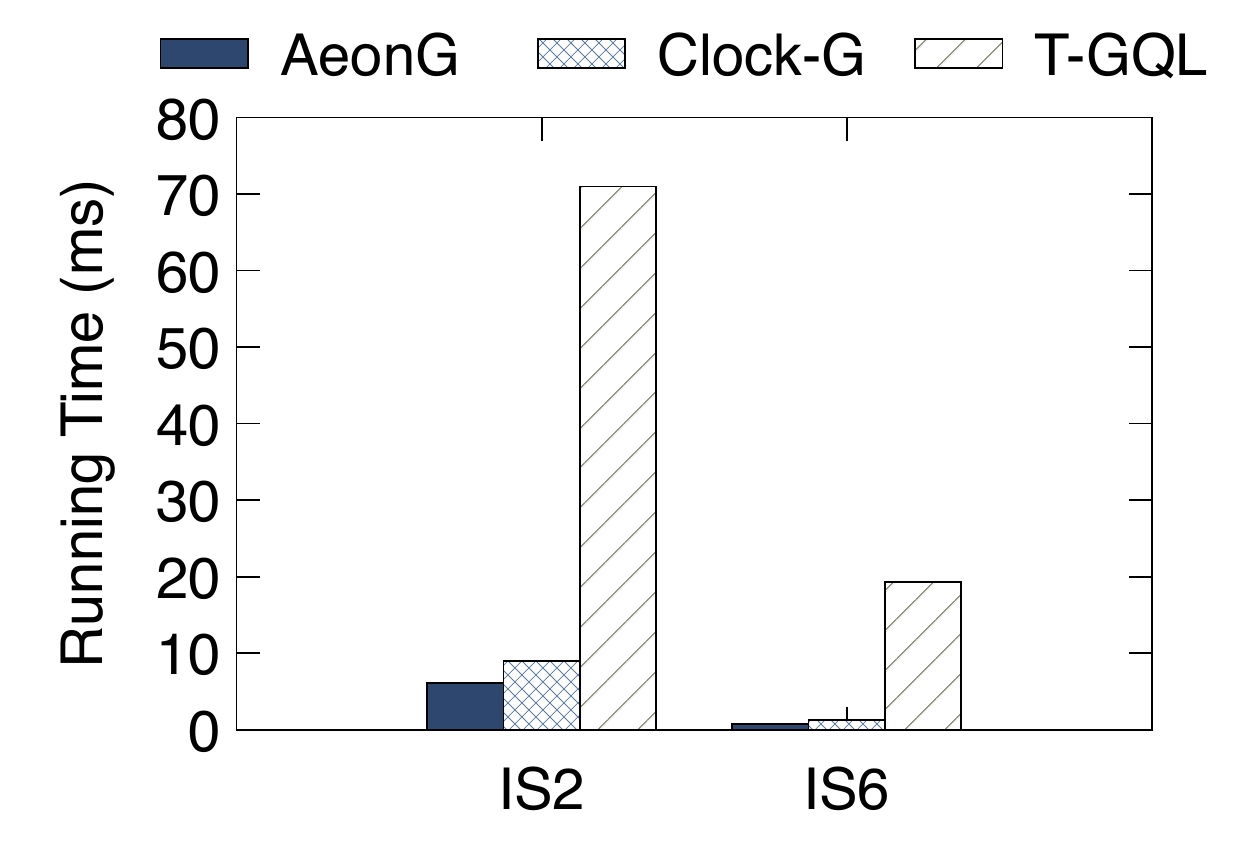}
    \end{minipage}%
    }%
    \subfigure[T-LDBC (S3)]{%varying $n$ (DBpedia)
    \begin{minipage}[t]{0.25\linewidth}
     \label{Fig.ldbc.s3}
    \includegraphics[width=1\textwidth]{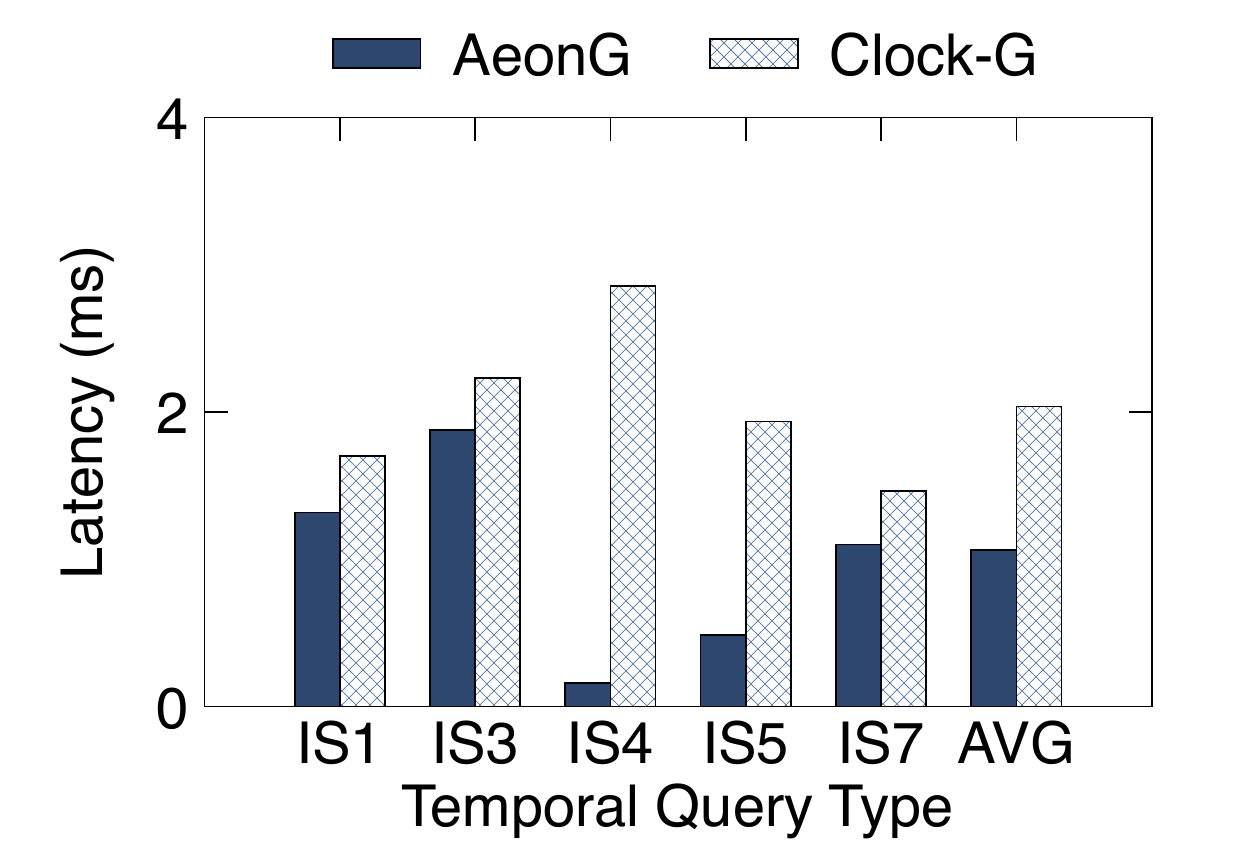}
    \end{minipage}%
    }%
   \subfigure[T-gMark]{%varying $n$ (DBpedia)
    \begin{minipage}[t]{0.25\linewidth}
     \label{Fig.gmark.time}
    \includegraphics[width=1\textwidth]{figures/results/2_gmark/time.pdf}
    \end{minipage}%
   }%
  \vspace{-3mm}
   \captionsetup{justification=raggedright}
\caption{\todo{Comparisons on Temporal Query Latency}}
%alternative workloads
  \label{Fig.alternative.workload}
  % \vspace{-2mm}
\end{figure*}
}

\subsection{Performance Analysis on {\tgdb}}
\label{sec:evaluation.analysis}
We now provide an in-depth analysis of {\tgdb}'s performance under diverse configurations. In the following experiments, we fix RocksDB's MemTable size to the default value of {64MB}.

\subsubsection{The performance of non-temporal queries.}
\label{sec:evaluation.non-temporal}
We first analyze the performance of {\tgdb} on non-temporal queries to study the impact of introducing temporal features in its fundamental system, Memgraph. 
% We evaluate all non-temporal queries derived from three origin unextended workloads: mgBench, LDBC, and gMark. 
% These queries are mixed in a 1:1 ratio, and the average execution time is measured. 
We use various non-temporal queries defined in three origin unextended workloads: mgBench, LDBC, and gMark. We run corresponding queries based on the datasets generated by T-mgBench, T-LDBC, and T-gMark, and plot the average query latency of each workload in Figure \ref{Fig.aeong.mgbench}. 
The results indicate that {\tgdb} experiences an acceptable performance drop of up to {9.74\%} compared with Memgraph, a trade-off for its support of temporal queries.
% , within a range of \makesure{9.8\%} to \makesure{17.97\%}.
% These outcomes substantiate that {\tgdb}'s temporal support is of a lightweight and non-intrusive nature to the original system.
\extended{In addition, we provide a detailed performance breakdown for each specific query of mgBench in Figure \ref{Fig.aeong.mgbench.various_q}.
We can observe that {\tgdb} maintains comparable performance with Memgraph in most of the tested queries.}
{\tgdb} adopts a design that separates the current database and asynchronously transfers historical data, ensuring minimal impact on dominant non-temporal queries.}

\maintext{\begin{figure}[]
  \centering
\includegraphics[width=0.48\textwidth]{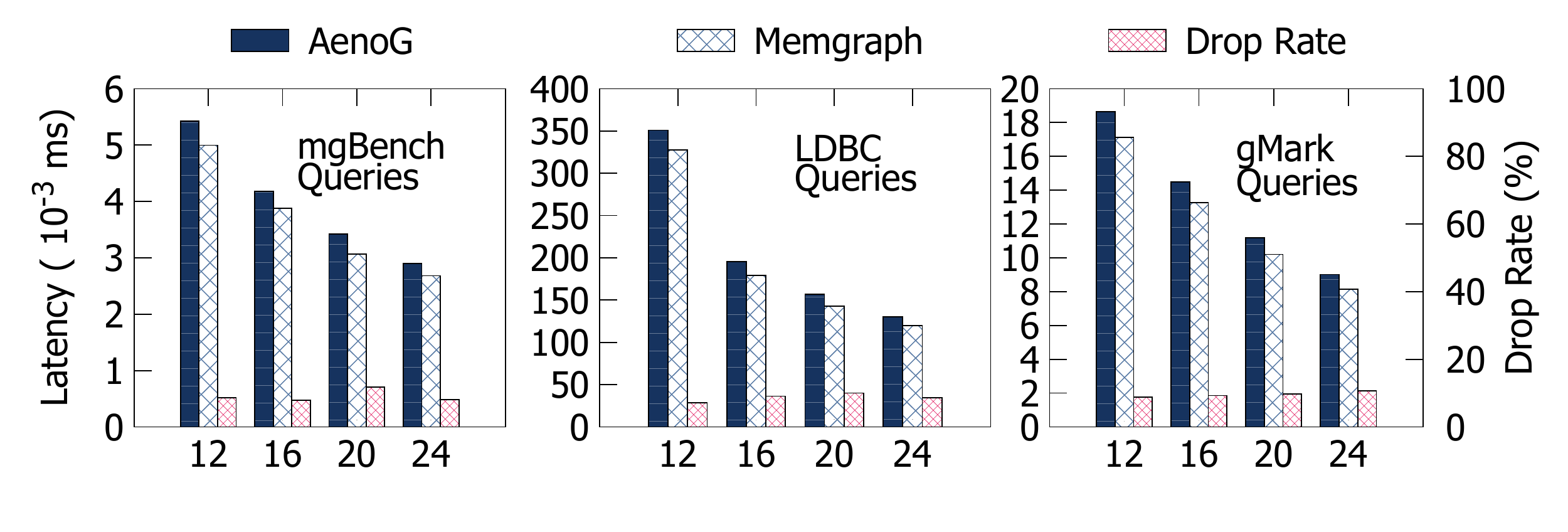}
  \vspace{-8mm}
   % \captionsetup{justification=raggedright}
\caption{\todo{{\tgdb} vs Memgraph on Non-temporal Queries}}
  \label{Fig.aeong.mgbench}
  % \vspace{-2mm}
\end{figure}}

\todo{
\subsubsection{The impact of historical data migration.}
\label{exp:read_perf}
We next use the T-mgBench workload to study the query performance with varying the GC interval to control the frequency of historical data migration.
We plot the latency of queries across different data types: current, reclaimed, and unreclaimed data, and graph operation in Figure \ref{Fig.aeong.gc}.
As observed, the query performance for current and unreclaimed data is relatively similar, both outperforming reclaimed data queries by up to {$25.8\%$}.
This difference is attributed to the fact that querying current and unreclaimed data both need to traverse the version chain in the current storage, while querying reclaimed data requires the additional step of reconstructing a historical version using anchors and deltas in the historical storage, as detailed in Section~\ref{sec:hybrid_storage}.
Further, we note that increasing the GC interval from 1s to 1000s leads to a {$17.3\%$} decrease in the graph operation latency and a {$9.5\%$} increase in the query latency.
This is expected as less frequent migrations can reduce contention with graph operation, thereby enhancing graph operation performance.
In contrast, less frequent migrations result in longer version-chain traversal in the current storage, negatively impacting query performance.
}

\extended{\begin{figure}[]
  \centering
   % \captionsetup{justification=raggedright}
    \subfigure[{\tgdb} vs Memgraph on All Workloads]{%varying $n$ (DBpedia)
    \begin{minipage}[t]{1\linewidth}
     \label{Fig.aeong.all}
    \includegraphics[width=1\textwidth]{figures/results/3_aeong/all_workloads.pdf}
    \end{minipage}%
   }%
   \vspace{-3mm}
   \subfigure[Running Time of Specific Queries in mgBench]{
    \begin{minipage}[t]{1\linewidth}
     \label{Fig.aeong.mgbench.various_q}
    \includegraphics[width=1\textwidth]{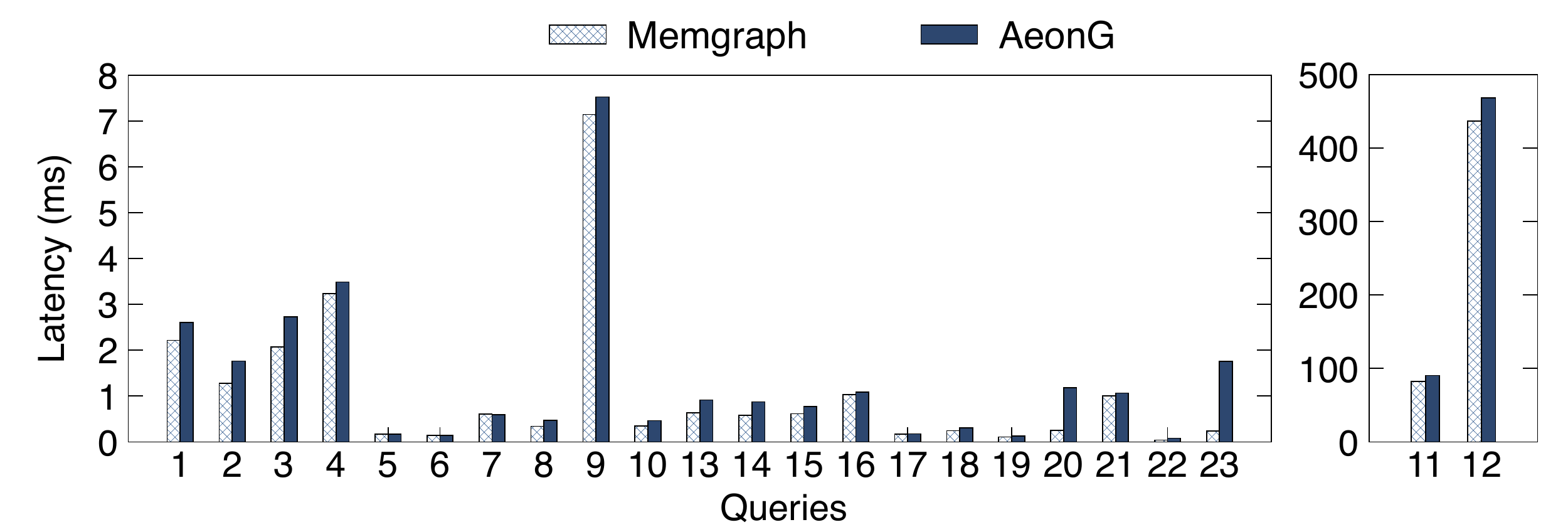}
    \end{minipage}%
   }%
   \vspace{-4mm}
\caption{{\tgdb} vs Memgraph on Non-temporal Queries}
  \label{Fig.aeong.mgbench}
  % \vspace{-2mm}
\end{figure}}
% \vspace{-1mm}

\todo{
\subsubsection{The analysis on the anchor interval.} 
\label{exp:anchor}
We evaluate the effectiveness of our adaptive anchoring approach using the T-LDBC workload. 
As shown in Figure \ref{Fig.aeong.anchor}, when we assign a fixed anchor interval $u$ to each graph object and vary $u$ from 1 to 1000, we observe that storage consumption of the historical storage decreases by {2.69$\times$} and the temporal query latency increases by {$2.15\times$}. 
In contrast, our adaptive anchoring approach consistently achieves near-optimal query performance and storage efficiency against all fixed anchor interval settings, because of its ability to properly balance query latency and storage overhead efficiency. 
% properly. 
% We now study the impact of anchor data on the T-LDBC workload in terms of historical storage consumption and temporal query execution time. We compare our adaptive anchoring approach with a global constant anchoring configuration, where anchors are placed at consistent regular intervals between the delta data. We vary the constant anchor interval threshold, denoted as $u$, from 1 to 1000.
% As depicted in Figure \ref{Fig.aeong.anchor}, higher values of $u$ lead to reduced storage consumption but degrade query performance. When transitioning $u$ from 1 to 1000, the storage consumption decreases by a factor of \makesure{2.69} and the query latency increases by up to \makesure{$2.15\times$}. Moreover, our proposed adaptive anchoring approach demonstrates enhanced flexibility and strikes a balance between storage and query performance compared to the global constant anchoring configuration. The query performance closely approximates that of $u=1$, while simultaneously achieving a similar storage consumption level as that of $u=1000$.
}

\todo{
\subsubsection{The impact of the historical retention period.}
\label{exp:historical_retention}
% We now activate historical data purging in {\tgdb}. 
We utilize the T-LDBC workload to evaluate the historical storage overhead and temporal query latency with varying historical data retention periods. 
We simulate one day's amount of graph evolution in just one minute, which is done by assuming a daily operation count of 100k and executing these operations within one minute. 
Consequently, we set historical data retention periods at 15, 30, 90, and 180 minutes, simulating real-world scenarios of half a month, one month, one quarter, and half a year, respectively.
% as 15 minutes, 30 minutes, 90 minutes, and 180 minutes
% , configuring the retention period of historical data in RocksDB from 15 minutes to 180 minutes.
% Employing the T-LDBC workload, we inject 100k graph operations per minute over a span of 360 minutes and evaluate the storage overhead and latency of IS4. 
The results, shown in Figure \ref{Fig.aeong.clean_history}, demonstrate that the storage consumption increases by {$6.02\times$} and query performance decreases by {$1.62\times$} as the data retention period extends from 15 to 180 minutes.
This trend is expected since longer retention period results in more historical data being maintained, leading to decreased query performance.
Based on this observation, we consider enabling users to set a proper \texttt{retention\_period} to achieve a balance among storage overhead, historical data duration, and query performance.
% We vary the historical data retention period from 15 minutes to 180 minutes and evaluate the historical storage overhead and latency. 
% To enhance experimental efficiency, we employ a time dilation technique in which each simulated minute of graph operations corresponds to a full day of real-world graph operations.  
% Data retention periods are configured to 15 minutes, 30 minutes, 90 minutes, and 180 minutes, corresponding to real-world scenarios simulating half a month, one month, a quarter, and half a year of data retention, respectively. 
}

% \begin{figure}[]
%   \centering
%    % \subfigure[{\tgdb} vs Memgraph]{%varying $n$ (DBpedia)
%    %  \begin{minipage}[t]{0.5\linewidth}
%    %   \label{Fig.aeong.mgbench}
%    %  \includegraphics[width=1\textwidth]{figures/results/3_aeong/mgbench_time.pdf}
%    %  \end{minipage}%
%    % }%
%    \subfigure[The effect of the data migration]{%varying $n$ (DBpedia)
%     \begin{minipage}[t]{0.5\linewidth}
%      \label{Fig.aeong.gc}
%     \includegraphics[width=1\textwidth]{figures/results/3_aeong/gc_rw.pdf}
%     \end{minipage}%
%    }%
%    \subfigure[The effect of the anchor interval]{%varying $n$ (DBpedia)
%     \begin{minipage}[t]{0.5\linewidth}
%      \label{Fig.aeong.anchor}
%     \includegraphics[width=1\textwidth]{figures/results/3_aeong/anchor.pdf}
%     \end{minipage}%
%    }%
   
%    \subfigure[Varying the retention period]{%varying $n$ (DBpedia)
%     \begin{minipage}[t]{0.5\linewidth}
%      \label{Fig.aeong.clean_history}
%     \includegraphics[width=1\textwidth]{figures/results/3_aeong/clean_history.pdf}
%     \end{minipage}%
%    }%
%    \subfigure[Extend on TiKV]{%varying $n$ (DBpedia)
%     \begin{minipage}[t]{0.5\linewidth}
%      \label{Fig.aeong.tikv}
%     \includegraphics[width=1\textwidth]{figures/results/3_aeong/tikv_all.pdf}
%     \end{minipage}%
%    }%
%   \vspace{-4mm}
%    \captionsetup{justification=raggedright}
% \caption{Performance Analysis on AeonG}
%   \label{Fig.aeong}
%   \vspace{-2mm}
% \end{figure}
% \vspace{-1mm}

\maintext{\begin{figure}[]
  \centering
   \subfigure[The Effect of the Data Migration]{%varying $n$ (DBpedia)
    \begin{minipage}[t]{0.5\linewidth}
     \label{Fig.aeong.gc}
    \includegraphics[width=1\textwidth]{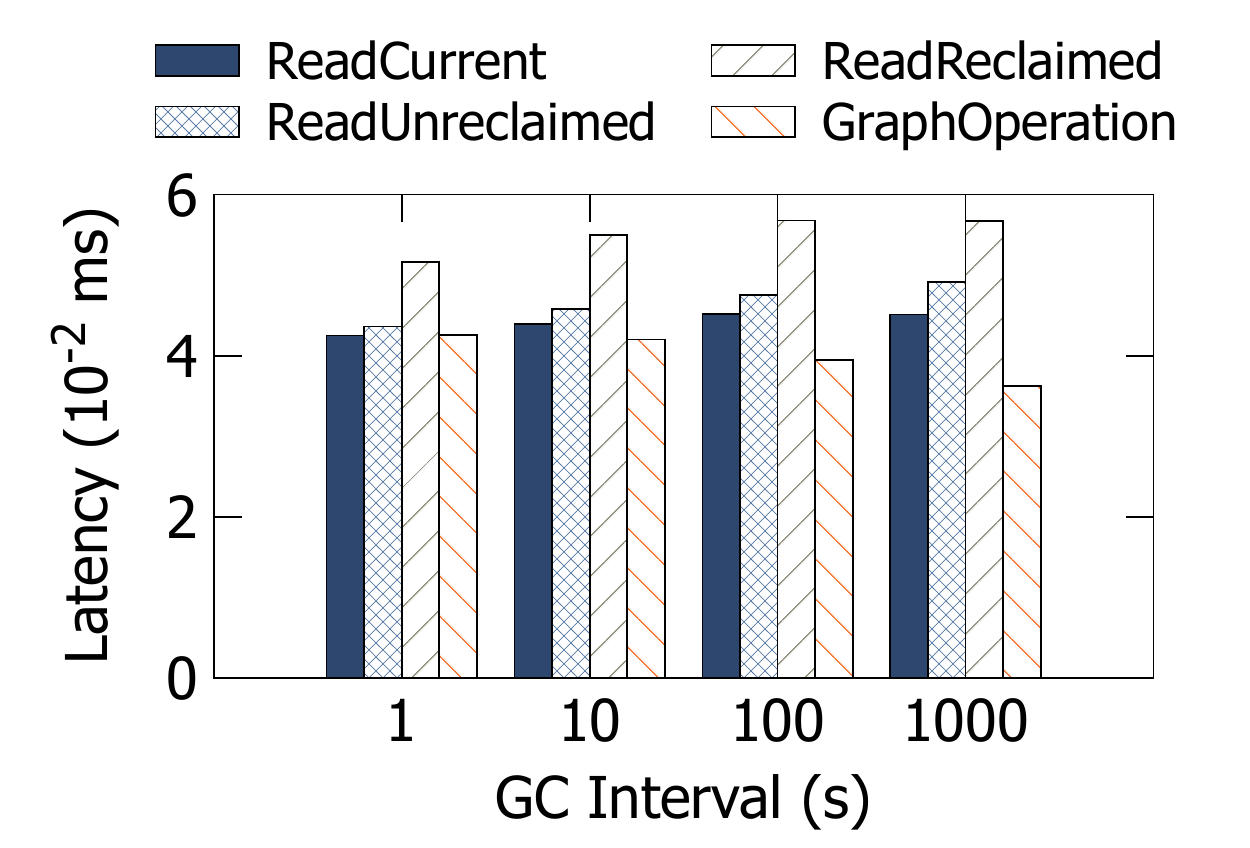}
    \end{minipage}%
   }%
   \subfigure[The Effect of the Anchor Interval]{%varying $n$ (DBpedia)
    \begin{minipage}[t]{0.5\linewidth}
     \label{Fig.aeong.anchor}
    \includegraphics[width=1\textwidth]{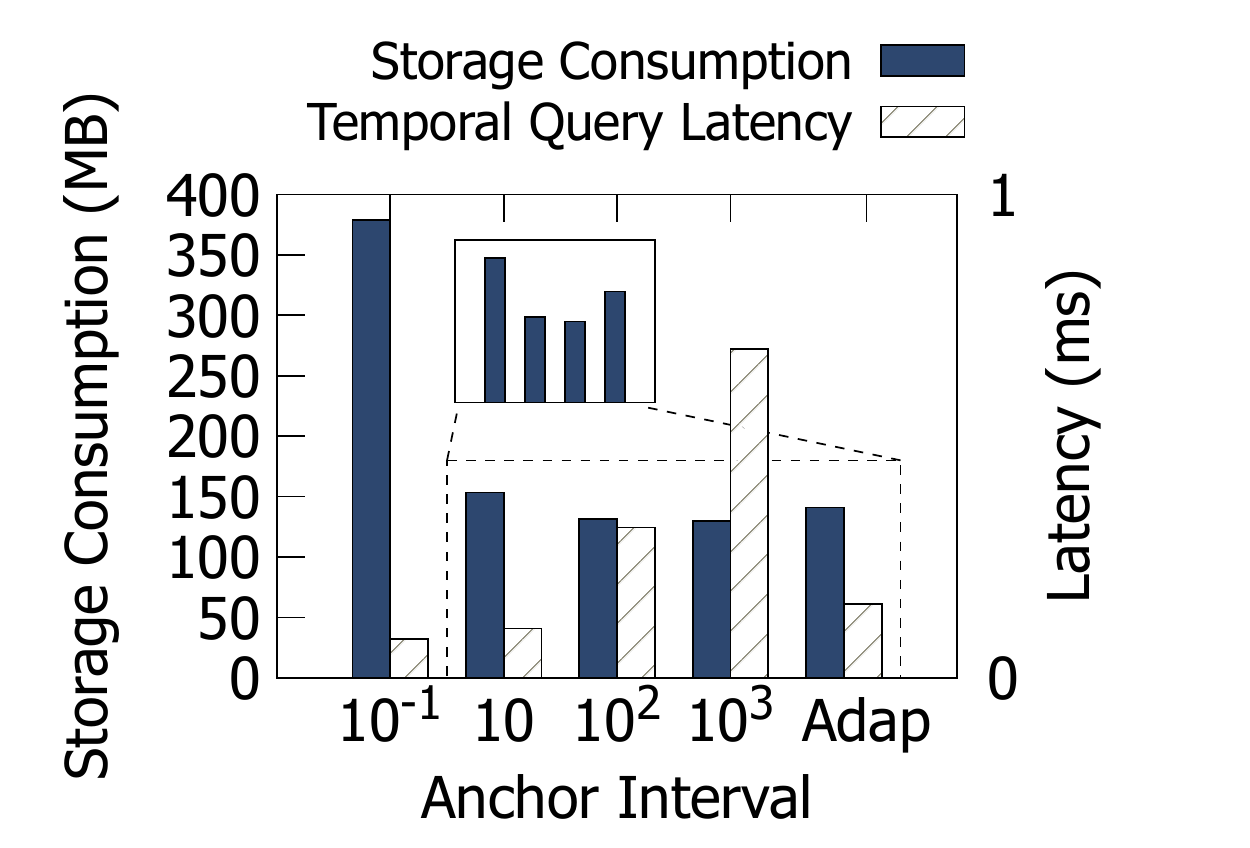}
    \end{minipage}%
   }%
    \vspace{-3mm}
   \subfigure[Varying the Retention Period]{%varying $n$ (DBpedia)
    \begin{minipage}[t]{0.5\linewidth}
     \label{Fig.aeong.clean_history}
    \includegraphics[width=1\textwidth]{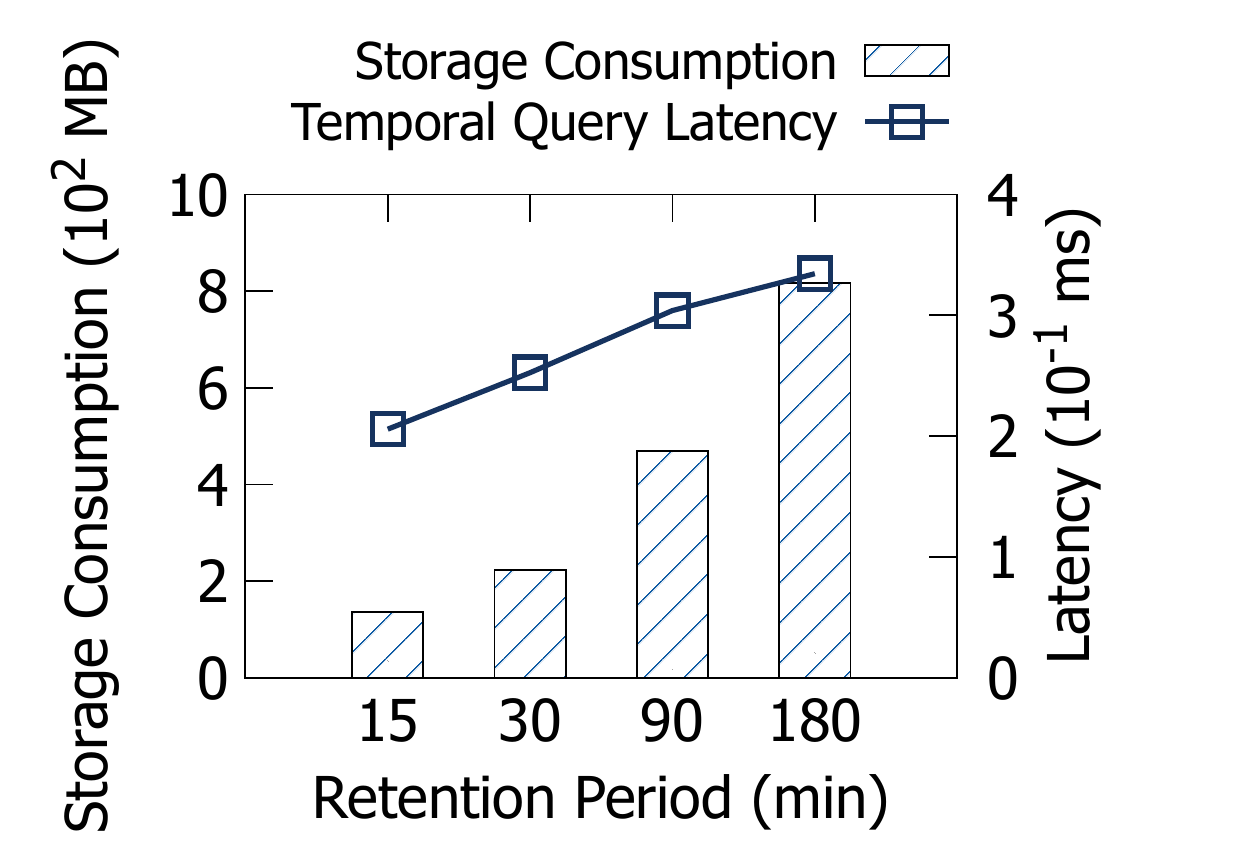}
    \end{minipage}%
   }%
   \subfigure[Extending on TiKV]{%varying $n$ (DBpedia)
    \begin{minipage}[t]{0.5\linewidth}
     \label{Fig.aeong.tikv}
    \includegraphics[width=1\textwidth]{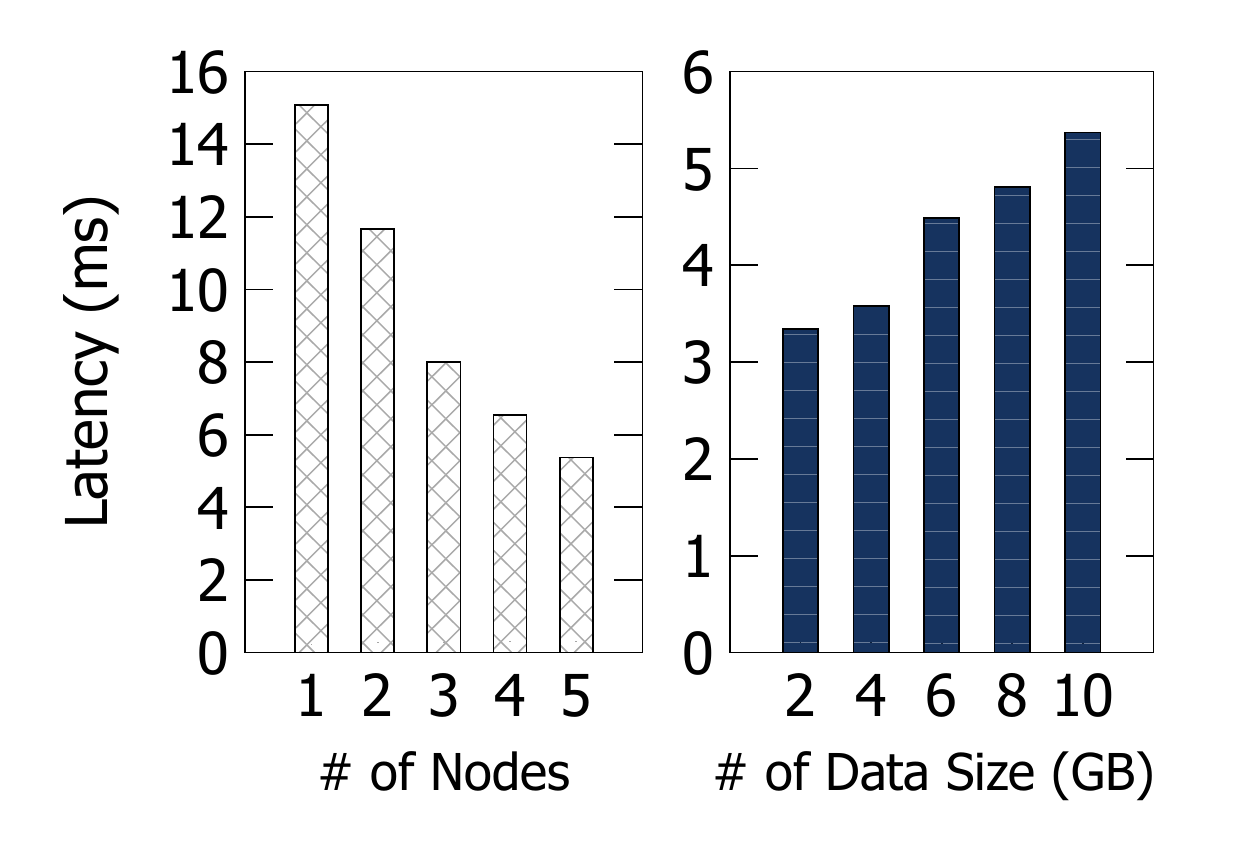}
    \end{minipage}%
   }%
  \vspace{-2mm}
   \captionsetup{justification=raggedright}
\caption{\todo{Performance Breakdown Analysis on {\tgdb}}}
  \label{Fig.aeong}
  % \vspace{-2mm}
\end{figure}
}
\todo{
\subsubsection{The scalability of {\tgdbd}.}
\label{exp:tikv}
We now deploy {\tgdbd} and TiKV across {5 nodes} by default, with historical data horizontally partitioned among these nodes. 
% Our scalability evaluation involved two key aspects: expanding the node count and handling increased data volumes.
First, we evaluate the impact of increasing the server count on performance \todo{with 10GB data volume} with the T-LDBC workload. 
The results, shown in the left part of Figure~\ref{Fig.aeong.tikv}, indicate that the temporal query of {\tgdbd} decreases by up to {$2.8\times$} when scaling from 1 to 5 servers.
The scalability of {\tgdbd} can be attributed to the improved parallelism achieved by adding more servers, where each server can independently process historical data retrieval requests with its TiKV instance.
Second, we assess temporal query latency using T-LDBC with the data volume increasing from 2GB to 10GB.
As shown in the right part of Figure~\ref{Fig.aeong.tikv}, the latency of {\tgdbd} increases by up to {$1.6\times$}, which is expected due to the greater cost of fetching graph objects from a larger database.
\todo{Similar trends are also reported in~\cite{clock-g,g*,khurana2016storing}.}

\extended{\begin{figure}[]
  \centering
   \subfigure[The Effect of the Data Migration]{%varying $n$ (DBpedia)
    \begin{minipage}[t]{0.5\linewidth}
     \label{Fig.aeong.gc}
    \includegraphics[width=1\textwidth]{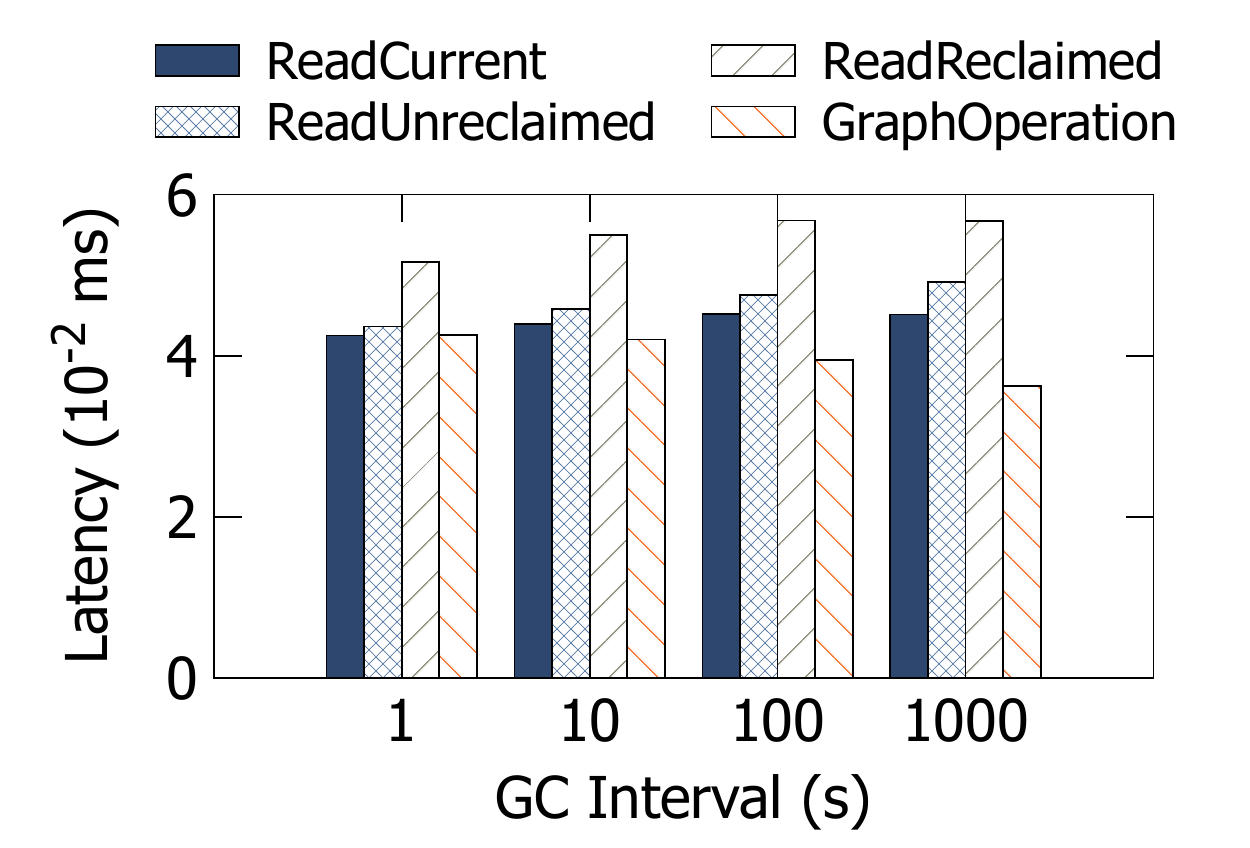}
    \end{minipage}%
   }%
   \subfigure[The Effect of the Anchor Interval]{%varying $n$ (DBpedia)
    \begin{minipage}[t]{0.5\linewidth}
     \label{Fig.aeong.anchor}
    \includegraphics[width=1\textwidth]{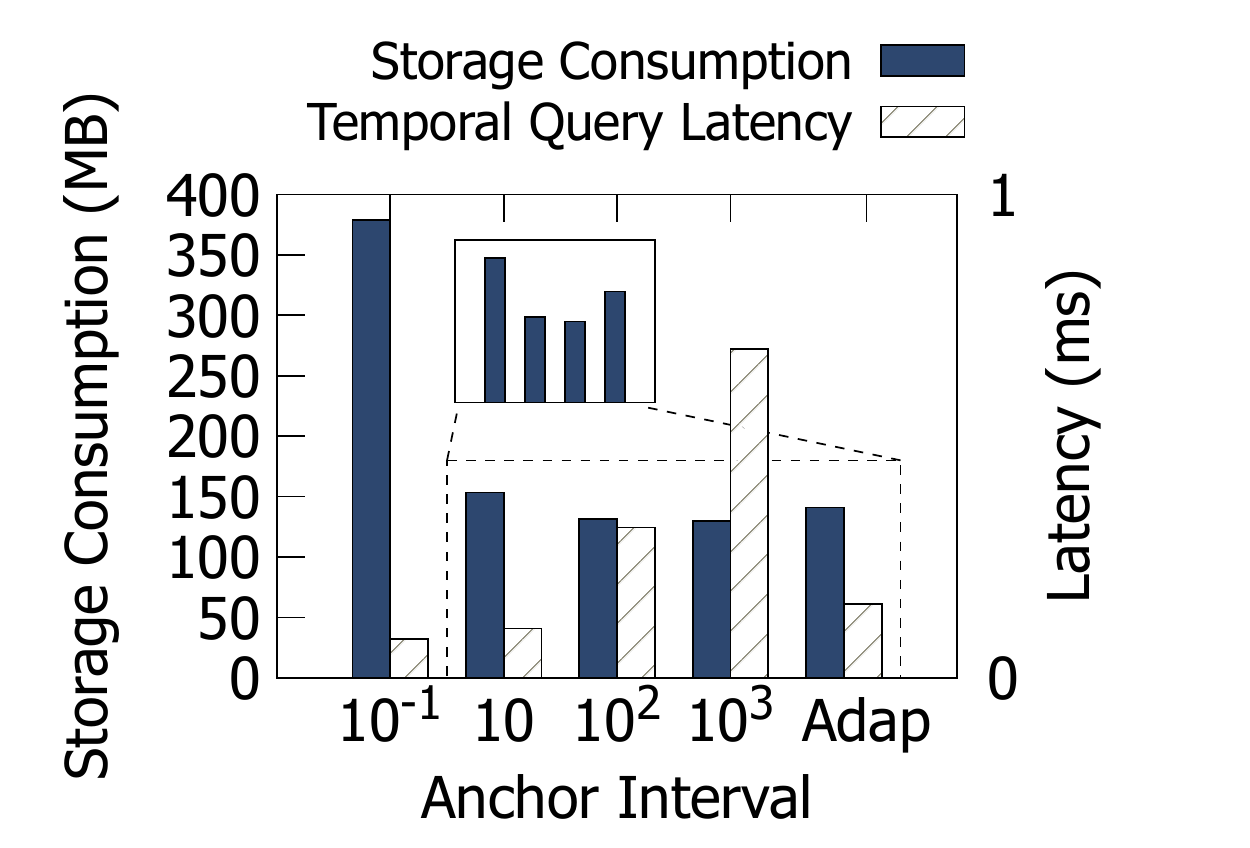}
    \end{minipage}%
   }%
    \vspace{-3mm}
   \subfigure[Varying the Retention Period]{%varying $n$ (DBpedia)
    \begin{minipage}[t]{0.5\linewidth}
     \label{Fig.aeong.clean_history}
    \includegraphics[width=1\textwidth]{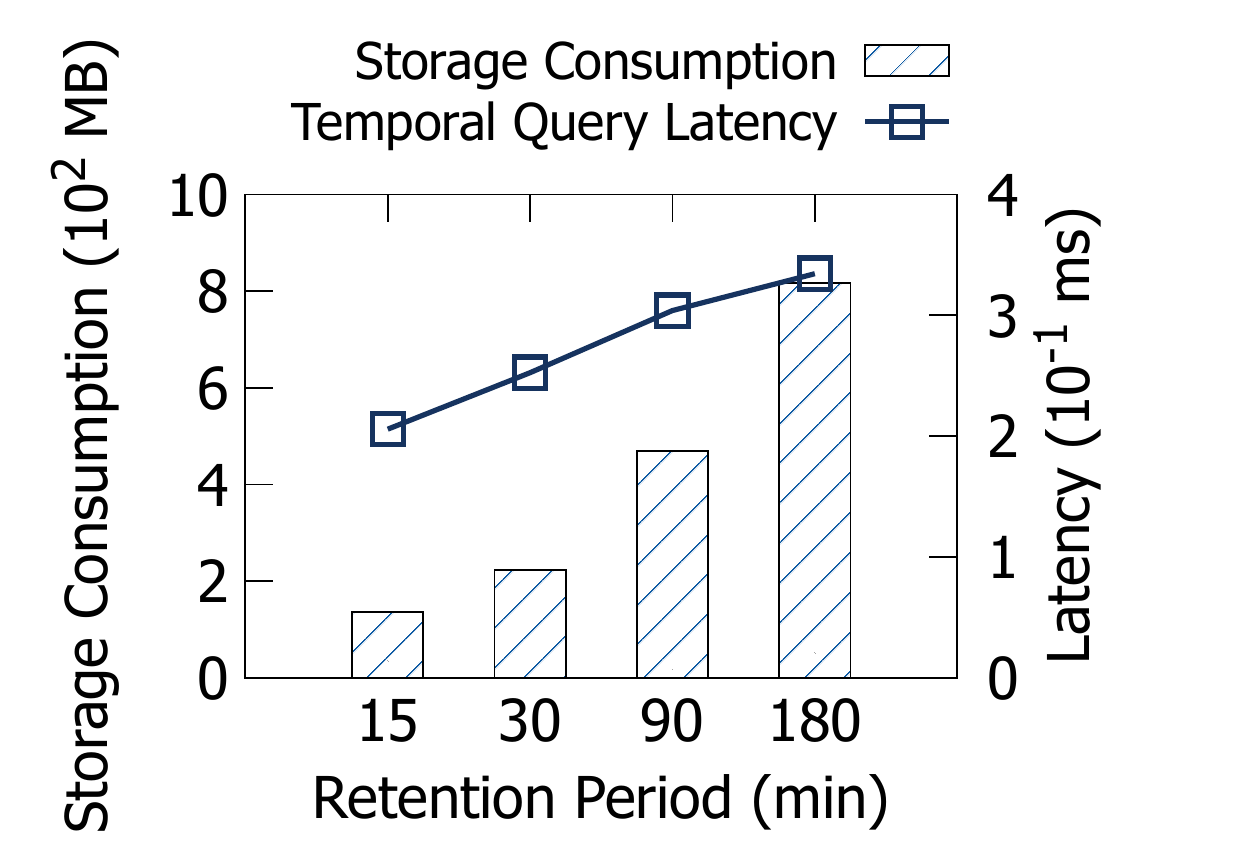}
    \end{minipage}%
   }%
   \subfigure[Extend on TiKV]{%varying $n$ (DBpedia)
    \begin{minipage}[t]{0.5\linewidth}
     \label{Fig.aeong.tikv}
    \includegraphics[width=1\textwidth]{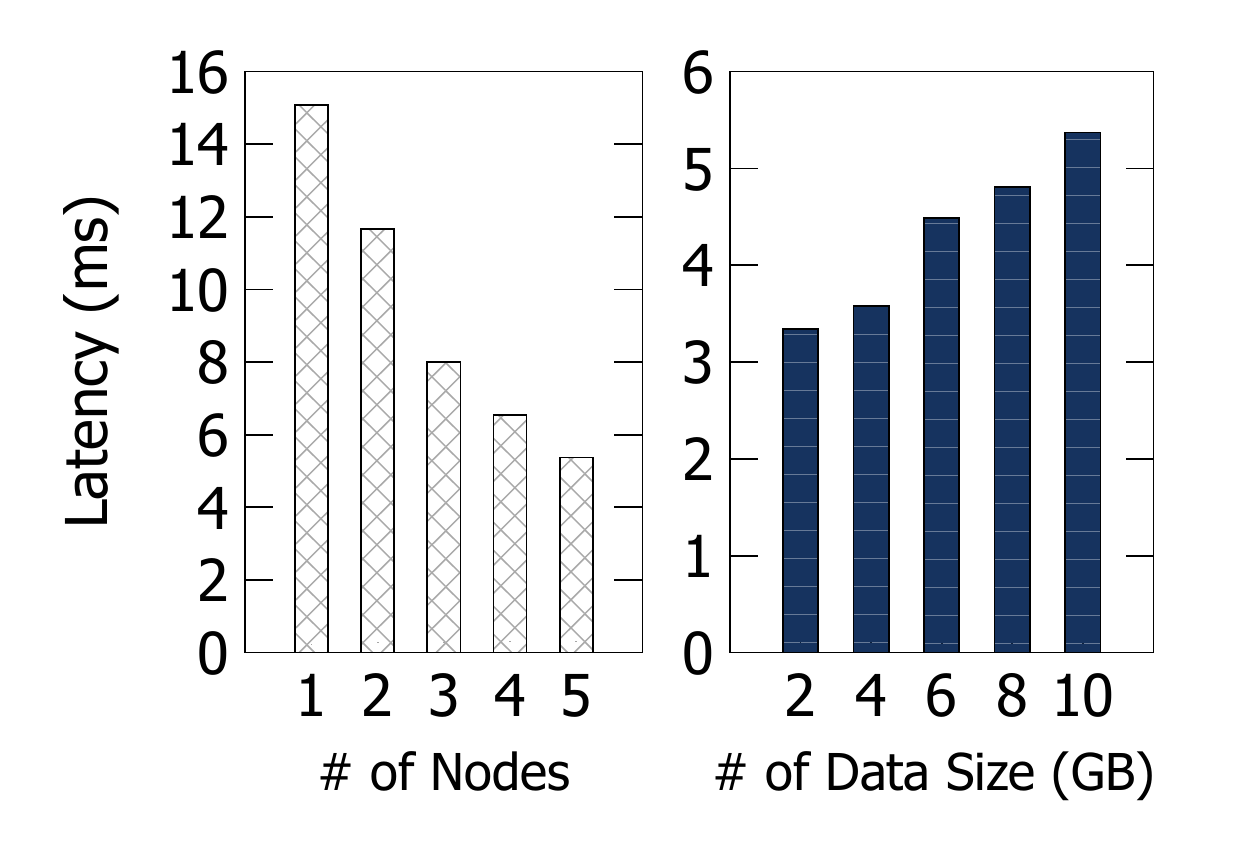}
    \end{minipage}%
   }%
  % \vspace{-4mm}
   \captionsetup{justification=raggedright}
\caption{\todo{Performance Breakdown Analysis on {\tgdb}}}
  \label{Fig.aeong}
  % \vspace{-2mm}
\end{figure}
}

}

\maintext{\section{Related Work}
\label{sec:relatedwork}
Temporal graph data management involves two primary approaches.
% : one at the application level and the other at the internal system level.
One approach integrates temporal features at the application level, utilizing commercial graph databases by attaching temporal metadata~\cite{Frame,liu2017keyword,GRADOOP,Timebased3,Timebased4, durand2017backlogs,T-GQL}. 
% Cattuto et al.~\cite{Frame} propose a temporal model with frames, defined as the finest temporal aggregation unit. 
Take a state-of-the-art approach T-GQL~\cite{T-GQL} in this field as an example. T-GQL adopts a specific representation of temporal graphs,
where conventional vertices are decomposed into {Object}, {Attribute}, {Value} vertices, and conventional edges remain the same. Time dimensions are introduced as properties of designed vertices and edges. 
% Other works exploit graph databases by representing the temporal dimension as an interval attribute to graph objects~\cite{liu2017keyword,GRADOOP,Timebased3,Timebased4, durand2017backlogs}. 
However, there may exhibit unpredictable performance due to underlying engines designed for static graphs.
% These systems are implemented at the application level upon existing database systems. They may exhibit unpredictable performance, as the underlying storage and query engines are designed for static graphs rather than evolving ones.  
An alternative line of research focuses on the system level, designing storage engines to handle growing historical data while enabling efficient querying. In this regard, two storage approaches, \textit{Copy} and \textit{Log}, are widely used to manage temporal graph data~\cite{DeltaGraph,ChronoGraph, Raphtory,llama,ImmortalGraph,clock-g, graphone, Chronos, Auxo, khurana2016storing}.
The \textit{Copy} approach~\cite{DeltaGraph,ChronoGraph}
stores an entire graph state whenever a batch of updates occurs. 
Although it simplifies graph querying, it results in excessive redundancy in the stored graph information. 
In contrast, the \textit{Log} approach~\cite{Raphtory,llama}
records every graph update activity in a log, offering a more compact solution but requiring costly reconstruction when executing a temporal graph query. To balance query performance and space overhead, the \textit{Copy+Log} approach~\cite{ImmortalGraph,clock-g, graphone, Chronos, Auxo, khurana2016storing}
combines a finite set of snapshots with a list of deltas between them. 
However, we argue the \textit{Copy+Log} approach is suboptimal since it still requires significant storage overhead to materialize the entire graph. Moreover, they lack support for a powerful temporal graph data model or a declarative temporal query language, restricting user convenience.

}
\maintext{% \vspace{-1mm}
\section{Conclusion}
\label{sec:conclusion}
In this paper, we propose {\tgdb}, a new graph database that efficiently offers built-in temporal support. {\tgdb} includes a formally defined temporal property graph model.
Based on this model, we propose a hybrid storage engine to store temporal data with minimal storage consumption.
Furthermore, {\tgdb} equips a native temporal query engine to enable efficient temporal query processing.
% ies with efficient and data consistent guarantee.
The results demonstrate that {\tgdb} achieves up to {5.73$\times$} lower storage consumption
% {2.82$\times$} lower latency for graph operations, 
and {2.57$\times$} lower latency for temporal queries against state-of-the-art approaches, while introducing only {{9.74\%}} performance degradation for supporting temporal features.

\begin{acks}
% We thank the anonymous reviewers for their feedback. 
This work was supported by the National Natural Science Foundation of China (Number 61972403, 62072458). 
\end{acks}}
\extended{\section{Related Work}
\label{sec:relatedwork}

Temporal data management has reached a mature stage in the field of relational database systems~\cite{RDBMS1,RDBMS2,RDBMS3,RDBMS4,RDBMS5,RDBMS6,RDBMS16}.
% , as demonstrated by extensive research~\cite{RDBMS1,RDBMS2,RDBMS3,RDBMS4,RDBMS5,RDBMS6} and the adoption of the well-established standard SQL:2011~\cite{sql-2011}. 
Nevertheless, the temporal data management in the graph context 
% study of temporal graph data management in the context of graphs 
remains relatively limited. In this section, we provide a review of temporal graph data management systems through two primary approaches: one at the application level and the other at the internal system level.

One approach to supporting temporal dimension is to utilize commercial graph databases by attaching temporal metadata. Cattuto et al.~\cite{Frame} propose a temporal model where the temporal data is represented as frames, with a frame being defined as the finest unit of temporal aggregation.
%frame这篇论文
Debrouvier et al. enhance the property graph model to store the temporal graph~\cite{T-GQL}. They assume that Objects, Attributes, and Values are stored as conventional property graph vertices, whereas time intervals are stored as properties of these vertices. Temporal edges are, in turn, stored as conventional edges, with time intervals as one of their properties.
Other works exploit the graph database by representing the temporal dimension as an interval attribute to graph objects~\cite{liu2017keyword,GRADOOP,Timebased3,Timebased4, durand2017backlogs}. These systems are implemented at the application level upon existing database systems. 
% , benefiting from existing database systems.
They may exhibit unpredictable performance, as the underlying storage and query engines are designed for static graphs rather than evolving ones. 
Another line of research focuses on underlying storage engines to organize temporal graphs. Their primary goal is to handle growing historical data while enabling efficient querying. In this regard, two storage approaches, \textit{Copy} and \textit{Log}, are widely used to manage temporal graph data~\cite{DeltaGraph,ChronoGraph, Raphtory,llama,ImmortalGraph,clock-g, graphone, Chronos, Auxo, khurana2016storing}.
% ~\cite{ChronoGraph, ImmortalGraph, Log, DeltaGraph,llama, clock-g, graphone}. 
The \textit{Copy} approach~\cite{DeltaGraph,ChronoGraph}
% , implemented by systems like DeltaGraph~\cite{DeltaGraph} and ChronoGraph~\cite{ChronoGraph}, 
stores a complete graph state whenever a batch of updates occurs. 
% periodically 
While this simplifies graph querying, it results in excessive redundancy in the stored graph information.
% , making graph querying easier but introducing excessive redundant graph information. 
In contrast, the \textit{Log} approach~\cite{Raphtory,llama}
% , adopted by systems like Raphtory~\cite{Raphtory} and LLAMA~\cite{llama}, 
records every graph update activity in a log. This approach is more compact but requires expensive reconstruction when executing a temporal graph query. As a trade-off between query performance and space overhead, the \textit{Copy+Log} approach~\cite{ImmortalGraph,clock-g, graphone, Chronos, Auxo, khurana2016storing}
% , followed by ImmortalGraph~\cite{ImmortalGraph} and Clock-G~\cite{clock-g}, 
combines a finite set of snapshots with a list of deltas between them. 
% is proposed as a hybrid of both approaches, combining a finite set of snapshots with a list of deltas between them. 
However, we argue that \textit{Copy+Log} is a sub-optimal approach since it still necessitates significant storage overhead to materialize the entire graph. Furthermore, all the above approaches primarily concentrate on offline historical data, limiting their usability for online queries. Worse still, to the best of our knowledge, they lack support for either a powerful temporal graph data model or a declarative temporal query language, restricting the convenience for users. 
}
\extended{% \vspace{-1mm}
\section{Conclusion}
\label{sec:conclusion}
In this paper, we present {\tgdb}, a new graph database that efficiently offers built-in temporal support. {\tgdb} is based on a novel temporal property graph model. To fit this
model, we design a temporal-enhanced storage engine and a query engine.
% Based on a proposed temporal property graph model, 
We first provide a hybrid storage engine with 
an asynchronous migration, 
guaranteeing small performance degradation on current storage and optimizing storage efficiency for historical storage. 
We then introduce a temporal query engine, which employs an {\visible} examination technique to ensure data consistency and an anchor-based version retrieval technique to boost query performance. We conducted extensive experimental evaluations using three workloads.
% The result demonstrated that {\tgdb} achieves up to \makesure{$4.5\times$} lower storage consumption and \makesure{$5.1\times$} lower query latency against state-of-the-art approaches, while only inducing {8.72\%} performance degradation for supporting temporal features.
The results demonstrated that {\tgdb} achieves up to {5.73$\times$} lower storage overhead, {{2.82$\times$} lower latency for graph operations, and {2.57$\times$} lower latency for temporal queries against the next-best system}, while introducing only {{9.74\%}} performance degradation for supporting temporal features.}

\balance 
\normalem
\bibliographystyle{ieeetr} %ACM-Reference-Format
\bibliography{bible}

\extended{
\setcounter{table}{0}
\newpage
\appendix

% \twocolumn[\centerline{\textbf{\Huge Temporal Queries of T-LDBC}}\vspace{2mm}]

\section{Temporal Queries of T-LDBC}
{We list the temporal queries of the T-LDBC workload used in our experiment in Table \ref{table:ldbc_temporal_q}. The values of the vertex id are sampled from the respective datasets on which the queries are evaluated.}

\begin{table}
 \centering
\setlength{\abovecaptionskip}{0.1cm}
\caption{Temporal queries of T-LDBC workload}
\resizebox{\columnwidth}{!}{
\begin{tabular}{|l|l|}
\hline
IS1 & {\color[HTML]{343434} \begin{tabular}[c]{@{}l@{}}MATCH (n:Person \{id:xx\})-{[}:IS\_LOCATED\_IN{]}-(p:Place)\\ \textbf{FOR TT AS OF $t$}\\ RETURN  \\ n.firstName AS firstName, n.lastName AS lastName,  \\ n.birthday AS birthday, n.locationIP AS locationIp,  \\ n.browserUsed AS browserUsed, n.gender AS gender,  \\ n.creationDate AS creationDate,  p.id AS cityId;\end{tabular}}                                                                                                                                                                                                                                                              \\ \hline
IS2 & \begin{tabular}[c]{@{}l@{}}MATCH (:Person \{id:xx\})\textless{}-{[}:HAS\_CREATOR{]}-(m)-{[}:REPLY\_OF*0..{]}-\textgreater{}(p:Post)\\ MATCH (p)-{[}:HAS\_CREATOR{]}-\textgreater{}(c)\\ \textbf{FOR TT AS OF $t$}\\ RETURN m.id as messageId,\\   CASE m.content is not null\\    WHEN true \\    THEN m.content,\\     ELSE m.imageFile, \\   END AS messageContent,\\   m.creationDate AS messageCreationDate,  p.id AS originalPostId,\\    c.id AS originalPostAuthorId, c.firstName as originalPostAuthorFirstName,\\    c.lastName as originalPostAuthorLastName\\ ORDER BY messageCreationDate DESC\\ LIMIT 10;\end{tabular}                   \\ \hline
IS3 & \begin{tabular}[c]{@{}l@{}}MATCH (n:Person \{id:xx\})-{[}r:KNOWS{]}-(friend)\\ \textbf{FOR TT AS OF $t$}\\ RETURN\\   friend.id AS personId, friend.firstName AS firstName,\\   friend.lastName AS lastName, r.creationDate AS friendshipCreationDate\\ ORDER BY friendshipCreationDate DESC, toInteger(personId) ASC;\end{tabular}                                                                                                                                                                                                                                                                                                                   \\ \hline
IS4 & \begin{tabular}[c]{@{}l@{}}MATCH (m:Message \{id:xx\})\\ \textbf{FOR TT AS OF $t$} \\ RETURN\\   CASE exists(m.content)\\     WHEN true THEN m.content\\     ELSE m.imageFile\\   END AS messageContent,\\   m.creationDate as messageCreationDate;\end{tabular}                                                                                                                                                                                                                                                                                                                                                                                       \\ \hline
IS5 & \begin{tabular}[c]{@{}l@{}}MATCH (m:Message \{id:xx\})-{[}:HAS\_CREATOR{]}-\textgreater{}(p:Person)\\ \textbf{FOR TT AS OF $t$} \\ RETURN p.id AS personId, p.firstName AS firstName, p.lastName AS lastName;\end{tabular}                                                                                                                                                                                                                                                                                                                                                                                                                             \\ \hline
IS6 & \begin{tabular}[c]{@{}l@{}}MATCH (m:Message \{id:xx\})-{[}:REPLY\_OF*0..{]}-\textgreater\\ (p:Post)\textless{}-{[}:CONTAINER\_OF{]}-(f:Forum)-{[}:HAS\_MODERATOR{]}-\textgreater{}(mod:Person)\\ \textbf{FOR TT AS OF $t$} \\ RETURN f.id AS forumId, f.title AS forumTitle,\\ mod.id AS moderatorId, mod.firstName AS moderatorFirstName, mod.lastName AS moderatorLastName;\end{tabular}                                                                                                                                                                                                                                                             \\ \hline
IS7 & \begin{tabular}[c]{@{}l@{}}MATCH (m:Message \{id:xx\})\textless{}-{[}:REPLY\_OF{]}-(c:Comment)-{[}:HAS\_CREATOR{]}-\textgreater{}(p:Person)\\ OPTIONAL MATCH (m)-{[}:HAS\_CREATOR{]}-\textgreater{}(a:Person)-{[}r:KNOWS{]}-(p)\\ \textbf{FOR TT AS OF $t$} \\ RETURN  c.id AS commentId, c.content AS commentContent, c.creationDate AS commentCreationDate,  \\ p.id AS replyAuthorId, p.firstName AS replyAuthorFirstName, p.lastName AS replyAuthorLastName,  \\ CASE r   \\ WHEN null THEN false    \\ ELSE true \\ END AS replyAuthorKnowsOriginalMessageAuthor\\ ORDER BY commentCreationDate DESC, toInteger(replyAuthorId) ASC;\end{tabular} \\ \hline
\end{tabular}
\label{table:ldbc_temporal_q}
}
\end{table}
}

\end{document}